\documentclass{CSML}

\def\dOi{12(2:10)2016}
\lmcsheading%
{\dOi}
{1--43}
{}
{}
{May\phantom..~23, 2015}
{Jun.~28, 2016}
{}

\ACMCCS{[{\bf Theory of computation}]: Semantics and
  reasoning---Program constructs---Type structures}
\keywords{recursive subtyping, higher-order types, duality, contract
  compliance, testing preorders}

\pdfoutput=1
\usepackage{tikz}
\usepackage{color}
\usepackage{stmaryrd}
\usepackage{proof}
\usepackage{xspace}
\usepackage{hyperref}
\usepackage{lscape}
\usepackage{listings}
\usepackage{txfonts}
\usepackage{tikz}
\usepackage{color}
\usetikzlibrary{calc,arrows,shapes,decorations.pathmorphing,backgrounds,positioning,fit}
\tikzstyle{state} = [rectangle,rounded corners,draw=black,thick,
                     minimum size=5mm, node distance=10mm and 12mm,
                     inner sep=3pt]
\usetikzlibrary{external}
\usepackage{model-macros}

\newtheorem{theo}{Theorem}[section]

\newtheorem{lemm}[theo]{Lemma}

\newcommand{\EndProofBox}{\null\hfill$\qed$}
\global\let\EndProof\EndProofBox
\newcommand{\boxHere}{\global\let\EndProof\empty\EndProofBox}

\newcommand{\EndDefBox}{\null\hfill$\qed$}
\global\let\EndDef\EndDefBox
\newcommand{\diaHere}{\global\let\EndDef\empty\EndDefBox}

\theoremstyle{plain}

\newcommand{\fllyabs}{fully-abstract\xspace}
\newcommand{\fabs}{full-abstraction\xspace}
\newcommand{\Fabs}{Full-abstraction\xspace}
\newcommand{\fp}{fixed point\xspace}
\newcommand{\fps}{fixed points\xspace}

\newcommand{\pfps}{pre\fps\xspace}

\title{Using higher-order contracts to model session types}

\author[G.~Bernardi]{Giovanni Bernardi\rsuper a}
\address{{\lsuper a}IMDEA Software Institute, Madrid, Spain}
\email{bernargi@tcd.ie}
\thanks{{\lsuper a}The first author was upported by SFI grant 06IN.11898, by
the Portuguese FCT project PTDC/EIA-CCO/122547/2010, and by the COST Action IC1201 (BETTY).}

\author[M.~Hennessy]{Matthew Hennessy\rsuper b}
\address{{\lsuper b}School of Statistics and Computer Science, O'Reilly Institute, Trinity College, Dublin 2, 
Ireland}
\email{matthew.hennessy@scss.tcd.ie}
\thanks{{\lsuper b}The second author was supported by SFI grant 06IN.11898 and by LERO - the Irish Software Research Centre, (13/RC/2094).}

\newcommand{\prfont}[1]{\ensuremath{ \mathbf {#1}}}

\definecolor{greenish}{rgb}{.24,.5,.26} 

\newcommand{\appto}[1]{[ \, #1 \, ]}
\newcommand{\defin}[1]{\mathtt{def} \,\, #1 \,\, \mathtt{in} \,\,}

\begin{document}
\maketitle

\begin{abstract}
Session types are used to describe and structure interactions between
independent processes in distributed systems.  Higher-order types are
needed in order to properly structure delegation of responsibility
between processes. 
In this paper we show that higher-order web-service contracts
can be used to provide a fully-abstract model of recursive
higher-order session types. The model is  set-theoretic,
in the sense that the meaning of  a contract is given in terms 
of the set of contracts with which it complies. 
A crucial step in the proof of full-abstraction is showing that
every contract has at least one other contract with which it complies.
\end{abstract}
\vfill

\newcommand{\pr}[1]{\texttt{#1}}

\section{Introduction}
The purpose of this paper is to show that session types
\cite{DBLP:conf/parle/TakeuchiHK94,DBLP:conf/esop/HondaVK98,DBLP:conf/wsfm/Dezani-Ciancaglinid09}
can be given a \fllyabs behavioural interpretation using
web-service contracts
\cite{DBLP:journals/tcs/Padovani10,DBLP:journals/toplas/CastagnaGP09}. Higher-order
session types are necessary to handle \emph{session delegation}, and
in turn this calls for the development of a novel form of \emph{peer
  compliance} between higher-order contracts.
Moreover, to prove the completeness of the model we introduce a novel
form of syntactic duality for higher-order contracts, which we call 
{\em peer-duality}.
We believe that this \emph{peer-duality}, when applied to 
session types,  captures the intuition of
  \emph{complementary behaviour} more faithfully than the standard
  notion of type duality from \cite{DBLP:conf/esop/HondaVK98}.

The current paper is the full-version of the extended abstract 
\cite{DBLP:conf/concur/BernardiH14}. It contains proofs for the
various results, more explanations, and more examples. 
It also corrects a mistake in 
\cite{DBLP:conf/concur/BernardiH14}, 
which occurred when defining 
our novel form of syntactic duality for higher-order contracts.
This will be explained in detail in \rsec{complement}.

\paragraph{{\bf Session types:}}
The interactions between processes in a complex distributed system
often follow a pre-ordained pattern. Session types 
\cite{DBLP:conf/parle/TakeuchiHK94,DBLP:conf/esop/HondaVK98} have been proposed as a mechanism for concisely describing and structuring these
interactions; they have been extensively studied
\cite{DBLP:conf/wsfm/Dezani-Ciancaglinid09} and are being put into
practice \cite{DBLP:conf/pvm/HondaMMNVY12}.
As a simple example consider a system consisting of two
entities
\begin{displaymath}
\nn{s} (\pin{\pch{urls}}{ s^+}{S} {\mathbf {store} } \Par 
    \pout{ \pch{urls}}{s^+} \prfont{cstmr}  )
\end{displaymath}
which first exchange a new private communication channel or
\emph{session}, $s$, over the public address of the store $\pch{urls}$;
using the conventions of \cite{DBLP:journals/acta/GayH05} the customer
sends to the store one endpoint of this private session, namely $s^+$, and keeps the other
endpoint $s^-$ for itself. The session type $S$ determines the nature
of the subsequent interaction allowed between the two entities, each
using its own exclusive endpoint of the private session $s$.

One form that the session type $S$ could take is
 \begin{equation}\label{eq:type1}
\begin{array}{l}
  \sinp{\gt{Id}} \branch{\gt{l}_1\as \sinp{\gt{Addr}} \sout{\int} T, \\
    \phantom{\sinp{\gt{Id}} \&\langle \,\,}\gt{l}_2 \as  \sinp{\gt{Addr}} \sout{\int} T, \\  
     \phantom{\sinp{\gt{Id}} \&\langle \,\,}   \gt{l}_3 \as  \sinp{\gt{Addr}} \sout{\int} T  }
    \end{array}
\end{equation}
where 
$\gt{\int},\gt{Addr},\gt{Id} $ are some base types of  integers,
addresses and credentials respectively,
and 
\begin{math}
        \branch{ \gt{l}_1\as S_1, \ldots, \gt{l}_n \as S_n }  
\end{math}
is a \emph{branch} type which accepts a choice between interaction on
any of the predefined labels $\gt{l_i}$, followed by the interaction
described by the residual type $S_i$. Thus (\ref{eq:type1}) above
dictates that $\prfont{store}$ offers a sequence of four interactions on
its end $s^+$ of the session, namely (i) reception of credential, 
(ii) acceptance of a choice among three commodities
labelled by $\gt{l}_1, \gt{l}_2, \gt{l}_3$, (iii) followed by the
receipt of an address, and (iv) the transmission of a
price, of type $\gt{Int}$; subsequent behaviour is determined by the
type $T$.

The behaviour of $\prfont{cstmr} $ on the other end of the session,
$s^-$, is required to match the behaviour described by $S$, thus
satisfying a session type which is intuitively dual to $S$. For
example,  the dual to (\ref{eq:type1}) above is
 \begin{equation}\label{eq:type2}
\begin{array}{l}
  \sout{\gt{Id}} \choice{\gt{l}_1 \as \sout{\gt{Addr}} \sinp{\int} T', \\
    \phantom{\sinp{\gt{Id}} \&\langle \,\,}\gt{l}_2 \as  \sout{\gt{Addr}} \sinp{\int} T', \\  
     \phantom{\sinp{\gt{Id}} \&\langle \,\,}\gt{l}_3 \as   \sout{\gt{Addr}} \sinp{\int} T' }
    \end{array}
\end{equation}
under the assumption that $T'$ is the dual of $T$.
Intuitively, input is dual to output and the dual to a branch type
is a {\em choice} type 
\begin{math}
   \choice{\gt{l}_1 \as S_1, \ldots , \, \gt{l}_k \as S_k}
\end{math}, 
which allows the process executing the role described by the type to
choose one among the labels $\gt{l}_i$. 
These two principles lead to a general definition of the \emph{dual} of
a session type $T$, denoted $\stdual{T}$ in \cite{DBLP:conf/parle/TakeuchiHK94,DBLP:conf/esop/HondaVK98}.
\label{dual.discussion}

In order to allow
flexibility to the processes fulfilling the roles described by these
types a subtyping relation  between session types, $T \subt S$,  
is essential; see
\cite{DBLP:journals/acta/GayH05} for a description of the crucial role played by subtyping. 
Intuitively $T \subt S$ means that any process or component fulfilling
the role dictated by the session type $S$ may be used where one is
required to fulfil the role dictated by $T$. Thus subtyping gives an
intuitive \emph{comparative semantics} to session types. However, to
the best of our knowledge, subtyping for session types has only ever
been given a purely syntactic definition; in \rdef{subtyping} of
\rsec{session.types} we slightly generalise the standard definition of
\cite{DBLP:journals/acta/GayH05}, so as to account also for base types
such as $\gt{\int}$ and $\bool$.


\paragraph{{\bf Recursive types:}}
Sessions over which interactions have reached completion are
described by the type $\e$.
Some sessions, though, may allow interactions between their endpoints
to go on indefinitely. This is required for instance by processes that
offer services or methods. To accommodate this, recursion has been
added in some process calculi, for instance in 
\cite{DBLP:conf/esop/HondaVK98,DBLP:journals/entcs/YoshidaV07,DBLP:conf/concur/DemangeonH11}.
To show instances of such processes below, and also in 
\rsec{discussion}, we use the syntax of
  \cite{DBLP:journals/entcs/YoshidaV07}. 
The only non-standard
  constructs in that syntax are the  two which describe  session
  delegation, that is {\tt throw} and {\tt catch}.
  The process $\throw{\kappa'}{\kappa} P$ writes the endpoint $\kappa$
  over the endpoint $\kappa'$, and continues as $P$. The process
  $\catch{\kappa'}{z} X \appto{ Q }$, on the other hand, inputs a name,
  say $\kappa$, over the endpoint $\kappa'$, and continues as 
  the process $Q \subst{ \kappa }{ z }$.

\begin{exa}
\name{An ever-lasting session}\label{ex:ever-going-session}\\
Here we use the syntax of \cite{DBLP:journals/entcs/YoshidaV07} to describe processes.
$$
\begin{array}{lll}
D_s & = & X(y) := \poffer{y}{ \gt{plus} \as
  \pinYo{y}{x}{\int}\pinYo{y}{z}{\int}\pout{y}{x + z} X\appto{y} \,\talloblong\, \\
& & \phantom{X(y) := y \rhd \{ }\gt{pos} \as
\pinYo{y}{z}{\real}\pout{y}{ z > 0 } X\appto{y}}\\[.5em]
D_c & = & Y(x) := \pchoice{  x  }{ \gt{pos}\as
  \pout{x}{\mathit{random()}}\pinYo{x}{z}{\bool}Y\appto{x}}\\
P & = & \nn{\kappa}(\defin{D_s} \defin{D_c} X\appto{\kappa^+} \Par Y\appto{\kappa^-})
\end{array}
$$
The peer $X\appto{ \kappa^+ }$ defined  by instantiating  $D_s$ accepts over $\kappa^+$ the invocation of one of
the two methods $\gt{plus}$ and $\gt{pos}$, reads the actual
parameters, sends the result of the chosen method and 
starts again. 
The peer $Y\appto{ \kappa^- }$ defined by  instantiating   $D_c$
invokes via its endpoint $\kappa^-$ the method $\gt{pos}$, sends a
random number, reads the result of the invoked method, and also starts
again. The composition of these  two peers in $P$
results in a never ending session in which interaction occurs
between the two peers forever.

Note that the definition of $D_s$ is a recursive version of the math
server of \cite{DBLP:journals/acta/GayH05}.
\end{exa}

The type $T$ that describes the behaviour of $D_s$ on 
an  endpoint $k^+$ is intuitively a  
solution of the equation
$$ 
T = \branch{ \gt{plus} \at \sinp{\int}\sinp{\int}\sout{\int}T, \,\, 
  \gt{pos}\at\sinp{\int}\sout{\int}T} 
$$
A natural way to express solutions to such equations  is to use
recursive types. For instance we could take the type $T$ required
above
to be 
$$
\Rec[X]{\branch{ \gt{plus} \at \sinp{\int}\sinp{\int}\sout{\int} X, \,\, 
  \gt{pos}\at\sinp{\int}\sout{\int} X} }
$$

The type equations that require recursive terms arise naturally in
presence of recursive processes. 
The presence of recursive types justifies the coinductive definition
of the subtyping (see \rdef{subtyping}), which follows
the approach of \cite{DBLP:journals/mscs/PierceS96}.
\leaveout{\MHf{
I don't think \cite{DBLP:journals/mscs/PierceS96} is the original
source for recursive types}}

\paragraph{{\bf Contracts:}}
Web services \cite{DBLP:journals/tcs/Padovani10,DBLP:journals/toplas/CastagnaGP09} are distributed components which may be
combined and extended to offer services to clients. These services are
advertised using \emph{contracts}, which are high-level descriptions of 
the expected behaviour of services.  These contracts come equipped with a \emph{sub-contract}
relation $\gt{cnt}_1 \sqsubseteq \gt{cnt}_2$; intuitively this means that the
contract $\gt{cnt}_2$, or rather a service offering the behaviour
described by this contract, may be used as a service which is required
to provide the contract $\gt{cnt}_1$. These abstract contracts are
reminiscent of process calculi as CCS and CSP
\cite{DBLP:books/daglib/0067019,csp}, and indeed the theory of these
process calculi have been used to give a behavioural interpretation of
contracts and the sub-contract preorder,
\cite{DBLP:journals/tcs/Padovani10,LP07}.

Contracts are very similar, at least
syntactically, to sessions types; for example (\ref{eq:type2}) above
can very easily be read as the following process description from CCS,
$$
!(\gt{Id}).( ?\gt{l}_1.?\gt{Addr}.!\int.\gt{cnt'} \extc 
 ?\gt{l}_2.?\gt{Addr}.!\int.\gt{cnt'} \extc
?\gt{l}_3.?\gt{Addr}.!\int.\gt{cnt'}
)
$$
In fact in Section~\ref{sec:contracts} we give a straightforward
translation $\encSym$ from the language of session contracts to that of
contracts.  Via the mapping $\encSym$ it should therefore be possible
to give a behavioural interpretation to session types, thereby
explaining how session types determine process behaviour, at least
along individual sessions. Indeed steps in this direction have already been made in
\cite{DBLP:conf/ppdp/Barbanerad10,DBLP:conf/sac/BernardiH12} restricting session types
to the first-order ones, that is types that cannot express session delegation.  
But, as we will now explain, the use of delegation in
session types requires the use of higher-order types, and in turn higher-order contracts, for which
suitable behavioural theories are lacking.

\paragraph{Session delegation:}
Consider the following system where the costumer $\prfont{cstmr}$ is replaced by 
$\prfont{girlf}$ and there are now four components: 
$$
\begin{array}{l@{\hskip -.5pt}l}
  \nn{s} \nn{p} \nn{b}(&   \pout{ \pch{urls}}{s^+} \pout{ \pch{urlb}}{p^+} \pout{ \pch{urlb}}{b^+} \prfont{girlf} \Par \\
& \pin{\pch{urls}}{s^+}{S}  \prfont{store}  \Par \\
&  \pin{\pch{urls}}{p^+}{S_p}  \prfont{bank}  \Par \\
&\pin{\pch{urlbf}}{b^+}{S_{b}}  \prfont{boyf}   
)
\end{array}
$$
Here we model a scenario in which a girlfriend opens a connection to a store,
lets her boyfriend choose a commodity, and then pays for the present.
Three private sessions $s,p, b$ are created and the positive endpoints are distributed to the 
\prfont{store}, \prfont{bank}, and \prfont{boyf} respectively. 
One possible script for the new customer $\prfont{girlf}$
is as follows:
\begin{enumerate}[label=(\roman*)] 
\item send credential to \prfont{store}: send id on  session $s^-$
\item \textbf{delegate} choice of commodity to \prfont{boyf}: send session $b^-$ on session $s^-$
\item await \textbf{delegation} from \prfont{boyf} to arrange payment:
  receive session $s^-$ back on session $b^-$.
\end{enumerate}
Thus the session type $S_b$ at which the boyfriend uses the session
end $b^+$ must countenance both the reception and transmission of
session ends, rather than simply data. In this case we can take $S_b$
to be the higher-order session type
\begin{equation}\label{eq:Sb}
\sinp{T_1} \sout{T_2} \e
\end{equation}
where in turn $T_1$ is the session type $\choice{\gt{l}_1 \as
  \sinp{\gt{Addr}}\e }$ and $T_2 $ must allow \prfont{girlf}  to arrange
payment through the \prfont{bank}. \leaveout{As we will see}  This in
turn means that $T_2$ is a higher-order session type as payment will
involve the transmission of the payment session $p$.

The combination of delegation and recursion leads to processes with
complicated behaviour which in turn puts further strain on the system
of session types. 
\begin{exa}\name{Everlasting generation of finite sessions}
\label{exa:intro-funny-types-needed}
\\Consider the process $P$ defined as follows,
$$
\begin{array}{lll}
D & := & X( x, y ) = \nn{\kF}(  \throw{x}{\kF^+}\void \Par
\catch{y}{z} X \appto{ z, \, \kF^- })\\[.5em]
P & = &  \nn{\kO}(\defin{D} X\appto{ \kO^+, \, \kO^- })
\end{array}
$$
Intuitively, at each iteration the code $X\appto{ \kO^+,\kO^- }$ has
the two endpoints  of a pre-existing session, $\kO$,
generates a new \emph{fresh} session, $\kF$, delegates over the
endpoint $\kO^+$ the endpoint $\kF^+$, and then recursively repeats
the loop using $\kF$ as pre-existing session.

\leaveout{\MHf{Is this true ?}}
According to the reduction semantics in 
\cite{DBLP:journals/entcs/YoshidaV07}  the execution of $P$ 
will never give rise to  a communication error or a  deadlock. 
But the endpoint $\kF^{+}$ can only be assigned a session type of the
form $\Rec[X]{\sout{X} \e}$. Such types are forbidden in 
\cite{Barbanera-deLiguoro-2013} but they are allowed in the typing
systems of \cite{DBLP:conf/esop/HondaVK98, DBLP:journals/acta/GayH05,
  DBLP:journals/entcs/YoshidaV07, DBLP:journals/iandc/Vasconcelos12};
we discuss the type inference in \rsec{discussion}.  
\end{exa}

If session types are to be explained behaviourally via the
translation $\encSym$ into contracts, the target language of contracts 
needs to be higher-order. For instance, the type of (\ref{eq:Sb})
above is mapped by $\encSym$ to the contract
$ ?( !\gt{l}_1.?(\gt{Addr}).\Unit ).?( \gt{cnt}_2  ).\Unit $, where $\gt{cnt}_2 = \M{T_2}$.
This in turn means that we require a
behavioural theory of higher-order contracts. This is the topic of the
current paper. In particular we develop a novel sub-contract preorder,
which we refer to as the \emph{peer} sub-contract preorder $\Peerleq$
with the property that, for all session types,
\begin{eqnarray}\label{eq:main}
  S \subt T \,\text{if and only if}\, {\M{S}} \Peerleq {\M{T}}
\end{eqnarray}
On the left hand-side we have the subtyping preorder between 
session types, which determines when processes with session type $T$
can play  the role required by type $S$; on the right-hand side we have a 
behaviourally determined sub-contract preorder between the
interpretation of the types as higher-order contracts. This
behavioural preorder is defined in terms of a novel definition of 
\emph{peer compliance} between these contracts.

In the remainder of this Introduction we briefly outline how this 
sub-contract preorder is defined; we will refer to it as
the \emph{peer sub-contract preorder}. 
Intuitively 
$\sigma_1 \Peerleq \sigma_2$, where $\sigma_i$ are contracts,
if every contract $\rho$ which \emph{complies} with $\sigma_1$ also
\emph{complies} with $\sigma_2$. In turn the intuition behind
\emph{compliance}
is as follows. We say that a contract $\rho$ complies with contract $\sigma$, 
written  $\rho \peercmpl  \sigma$, if any pair of processes in the source
language $p,q$ which guarantee
the contracts $\rho, \sigma$ respectively, can interact indefinitely 
to their mutual satisfaction; in particular if no further interaction
is possible between them, individually they both have reached
\emph{successful} or \emph{happy} states. We call this 
relation, which is symmetric, \emph{mutual} or \emph{peer} compliance, as both participants are
required to attain a \emph{happy} state simultaneously. This is in
contrast to 
\cite{DBLP:journals/toplas/CastagnaGP09,DBLP:journals/tcs/Padovani10,DBLP:conf/sac/BernardiH12} 
where an asymmetric compliance is used, in which only one participant,
the client,
 is required to reach a \emph{happy} state.

In this paper, rather than discussing processes in the source language, how they
can interact and how they guarantee contracts, we mimic the
interaction between processes using a symbolic semantics between
contracts.  We define judgements of the form
\begin{eqnarray}\label{eq:taus}
  \rho \Par \sigma  \ar{\tau}  \rho' \Par \sigma'
\end{eqnarray}
meaning that if $p,q$, from the source language, guarantee
the contracts $\rho, \sigma$ respectively, then they can interact and evolve to 
processes $p',q'$ which guarantee the residual contracts
$\rho', \sigma'$ respectively. 

For example we will have the judgement
\begin{eqnarray*}
  !\gt{\int}.\rho'  \,\Par\,   ?\gt{\real}.\sigma'   &\ar{\tau}&   \rho' \Par \sigma' 
\end{eqnarray*}
On the right-hand side of the parallel constructor $\Par$ 
we have a contract guaranteed by a process 
which will accept  a datum which can
be used as a real;
on the left-hand side there is a 
contract guaranteed by a process
that supplies an $\int$.
Since we are assuming that integers can be interpreted as
reals, that is $\int \bsubtype \real$, we know that an interaction
described by the judgement above takes place. 

However it is unclear when  an interaction of the form
\begin{eqnarray}\label{eq:hotaus}
  ! (\sigma_1).\rho' \,\Par\,    ?(\sigma_2).\sigma'  \ar{\tau}  \rho' \Par \sigma' 
\end{eqnarray}
should take place. Here on the right is a 
contract satisfied by a process which provides a session endpoint that satisfies the contract
$\sigma_2$; 
on the left is a contract satisfied by one which is willing to accept any session endpoint which guarantees 
the contract $\sigma_1$. 
Intuitively the interaction should be allowed if $\sigma_1$ is a
sub-contract of $\sigma_2$, that is ${\sigma_1} \Peerleq {\sigma_2}$. 
However the whole purpose of defining the judgements (\ref{eq:taus})
above is in order to define the preorder $\Peerleq$; there is a
circularity in our arguments. 

We break this circularity by supposing a predefined sub-contract
preorder $\B$ and allowing the interaction (\ref{eq:hotaus}) whenever 
$\sigma_1 \,\B\, \sigma_2$. More generally we develop a parametrised
theory, with interaction judgements of the form
\begin{eqnarray*}
 \rho \Par \sigma  \ar[\B]{\tau} \rho' \Par \sigma' 
\end{eqnarray*}
leading to a parametrised peer-compliance relation $\sigma \peercmpl[\B]
\rho$ which in turn leads to a parametrised sub-contract preorder
$\rho_1 \peerleq[\B] \rho_2$. We then prove the main result of the
paper, (\ref{eq:main}) above, by showing: 
\begin{quote}
  There exists a preorder $\B_0$ over higher-order contracts such
  that  $S \subt T$  if and only if $\M{S} \peerleq[\B_0] \M{T}$
 \end{quote}
This particular preorder $\B_0$ which we construct, and which 
has been referred to in (\ref{eq:main}) above as $\Peerleq$, 
has a natural behavioural interpretation. It satisfies the behavioural equation
\begin{eqnarray}\label{eq:fixpoint}
  \sigma_1 \B_0   \sigma_2 \;\text{if and only if}\; \sigma_1 \peerleq[\B_0] \sigma_2
\end{eqnarray}
Moreover it is the largest preorder between higher-order contracts
which satisfies (\ref{eq:fixpoint}).

The proof of (\ref{eq:fixpoint}) depends on an alternative syntactic
characterisation of the set-based preorders   
$\peerleq[\B]$.
The proof of this syntactic characterisation 
 relies in turn on a crucial  property
of the peer-compliance relation:
\begin{quote}
  for every contract $\sigma$ there exists a complementary contract
  $\rho$
 which complies with it,  $\sigma \peercmpl[\B] \rho$.
\hfill{($\star$)}
\end{quote}
\noindent
However constructing such complementary contracts is not
straightforward. One possibility is to use the notion of the
\emph{dual} of a session type,
\cite{DBLP:conf/parle/TakeuchiHK94,DBLP:conf/esop/HondaVK98},
already discussed informally on page~\pageref{dual.discussion}. A 
formal definition may be found in Figure~\ref{fig:std-dual},  and is readily
adapted  to contracts via the translation  $\encSym$;
specifically we can define $\stdual{\sigma}$ to be 
$\M{\stdual{\invM{\sigma}}}$.

However there are contracts $\sigma$ which do not
comply with their  duals,  $\sigma \Npeercmpl[\B] \stdual{\sigma}$.
Moreover these are not esoteric contracts but occur naturally 
when modelling reasonable processes. A typical example is 
the contract $\Rec[x]{?(x).\Unit}$, corresponding to the session type
$\Rec[X]{\sout{X} \e}$ needed to type the process  
in \rexa{intro-funny-types-needed}. 
In  \rexa{issues-stdual}  we explain why this contract does not comply
with its dual, assuming $\B$ satisfies some minimal requirements. 

The problem occurs with recursive contracts in which recursion
variables occur in message fields. Indeed  in 
\rthm{standard-duality-comply} we show that 
$\sigma \peercmpl[\B] \stdual{\sigma}$ whenever 
$\sigma$ has no such occurrences of recursion variables;
again subject to minor conditions on $\B$. 

However for arbitrary contracts we need a more general method of
constructing a complementary contract, which for example applies
to useful contracts such as   $\Rec[x]{?(x).\Unit}$.
In \cite{DBLP:conf/concur/BernardiH14}
we used a function, $\comp{-}$, to fulfil this role;
the same function was independently proposed in
\cite{DBLP:journals/corr/abs-1202-2086}, although with a different
purpose in mind. 
But unfortunately this has the same defect as duality: 
there are also contracts $\sigma$ such that
 $\sigma \Npeercmpl[\B] \comp{\sigma}$; see
 \rexa{comp-broken}.

Here  we propose 
another, more complicated, function 
$\prdual{-}$, see \rdef{novel-dual}; we refer to  $\prdual{\sigma}$ as
the \emph{peer-dual} of the contract $\sigma$. 
The intuitive idea is that $\prdual{\sigma}$ 
\begin{itemize}
\item first 
syntactically  transforms   $\sigma$ into another contract
$\sigma_{\text{ok}}$ 
which has no offending recursion variables, but is in some sense
still functionally equivalent to the original $\sigma$

\item
then returns the standard dual of $\sigma_{\text{ok}}$, 
namely
$\stdual{\sigma_{\text{ok}}}$. 
\end{itemize}
In 
 Theorem~\ref{thm:complements-comply} we
prove that 
 $\sigma \peercmpl[\B] \prdual{\sigma}$  for every session contract 
$\sigma$, thus establishing $(\star)$ above.


\paragraph{\bfseries Paper structure:}
The remainder of the paper is organised as follows. 
In \rsec{session.types} we recall the standard theory of
session types, while \rsec{contracts} is devoted to our
exposition of higher-order contracts and our novel notion of
parametrised peer sub-contract preorder $\peerleq[\B] $.
In \rsec{syntactic-peer-preorders} we show that, although the
definition of this preorder is set-theoretic, it can be characterised 
using only the syntactic form of contracts; this stems from  the 
very restricted form that our higher-order contracts can take. 
Using this syntactic characterisation we can develop enough 
properties of the preorders~$\peerleq[\B] $ to ensure the existence
of the particular preorder $\B_0$ alluded to in (\ref{eq:fixpoint})
above; this is the topic of \rsec{results}.\enlargethispage{\baselineskip}

As we have already stated, this result depends on 
being able to construct the complement of a contract. 
Our proposal, $\prdual{\sigma}$, is discussed separately 
in \rsec{complement}.
This section also contains explanations of the deficiencies
of the previous proposals $\stdual{\sigma}$, and $\comp{\sigma}$.
It then finishes, in 
\rsec{discussion},  with  a discussion of how  our proposed operator 
$\prdual{-}$ 
can be used to improve on the type-checking systems for 
session types in papers such as \cite{DBLP:journals/entcs/YoshidaV07,DBLP:journals/iandc/Vasconcelos12}.

Related work is then discussed in \rsec{conclusions}, and the
appendices contain some standard material and some minor  technical results.

\section{Session types}
\label{sec:session.types} 
    \begin{figure}
    \hrulefill
      \leaveout{
    $$
    \begin{array}{lllll}
      T, S & :: = & & \textbf{Types}\\
      & & U & \text{Session type}\\
      & & \ch{ U } & \text{Standard channel type}\\
      & & \Rec{ U } & \text{Recursive channel type}
      \\[1em]

      U, V & ::= & & \textbf{Session types}\\
      & & \e & \textit{Terminated session} \\
      & & \sinp{M}U & \textit{Input}\\
      & & \sout{M}U & \textit{Output}\\
      & & \branch{ \gt{l_1}\as U_1, \ldots, \gt{l_n}\as U_n } &
      \textit{Branch} ,\, n \geq 1 \\
      & & \choice{ \gt{l_1}\as U_1, \ldots, \gt{l_n}\as U_n } &
      \textit{Choice} ,\, n \geq 1\\
      & & \Rec{U} & \textit{Recursive session type} \\
      &  &  X & \textit{Type variable}
      \\[1em]

      M, N  & ::= & & \textbf{Message types}\\
      & & T & \textit{Type} \\
      & & \gt{t} & \textit{Base type}
      \\[1em]
      \gt{t}   & ::= & & \textbf{Base types}\\
      & & \gt{Id}, \gt{Addr}, \int,\real, \ldots & 
    \end{array}
    $$
}

    $$
    \begin{array}{lllll}
      T, S & ::= & & \textbf{Session types}\\
      & & \e & \textit{Terminated session} \\
      & & \sinp{M}T & \textit{Input}\\
      & & \sout{M}T & \textit{Output}\\
      & & \branch{ \gt{l}_1 \as T_1, \ldots, \gt{l}_n \as T_n } &
      \textit{Branch} ,\, n \geq 1 \\
      & & \choice{ \gt{l}_1 \as T_1, \ldots, \gt{l}_n \as T_n } &
      \textit{Choice} ,\, n \geq 1\\
      & & \Rec{T} & \textit{Recursive session type} \\
      &  &  X & \textit{Type variable}
      \\[1em]

      M, N  & ::= & & \textbf{Message types}\\
      & & T & \textit{Session Type} \\
      & & \gt{t} & \textit{Base type}
      \\[1em]
      \gt{t}   & ::= & & \textbf{Base types}\\
      & & \gt{Id}, \gt{Addr}, \int,\real, \ldots & 
    \end{array}
    $$

    \caption{Grammar of \leaveout{session} types \label{fig:st}}
\hrulefill
\end{figure}
Here we recall, using the notation from
\cite{DBLP:journals/acta/GayH05}, the standard theory of subtyping for
session types. 
The grammar for the language $L_{\st}$ of session 
type terms is given in 
\rfig{st}. That grammar uses a collection 
of unspecified base types $\gt{BT}$, of which we enumerate a sample.
It also uses a denumerable set of labels, 
$\labels = \sset{ \gt{l}_1, \gt{l}_2, \gt{l}_3, \ldots }$, in the
\emph{branch} and \emph{choice} constructs. Recall from the Introduction
that 
$\branch{ \gt{l}_1 \as {S}_1, \ldots, \gt{l}_n \as S_n } $ 
offers different possible behaviours based on a set of 
labels $\sset{ \gt{l}_1, \gt{l}_2, \gt{l}_3, \ldots \gt{l}_n }$
while $\choice{ \gt{l}_1 \as {S}_1, \ldots, \gt{l}_n \as {S}_n }$
takes a choice of behaviours; in both constructs the labels used
are assumed to be distinct.

 We use ${\st}$ to denote the set of session type terms in $L_{\st}$
 which are \emph{closed} and \emph{guarded}, in other words the terms
 that contain only variables which are not free, and that appear at 
 least after one non-recursive type constructor. 
 The definitions of both concepts are
 standard, and they may be found in \rapp{defs}.

Subtyping is defined coinductively and uses some unspecified
subtyping preorder  $\bsubtype$ between  base types, a typical
example being $\int \bsubtype \real$, meaning that an integer may be
supplied where a real number is required. 
Intuitively subtyping between session types is determined by
two principles:
\begin{enumerate}[label=\({\alph*}]
\item \emph{Branch}: 
  a branch type $T_1$ can be replaced 
  with a branch type $T_2$ that allows {\em more} choices than~$T_1$.

\item \emph{Choice}:
  a choice type $T_1$ can be replaced 
  with a choice type $T_2$ that allows {\em fewer} choices than~$T_1$.
\end{enumerate}
\noindent
We give two examples of these principles.
\begin{exa}
  Let
  $ \textsc{Bartender} = \branch{ \gt{espresso} \at T_1 } $ and
  \begin{align*}
  \textsc{FancyBartender} = & \, \branch{ \gt{espresso} \at T'_1, \\
    &\phantom{\ \&\langle} \gt{deka}\at T'_2, \\
    &\phantom{\ \&\langle} \gt{double-espresso}\at T'_3 }
  \end{align*}
  The type $\textsc{Bartender}$ accepts the choice of only one
  option, \gt{espresso}, so all the customers satisfied by $\textsc{Bartender}$, are satisfied by any other type that offers
  {\em at least} the label~$\gt{espresso}$. It follows that 
  $\textsc{FancyBartender}$ satisfy all the customers satisfied by
  the~$\textsc{Bartender}$.
  This is formalised by the subtyping, which relates the two types
  $$
  \textsc{Bartender} \subt \textsc{FancyBartender}
  $$
  as long as also the continuations~$T_1$ and~$T'_1$ are related (ie.~$ T_1 \subt T'_1$).
\end{exa}
\begin{exa}
\label{ex:choice}
Let \textsc{ItalianCustomer} describe the different coffees that a
process may want to order when interacting with a bar tender.
\begin{align*}
\textsc{ItalianCustomer} = & \choice{\gt{espresso} \at T_1, \\
    &\phantom{\ \ \oplus\langle} \gt{deka}\at T_2, \\
    &\phantom{\ \ \oplus\langle} \gt{double-espresso}\at T_3
}
\end{align*}
All the bar tenders that are able to accept this range of choices,
have to offer {\em at least} the three labels that appear in \textsc{ItalianCustomer}.
Now consider the type
$$
\textsc{Customer} = \choice{ \gt{espresso}\at T'_1}
$$
Since \textsc{Customer} chooses among fewer options than
\textsc{ItalianCustomer}, a process that behaves according to
\textsc{Customer} can be used in place of one that behaves as
prescribed by \textsc{ItalianCustomer}.
This is formalised by the subtyping relation as follows,
$$
\textsc{ItalianCustomer} \subt \textsc{Customer}
$$
if also the continuations are related (ie.~$ T_1 \subt T'_1$).
\end{exa}

Added to the principles sketched above are the  standard 
covariant and contravariant requirements on the input and output 
constructs. Recursive types are handled by a standard function 
$\unfold{T}$ which unfolds all the  first-level  occurrences of  
$\Rec{-}$ in the (guarded) type $T$.
The formal definition of $\unfoldSym$ in turn depends on the
definition of substitution $T \subst{S}{X}$, the syntactic
substitution of the term $S$ for all free occurrences of $X$ in $T$.
The details may be found in \rapp{defs}.

\begin{defi}\name{Subtyping}\label{def:subtyping}\\
Let $\Fun{\subt}: \parts{{\st}^2} \longrightarrow  \parts{{\st}^2}$ be
  the functional defined so that $(T, \, S) \in \Fun{\subt}(\R)$ whenever:
  \begin{enumerate}[label=(\roman*)]
\leaveout{
  \item\GBf{Now (i) account for channel types.}
    \GBc{
      if $\unfold{T} = \ch{ T' }$ then $\unfold{S} = \ch{ S' }$, $T' \R S'$ and $S' \R T'$
    }
}
    
  \item 
    if $\unfold{T}=\e$ then $\unfold{S}=\e$

  \item \label{pt:fo-inp}
    if $\unfold{T}=\sinp{\gt{t}_1}S_1$ then
    $\unfold{S}=\sinp{\gt{t}_2}S_2$ and $ S_1 \R S_2$ and
      $\gt{t}_1 \bsubtype \gt{t}_2$

 \item \label{pt:fo-out}
    if $\unfold{T}=\sout{\gt{t}_1}S_1$ then
    $\unfold{S}=\sout{\gt{t}_2}S_2$ and $S_1 \R S_2$ and
     $\gt{t}_2 \bsubtype \gt{t}_1$

  \item  \label{pt:ho-out}
    if $\unfold{T}=\sout{T_1}S_1$ then
    $\unfold{S}=\sout{T_2}S_2$ and $S_1 \R S_2$ and
     $T_2 \R T_1$

  \item \label{pt:ho-inp}
    if $\unfold{T}=\sinp{T_1}S_1$ then
    $\unfold{S}=\sinp{T_2}S_2$ and $S_1 \R S_2 $ and
      $T_1 \R T_2$

  \item \label{pt:branch}
    if $\unfold{T} = \branch{\gt{l}_1\colon T_1,\ldots, l_m\colon T_m}$ then
    $\unfold{S} = \branch{\gt{l}_1\colon S_1,\ldots,\gt{l}_n\colon S_n}$ 
    where  $m\leq n$ and $T_i \R S_i$ for all  $i\in [1,\ldots,m]$

  \item \label{pt:choice}
    if $\unfold{T}=\choice{\gt{l}_1\colon  T_1,\ldots, l_m\colon  T_m}$ then
    $\unfold{S}=\choice{\gt{l}_1\colon  S_1,\ldots,\gt{l}_n\colon  S_n}$ 
    where  $n\leq m$ and $T_i \R  S_i$ for all $i\in [1,\ldots,n]$

  \end{enumerate}
  If~$ {\R} \subseteq \Fun{\subt}( \R )$, then we say that~$R$
  is a {\em type simulation}.
The monotonicity of $\Fun{\subt}$ and the Knaster-Tarski Theorem \cite{tarski1955} 
  ensure that there exists the greatest
  solution of the equation~$ {\R}  = \Fun{\subt}( \R ) $; we call this 
  solution the {\em subtyping}, and we denote it~$\subt$.
\end{defi}
\noindent
The cases (\ref{pt:fo-inp}) and (\ref{pt:fo-out}) are not present in the original definition of \cite{DBLP:journals/acta/GayH05}, so our subtyping relation is in fact slightly more general than the one of that paper.
Note the use of covariance in the first argument in  (\ref{pt:fo-inp}),  (\ref{pt:ho-inp}) and the 
contravariance in  (\ref{pt:fo-out}) and  (\ref{pt:ho-out}); the intuitions of (a) and (b) above are reflected in (\ref{pt:branch}) and (\ref{pt:choice}) respectively.
 
In the next example we show how to prove that two types are related by the subtyping.
\begin{exa}
\label{ex:subt}
Let $S = \Rec[X]{\sinp{X}X}$ and $T = \Rec[Y]{\sinp{Y}\sinp{Y}Y}$.
In this example we prove that 
$$ \Rec[X]{\sinp{X}X} \subt \Rec[Y]{\sinp{Y}\sinp{Y}Y}$$

Thanks to the Knaster-Tarski Theorem, we need only 
exhibit a prefixed point of the functional~$\Fun{\subt}$, that
contains the pair $( S, \, T)$.
Consider the following relation
$$
{\R} = \sset{ (S, T), \, (\sinp{S}S, \sinp{T}\sinp{T}T ), \,
  (\sinp{S}S,  \sinp{T}T ), \, (S,  \sinp{T}T) }
$$
and let $A =  \Fun{\subt}( \R )$. The set inclusion ${\R} \subseteq
A$ follows from case (\ref{pt:ho-inp}) of Definition~(\ref{def:subtyping}).
\end{exa}

\newcommand{\client}[1]{{\sf clientbody}(#1)}
\newcommand{\server}[2][1]{{\sf serverbody}_{#1}(#2)}

Now that we have defined the subtyping relation $\subt$, and
shown a proof method for it, we discuss its meaning.
In the context of $\pi$-calculus equipped with sessions, 
the relation $\subt$ formalises two notions of safe substitutivity.
Suppose that $\int \bsubtype \real$. If $S \subt T$, then both of the
following are true:
\begin{enumerate}
\item 
a session end-point $\kappa_2$ at type $S$ may safely be used in place of 
a session endpoint $\kappa_1$ at type $T$. For example - using the standard $\pi$ syntax -
a process
$$
P = \pin{c}{x}{\sinp{\real}\e } \pin{x}{y}{\real} \pout{b}{y / 2 } \void
$$
may safely receive along the channel $c$ an endpoint $\kappa_2$ that has type $\sinp{\int}\e$,
instead of the declared type $\sinp{\real}\e$.
Intuitively, this is the case because $ \sinp{\int}\e \subt \sinp{\real}\e$,
and thus if $P$ can read over $x$ a value of type $\real$,
then it can read also a value of type $\int$.

\item \label{safe-subst-processes}
a process $P$ that uses an endpoint $\kappa$ at type $T$ may be safely
used in place of a process $Q$ that uses $\kappa$ at type $S$.
To sketch this phenomenon, we follow \cite[Section 2]{DBLP:journals/acta/GayH05} and
discuss the following processes,

$$
\begin{array}{lll}
  \client{ x } & = & \pchoice{ x }{ \gt{plus}\as \pout{x}{2} \pout{x}{3} \pin{x}{u}{\int} \void}

\\[.5em]

  \server{ x } & = & \poffer{ x }{ \gt{plus} \as \pin{x}{z}{\int}  \pin{x}{y}{\int} \pout{x}{y + z} \void }

\\[.5em]

  \server[2]{ x } & = & 
  \poffer{ x }{ \gt{plus} \as \pin{x}{z}{\int}  \pin{x}{y}{\int} \pout{x}{y + z} \void \\
  & &\phantom{ x \kern-1pt\rhd \{\,}\gt{mult} \as  \pin{x}{z}{\int}  \pin{x}{y}{\int} \pout{x}{y * z} \void}
\end{array}
$$
The session type at which $\server{ x }$ employs $x$ is
$S = \branch{ \gt{plus} \as S'}$,
where 
$$S' = \sinp{\int}\sinp{\int}\sout{\int}\e,$$
while the type at which $\server[2]{x}$ uses $x$ is $T = \branch{ \gt{plus} \as S', \gt{mult} \as S' }$,
and it is routine work to check that $S \subt T$. 
Now observe that - at least intuitively - 
the process 
$$
Q_1 = \nn{\kappa} (\client{ \kappa^- } \Par \server{ \kappa^+ })
$$ is well-typed. 
Thanks to $S \subt T$ and subsumption, one can adapt the type derivation for $Q_1$ to prove
that also the following process is well-typed,
$$ 
Q_2 = \nn{\kappa} (\client{ \kappa^- } \Par \server[2]{ \kappa^+ }).
$$ 
Essentially, by knowing that $Q_1$ is well-typed and that $S \subt T$,
one can show that the process $\server{ x }$ can be safely replaced by 
the process $\server[2]{ x }$,
in the sense that also $ Q_2 $ is well-typed.
\end{enumerate}

\noindent The argument in (\ref{safe-subst-processes}) above let us argue that 
  $\server[2]{x}$ may be used where $\server{ x }$ is required.
  Since the behaviours of these processes are described respectively
  by the types $S$ and $T$, we can reason directly on types, and
  argue that $S \subt T$ means that processes adhering to the role 
  dictated  by $T$  may be used where processes following the role
  dictated by $S$ are required.
  Out aim is to  formalise this intuition, 
  by proving that the higher-order contracts determined by these 
  types, respectively $\M{S}$ and $\M{T}$,
  are related behaviourally using our notion of \emph{peer} 
sub-contract preorder.

\leaveout{
Intuitively $S \subt T$ means that processes adhering to the role 
dictated  by $T$  may be used where processes following the role
dictated by $S$ are required.}

\leaveout{
Our aim is to formalise \GBc{the intuition sketched in (\ref{safe-subst-processes}) above,
so that if $S \subt T$, by proving that the higher-order contracts determined by these 
types, respectively $\M{S}$ and $\M{T}$, are related behaviourally
using our notion of \emph{peer} compliance, \GBc{then we are guaranteed
that $\M{T}$ satisfies all the peers satisfied by $\M{S}$.}}
}

\section{Higher-order  contracts}
\label{sec:contracts}

\begin{figure}[t]
 \hrulefill
 $$
 \begin{array}{lllll}
   \rho, \sigma & ::= & & \textbf{Higher-order session contracts}\\
   & & \Unit & \textit{Satisfied contract} \\
   & & ?\gt{t}.\sigma & \textit{Input}\\
   & & !\gt{t}.\sigma & \textit{Output}\\
   & & !(\sigma).\sigma & \textit{Contract output},\\
   & & ?(\sigma).\sigma & \textit{Contract input}, \\
   & & \sum_{i \in I}?\gt{l}_i.\sigma_i &  \textit{External sum}  \, I\text{ non-empty, finite}\\
   & & \bigoplus_{i \in I}!\gt{l}_i.\sigma_i & \textit{Internal sum} \, I\text{ non-empty, finite}\\
   & & \Rec[x]{\sigma} & \textit{Recursive session contract} \\
   &  &x & \textit{Session contract variable}
 \end{array}
 $$
\caption{Grammar of higher-order session contracts \label{fig:sc}}
\hrulefill
\end{figure}

In the first section we define higher-order session contracts and explain the set-based subcontract preorder on them; this uses the notion of 
\emph{peer compliance} between them. In the following section we show
that
this set-based preorder can be characterised by comparing the purely
syntactic structure of contracts, up-to a parameter $\B$.

\subsection{Contracts and compliance }
The grammar for the language of contract terms $L_\sc$ is given in
\rfig{st}; there we assume the labels $\gt{l}_i$s to be pairwise distinct. 
We use $\sc$ to denote the set of terms which
are guarded and closed. These will be referred to as higher-order
session contracts, or simply contracts.

\begin{figure}
\hrulefill
  \[
  \begin{array}{l@{\hskip 4em}l}
    \infer[\rname{a-Ok}] { \Unit \ar{\ok}}
                {\; }
     &
    \infer[ \lambda \in \Act \,\,  \rname{a-Pre}]
              {\lambda.\sigma\ar{\lambda} \sigma}
              {\;\;}
    \\[2em]
 \infer[| \, I \, | > 1 \; \rname{a-Int}]{  \bigoplus_{i \in I}!\gt{l_i}.\sigma_i
   \ar{\tau}  !\gt{l_i}.\sigma_i   }{}
     &
    \infer[\rname{a-Ext}]
    {\sum_{i \in I}?\gt{l_i}.\sigma_i    \ar{?\gt{l_i}}  \sigma_i}
    {  } 
\\[2em]
\infer[\rname{a-Unfold}]
{\Rec[x]{\sigma} \ar{\tau} \sigma \subst{\Rec[x]{\sigma}}{x}}
{}
\\[1em]
\end{array}
\]
\caption{The operational semantics of session contracts}
\label{fig:opsem}
\hrulefill
\end{figure}

The operational meaning of contracts is given by 
viewing them as processes from a simple process calculus and 
interpreting them as states in a (higher-order) labelled transition
system.
To this end let 
\begin{align*}
\Act = 
  \setof{?\gt{l},\, !\gt{l}}{\gt{l} \in \labels}  
\; \cup \; 
 \setof{?\gt{t},\, !\gt{t}}{t \in \gt{BT}}
\; \cup \;
\setof{?(\sigma),\, !(\sigma), }{\sigma \in \sc}
\end{align*}
be the set of prefixes, ranged over by $\lambda$.
We use \Acttt to denote the set $\Act \cup \sset{\tau,\ok}$ to
emphasise that the special symbols  $\tau$ and $\ok$ are not in
$\Act$.  In \rfig{opsem} we give a set of axioms
which define 
judgements of the form 
\begin{align*}
  \sigma_1 \ar{\mu}\sigma_2
\end{align*}
where $\mu \in \Acttt$ and $\sigma_1,\sigma_2 \in \sc$.
Note that  that terms like $!\gt{l}.\sigma$ are singleton
  internal sums and 
we infer their semantics by using the rule for prefixes, \rname{a-Pre};
for example $!\gt{l}.\sigma \ar{!\gt{l}} \sigma$.
The rule \rname{a-Unfold} uses a form of substitution of a closed term $\rho$ for free 
occurrences of a variable $x$ in a term $\sigma$, denoted $\sigma\subst{\rho}{x}$. 
In Appendix~\ref{app:defs} we give a more general definition of the application of 
a substitution $\subsS$ to a term $\sigma$, denoted $\sigma \, \subsS$; then 
$\sigma\subst{\rho}{x}$ corresponds to $\sigma \subsS_x$, where $\subsS_x$ is a simple
substitution which maps the variable $x$ to the closed term $\rho$. 
We say that a contract $\sigma$ is {\em stable} whenever~$\sigma \nar{\tau}$.

In order to define \emph{compliance} between two contracts
$\rho,\, \sigma$,  we also need to 
say when two processes $p,\, q$ satisfying these contracts can
interact. This is formalised indirectly as a 
relation of the form 
\begin{displaymath}
  \rho \Par  \sigma \ar[\B]{\tau} \rho' \Par \sigma'
\end{displaymath}
which, as explained in the Introduction, is designed to capture the 
informal intuition that if processes $p,q$ satisfy the contracts 
$ \rho, \sigma$ respectively, then they can interact and their residuals
will satisfy the residual contracts $\rho',\sigma'$ respectively. 
This reduction relation is parametrised on a  relation $\sigma_1 \B \sigma_2$ between contracts, which 
determines when the contract $\sigma_2$ can be accepted  when  $\sigma_1$ is
required. Using such a $\B$ we define an \emph{interaction}
relation between contracts as follows:
$$
\lambda_1 \bowtie_{\B} \lambda_2 =
\begin{cases}
  \lambda_1 = {!\gt{l}}, \, \lambda_2 = {?\gt{l}} \\
  \lambda_1 = {?\gt{l}}, \, \lambda_2 = {!\gt{l}} \\
  \lambda_1 = {!\gt{t}_1}, \, \lambda_2 = {?\gt{t}_2} & \gt{t}_1 \bsubtype \gt{t}_2\\
  \lambda_1 = {?\gt{t}_1}, \, \lambda_2 = {!\gt{t}_2} & \gt{t}_2 \bsubtype \gt{t}_1\\
  \lambda_1 = {!( \sigma_1 )}, \, \lambda_2 = {?( \sigma_2 )} & \sigma_1  \B  \sigma_2\\
  \lambda_1 = {?( \sigma_1 )}, \, \lambda_2 = {!( \sigma_2 )} & \sigma_2 \B \sigma_1 \\
\end{cases}
$$
\noindent
Essentially the relation $\bowtie_{\B}$ treats $\B$ as a subtyping on
contracts; note that by definition $\bowtie_{\B}$ is symmetric, for any $\B$.

The inference rules in \rfig{interaction} are now
straightforward; $\rho \Par \sigma$ can proceed if either of the
components
$\sigma,\, \rho$ can proceed independently, or if the components can
interact, as dictated by $\bowtie_{\B}$.

\begin{figure}[t]
\hrulefill
\[
  \begin{array}{l@{\hskip 2em}l}
    \infer[\rname{I-Left}]
    {  \rho \Par  \sigma \ar[\B]{\tau} \rho' \Par \sigma}
    {  \rho \ar{\tau} \rho' }
    &
    \infer[\rname{I-Right}]
    { \rho \Par \sigma \ar[\B]{\tau}  \rho \Par \sigma'}
    { \sigma \ar{\tau} \sigma'}
    \\[1em]
\multicolumn{2}{c}{
      \infer[\lambda_1 \bowtie_{\B} \lambda_2\; \, \rname{I-Synch}]
      { \rho \Par \sigma \ar[\B]{\tau} \rho' \Par \sigma'}   
      { \rho \ar{\lambda_1} \rho'\quad \sigma \ar{\lambda_2}  \sigma'}
    }
\end{array}
\]
\caption{Interacting session contracts}
\label{fig:interaction}
\hrulefill
\end{figure}

We are now ready to define our version of \emph{(peer) compliance}.
\begin{defi}
\label{def:peercmpl}
\leaveout{
$\rho \peercmpl[\B] \sigma$ whenever $ \rho \Par \sigma
  \wt[\B]{} \rho' \Par \sigma' $ and $\rho' \Par \sigma' \stable[\B]$
  imply $\rho' \ar[\B]{ \ok }{\B}$ and $\sigma' \ar[\B]{ \ok }$.
}
Let $\FunCpeer : \parts{\sc^2} \times \parts{\sc^2}
\longrightarrow  \parts{\sc^2}$ be the rule functional defined so that
$(\rho, \sigma) \in \FunCpeer(\R, \B)$ whenever both the
following conditions hold:

\begin{enumerate}[label=(\roman*)]\itemsep 0pt

\item\label{pt:peercmpl-stuck-ok} 
if $ \rho  \Par \sigma \stable[\B]$ then $ \rho \ar{\ok}$ and $\sigma
\ar{\ok}$

\item\label{pt:peercmpl-derivatives}
if $ \rho \Par \sigma \ar[\B]{\tau}  \rho '\Par\sigma'$ then
  $\rho ' \R \sigma' $

 \end{enumerate}
 If~$ {\R} \subseteq \FunCpeer( \R, \B )$, then we say that~$\R$
 is a {\em $\B$-coinductive peer compliance}.
 Fix a $\B$. Standard arguments ensure that there exists the greatest
 solution of the equation~$ X = \FunCpeer( X, \B ) $; we call this 
 solution the {\em $\B$-peer compliance}, and we denote
 it~$\peercmpl[\B]$.
\end{defi}
\noindent
The intuition here is that if $\rho \peercmpl[\B] \sigma$ then pairs of
processes satisfying these contracts can interact indefinitely, until 
such time that they can both simultaneously perform the success action
$\ok$. 
So the interaction between them can continue indefinitely, even
forever; but if further interaction is not possible then condition
(\ref{pt:peercmpl-stuck-ok}) ensures that 
both participants 
must have reached 
the \emph{happy} state. 
\begin{exa}
We show the impact of the symmetric requirement in \rdefpt{peercmpl}{peercmpl-stuck-ok}.
Consider the contracts $\rho = {!\int.\Unit}$, and $ \sigma =  \Rec[y]{?\int. y}$.
While $\rho$ requires just one interaction to reach the satisfied state $\Unit$,
the contract $\sigma$  supports  an infinite number of interactions, but never
reaches a satisfied state. Under the reasonable assumption that 
$\int \bsubtype \int$, after one interaction the composition
$\rho \Par \sigma$ reduces to a stable state, in which the derivative of $\sigma$ 
is not satisfied; therefore $\rho \Npeercmpl[\B] \sigma$.
\end{exa}

In general the properties of the compliance relation 
$\peercmpl[\B] $ depends on those of the underlying relation $\B$, and
of the first-order subtyping relation $\bsubtype$, for example
$!( \Cnil ).\Unit \Npeercmpl[\emptyset] ?( \Unit ).\Unit $, while
$!( \Cnil ).\Unit \peercmpl[\B] ?( \Unit ).\Unit $ for ${\B} = \sset{ (\Cnil, \, \Unit) }$.
The ability of a contract to comply with another depends not on its
syntax but rather on its behaviour, which is determined by its operational
semantics as given by the rules in \rfig{opsem}.
To emphasise this point let us adapt the standard notion
of \emph{bisimulation equivalence} \cite{ccs} to our setting.

\leaveout{
In general the properties of the compliance relation 
$\peercmpl[\B] $ depends on those of the underlying relation $\B$, and
of the first-order subtyping relation $\bsubtype$. 
For example if both of
these are symmetric then one can show that $\peercmpl[\B] $ is also
symmetric. 
But in general the ability of a contract to comply with
another depends not on its syntax but rather on its behaviour as
determined by its operational semantics as given by the rules in
Figure~\ref{fig:opsem}. To emphasise this point let us adapt the standard notion
of \emph{bisimulation equivalence} \cite{ccs} to our setting.
}

\begin{defi}\name{ Strong bisimulation }\label{def:sb}\\
A relation ${\R} \subseteq \sc \times \sc $ is called a (strong) 
bisimulation if whenever ${\sigma_1} \R  {\sigma_2}$ then 
\begin{enumerate}
\item \label{pt:sb-ok}
$\sigma_1 \ar{\ok} $ if and only if
$\sigma_2 \ar{\ok} $, and
\item for every $\mu \in \Actt$
\begin{enumerate}[label=\({\alph*}]
\item \label{pt:sb-lr}
$\sigma_1 \ar{\mu} \sigma'_1$ implies 
$\sigma_2 \ar{\mu} \sigma'_2$ for some $\sigma'_2$, such 
that ${ \sigma'_1 } \R { \sigma'_2 }$

\item \label{pt:sb-rl}
conversely,
$\sigma_2 \ar{\mu} \sigma'_2$ implies 
$\sigma_1 \ar{\mu} \sigma'_1$ for some $\sigma'_1,$ such 
that  ${ \sigma'_1 } \R {\sigma'_2}$,
\end{enumerate}
\end{enumerate}
We write $\sigma_1 \bisim{} \sigma_2$ whenever there is some 
bisimulation $\R$ such that $\sigma_1 \R \sigma_2$. 
\end{defi}
\noindent
Our main interest in this strong form of bisimulation is encapsulated in 
\rprop{bisim}, which in turn uses the next lemma.

\begin{lem}
\label{lem:bisim-preserves-stuck-states}
  Suppose $\rho_1 \bisim{} \rho_2$.
  Then for every $\sigma \in \sc$ and every $\B$, 
  $\rho_1 \Par \sigma \nar[\B]{\tau}$ implies 
  $\rho_2 \Par \sigma \nar[\B]{\tau}$.
\end{lem}
\begin{proof}
  The hypothesis $\rho_1 \Par \sigma \nar[\B]{\tau}$ imply that
  both $\rho_1$ and $\sigma$ are stable.
  \rdef{sb} ensures that
  $\rho_2$ is also stable.
  To prove that $\rho_2 \Par \sigma \nar[\B]{\tau}$ it suffices to show that
  (a) if $\rho_2 \ar{ \lambda }$ then
  $\sigma \ar{ \lambda' }$ implies $\lambda \not \bowtie_{\B} \lambda'$.
  Suppose that $\rho_2 \ar{ \lambda }$ for some $\lambda$, \rpt{sb-lr} of 
  \rdef{sb} and
  the hypothesis that $\rho_1 \bisim{} \rho_2$ guarantee that $\rho_1 \ar{ \lambda }$.
  The hypothesis that $\rho_1 \Par \sigma \nar[\B]{\tau}$ and rule \rname{I-Synch}
  in \rfig{interaction} ensure that $\sigma \ar{ \lambda' }$ implies 
  $\lambda \not \bowtie_{\B} \lambda'$.
  In view of the assumption on $\lambda$, we have proven (a).
\end{proof}

In the proof of the next proposition we denote the
transitive closure of a relation $\R$ with $\R^{+}$, and the
reflexive and transitive closure with $\R^\star$.
We will use this notation throughout the paper.

\begin{prop}
\label{prop:bisim}
  Suppose $\rho_1 \bisim{} \rho_2$.
  Then for every $\sigma \in \sc$ and every $\B$, 
  $\rho_2 \peercmpl[\B] \sigma$ implies 
  $\rho_1 \peercmpl[\B] \sigma$.
\end{prop}
\begin{proof}
\leaveout{
  It is sufficient to define a relation $\R$ such that ${\bisim{}}  \subseteq {\R}$, and
  $ {\R} \subseteq \FunCpeer( \R, \B )$. }

We have to prove the inclusion ${\bisim{} \cdot \peercmpl[\B]}  \subseteq  {\peercmpl[\B]}$,
and it suffices to show a relation $\R$ such that~${\bisim{} \cdot \peercmpl[\B]} \subseteq {\R} $ and that $ {\R} \subseteq \FunCpeer( \R, \B )$.
Let
$$
{\R}  = \setof{ (\rho', \sigma') }{ 
\text{ for some } \rho_1, \rho_2, \sigma \in \sc \text{ such that }
\rho_1 \arstar{\tau} \rho', \, \sigma \arstar{\tau} \sigma', \,
\rho_2 \peercmpl[\B] \sigma, \, \rho_1 \bisim{} \rho_2 } 
$$

The relation $\R$ contains by construction the relation $\bisim{} \cdot \peercmpl[\B]$,
and in the rest of the proof we show that $ {\R} \subseteq \FunCpeer( \R, \B )$.

Fix a pair $\rho' \R \sigma'$,
the construction of $\R$ ensures that there exist
three contracts $\rho_1, \rho_2, \sigma \in \sc$ that enjoy the following properties
$$
\rho_1 \bisim{} \rho_2, 
\quad \rho_2 \peercmpl[\B] \sigma, 
\quad \rho_1 \arstar{\tau} \rho', \quad \sigma \arstar{\tau} \sigma'
$$
\rdef{sb},  $\rho_1 \bisim{} \rho_2$, and  $\rho_1 \arstar{\tau} \rho'$ imply that there exists
a $\rho'_2$ such that $\rho_2 \arstar{\tau} \rho'_2$ and $ \rho' \bisim{} \rho'_2$.
  \rPt{peercmpl-derivatives} of \rdef{peercmpl}, $\rho_2 \peercmpl[\B] \sigma$, and $\rho_2 \arstar{\tau} \rho'_2$ let us prove that $ \rho'_2 \peercmpl[\B] \sigma'$.

  \rdef{peercmpl} requires us to prove two facts,
  namely 
  \begin{enumerate}[label=(\roman*)]\itemsep 0pt
    
  \item\label{pt:peercmpl-1} 
    if $ \rho'  \Par \sigma' \stable[\B]$ then $ \rho \ar{\ok}$, $\sigma'
    \ar{\ok}$
    
  \item\label{pt:peercmpl-2}
    if $ \rho' \Par \sigma' \ar[\B]{\tau}  \rho ''\Par\sigma''$ then
    $\rho '' \R \sigma'' $

 \end{enumerate}
To prove \rpt{peercmpl-1} suppose that $ \rho'  \Par \sigma' \stable[\B]$.
Since  $\rho' \bisim{} \rho'_2$,
we apply \rlem{bisim-preserves-stuck-states}, and obtain that $\rho'_2 \Par \sigma' \stable[\B]$. Since $\rho'_2 \peercmpl[\B] \sigma'$,
\rpt{peercmpl-stuck-ok} of \rdef{peercmpl} imply that $\rho'_2 \ar{\ok}$
and $\sigma' \ar{\ok}$. Now $\rho' \bisim{} \rho'_2$ and \rdef{sb} ensure
that $\rho' \ar{ \ok }$, as required.

Now we prove \rpt{peercmpl-2}. 
Suppose that $\rho' \Par \sigma' \ar[\B]{\tau} \rho''\Par\sigma''$,
we have to show that $\rho'' \R \sigma''$.
The argument is by case analysis on the rule
of \rfig{interaction} used to derive the silent move at hand.
If \rname{I-Left} were applied, then $\rho' \ar{\tau} \rho'' $ and $\sigma' = \sigma''$.
It follows that $\rho_1 \arstar{\tau} \rho''$, thus the construction of $\R$ ensures that
$\rho'' \R \sigma''$. If \rname{I-Right} was applied, we reason in the same manner.
Suppose now that \rname{I-Synch} was applied. In this case
$ \rho' \ar{ \lambda } \rho''$ and $ \sigma' \ar{ \lambda' } \sigma''$ for some $\lambda$ and
$\lambda'$ such that $ \lambda \bowtie_{\B} \lambda' $.
To prove that $\rho'' \R \sigma''$ we exhibit a contract
$\hrho'_2$ which enjoys the following two properties,
\begin{center}
$\rho' \bisim{} \hrho'_2, \qquad \hrho'_2 \peercmpl[\B] \sigma'$
\end{center}
Since $\rho' \bisim{} \rho'_2$ and $\rho' \ar{ \lambda } \rho''$,
\rdef{sb} ensures that there exists a contract $ \hrho'_2$ such that
$\rho'_2 \ar{ \lambda } \hrho'_2$ and that $\rho'' \bisim{} \hrho'_2$.
The action performed by $\rho'_2$ lets us infer the hand-shake $\rho'_2 \Par \sigma' \ar[\B]{\tau} \hrho'_2 \Par \sigma''$.
Now $\rho'_2 \peercmpl[\B] \sigma'$, and \rpt{peercmpl-derivatives} of \rdef{peercmpl} ensure
that $\hrho'_2 \peercmpl[\B] \sigma''$. Observe that $\hrho'_2$ enjoys the two properties we required
above, thus $\rho'' \R \sigma''$.
\end{proof}

\begin{defi}\name{$\B$-peer subcontract preorder}\label{def:peer.subc}\\
For $\sigma_1,\, \sigma_2 \in \sc$ let
$\sigma_1 \peerleq[\B] \sigma_2$ whenever $\rho \peercmpl[\B] \sigma_1$
implies $\rho \peercmpl[\B] \sigma_2$, for every $\rho \in \sc$.
We use $\sigma_1 =^{\B}_\peer \sigma_2$ to denote the equivalence associated to $\peerleq[\B]$.
\end{defi}

\subsection{Syntactic characterisation}
\label{sec:syntactic-peer-preorders}

The parametrised peer subcontract preorder $\sigma_1 \peerleq[\B] \sigma_2$ is set based, and quantifies over the result of all 
peers in $\B$-compliance with $\sigma_1$. However, because of the restricted nature of higher-order contracts, \rfig{sc}, it turns out
that $\peerleq[\B]$ can be characterised by the syntactic structure of
$\sigma_1$ and $\sigma_2$, at least for relations $\B$
which satisfy certain minimal conditions.

Let $\FunS{}: \parts{\sc^2} \times \parts{\sc^2}
\longrightarrow  \parts{\sc^2}$ be the functional such that 
$(\sigma_1, \, \sigma_2) \in \FunS{}(\R, \B)$ whenever one of the
following holds:
\begin{enumerate}[label=(\roman*)] 
  \leaveout{  \item 
    \GBc{
      if $\unfold{ \sigma_1 } = \# ( \sigma'_1 ).\Unit$ then $\unfold{ \sigma_2 } = \# (\sigma'_2 ).\Unit$, $\sigma'_1 \R \sigma'_2$ and $\sigma'_2 \R \sigma'_1$
    }
  }
  
\item\label{pt:base-sempeerleq} 
  if $\unfold{\sigma_1}= \Unit$ then $\unfold{\sigma_2}= \Unit$
  
\item 
  if $\unfold{\sigma_1}= {?\gt{t}_1.\sigma'_1}$ then
  $\unfold{\sigma_2}= {?\gt{t}_2.\sigma'_2}$ and $ \sigma'_1 \R \sigma'_2$ and
  $\gt{t}_1 \bsubtype \gt{t}_2$
  
\item 
  if $\unfold{\sigma_1}= {!\gt{t}_1.\sigma'_1} $ then
  $\unfold{\sigma_2}= {!\gt{t}_2.\sigma'_2}$ and $  \sigma'_1 \R \sigma'_2$ and
  $\gt{t}_2 \bsubtype \gt{t}_1$
  
\item\label{pt:ho-output} 
  if $\unfold{\sigma_1}=  {!(\sigma''_1).\sigma'_1}$ then
  $\unfold{\sigma_2}= {!(\sigma''_2).\sigma'_2}$ and $  \sigma'_1 \R
  \sigma'_2$ and $\sigma''_2 \B \sigma''_1$

\item\label{case:ho-input} 
    if $\unfold{\sigma_1}= { ?(\sigma''_1).\sigma'_1 }$ then
    $\unfold{\sigma_2}= {?(\sigma''_2).\sigma'_2 }$ and
    $  \sigma'_1 \R \sigma'_2 $ and $ \sigma''_1 \B \sigma''_2 $

  \item 
    if $\unfold{\sigma_1} = \extsum_{i \in I}?\gt{l}_i.\sigma^1_i $ then
    $\unfold{\sigma_2} = \extsum_{j \in J}?\gt{l}_j.\sigma^2_j$ 
 where  ${I}  \subseteq {J}$ and $\sigma^1_i \R \sigma^2_i$ for all  $i\in I$

  \item 
    if $\unfold{\sigma_1}=\intsum_{i \in I}!\gt{l}_i.\sigma^1_i $ then
    $\unfold{\sigma_2}= \intsum_{j \in J}!\gt{l}_j.\sigma^2_j$
    where  ${J} \subseteq {I}$ and $\sigma^1_j \R \sigma^2_j$ for all $j
    \in J$
  \end{enumerate}

\begin{lem}
\label{lem:FunS-monotone}
  The functional $\FunS{}$ is monotone in both arguments:
  \begin{enumerate}[label=\({\alph*}]\itemsep0pt
    \item\label{pt:FunS-monotone-R}
      Fix a $\B$. If ${\R}  \subseteq {\R'}$, 
      then ${\apFunS{\R, \B}} \subseteq {\apFunS{\R', \, \B}}$
    \item\label{pt:FunS-monotone-B}
      Fix a $\R$. If ${\B} \subseteq {\B'}$, 
      then ${\apFunS{\R, \B}}
      \subseteq {\apFunS{\R, \, \B'}}$
\end{enumerate}
\end{lem}
\begin{proof}
  The proofs of both a) and b) are straightforward.
\end{proof}

\begin{defi}\name{$\B$-syntactic peer preorder}\label{def:sempeerleq}\\
If~$ {\R} \subseteq \FunS{}( \R, \,  \B )$, then we say that~$\R$
  is a {\em $\B$-coinductive peer preorder}.
  Fix a $\B$. Standard arguments based on \rpt{FunS-monotone-R} of the previous lemma ensure that there exists the greatest
  solution of the equation~$ X = \FunS{}( X, \, \B ) $; we call this  solution the {\em $\B$-syntactic peer preorder}, and we denote
  it by $\sempeerleq[\B]$.
\end{defi}

One immediate consequence of the  syntactic nature of 
this preorder is that it is preserved by unfolding.
This is the first part of the following lemma:
\begin{lem}
 \label{lem:FunS-closed-unfold} 
\qquad
\begin{enumerate}
\item \label{pt:FunS-closed-unfold}
For every $\sigma \in \sc$,  $\sigma \sempeerleq[\B] \unfold{\sigma}
\sempeerleq[\B]  \sigma$. 

\item\label{pt:FunS-stable} If $\sigma_1 \sempeerleq[\B] \sigma_2$ 
and $\sigma_2$ is stable, then 
there exist some stable 
$\sigma'_1$ such that $\sigma_1 \arstar{\tau} \sigma'_1$
and that $\sigma'_1 \sempeerleq[\B] \sigma_2$.
\end{enumerate}

\end{lem}
\begin{proof}
\rPt{FunS-closed-unfold}
of the lemma follows from the fact that 
$ \rho \apFunS{\R,\B}  \sigma$
if and only if  $ \unfold{\rho} \apFunS{\R,\B}  \unfold{\sigma}$.
This property of $\apFunS{-}$ in turn relies on the fact that 
$\unfoldSym$ is idempotent, that is 
$$\unfold{\unfold{\sigma}} = \unfold{\sigma}$$

To prove \rpt{FunS-stable} suppose $\sigma_1 \sempeerleq[\B] \sigma_2$. Part (1) ensures 
that $\unfold{\sigma_1} \sempeerleq[\B] \sigma_2$. 
If $\unfold{\sigma_1}$ is not stable then it must be an internal sum 
of the form 
$
\bigoplus_{i \in I}!\gt{l}_i.\rho_i 
$
and since $\sigma_2$ is stable it must be of the form 
$!\gt{l}_k.\rho'_k$ for some $k \in I$. The required
$\sigma'_1$ in this case is $!\gt{l}_k.\rho_k$. 
\end{proof}

\leaveout{
\noindent
The next example is an instance of \rpt{FunS-closed-unfold} the previous lemma.
\begin{exa}
  Let $\rho = \Rec[x]{!\gt{moka}.x}$, $\sigma =
  \Rec[x]{!\gt{moka}.!\gt{moka}.x}$, and  
  $$
{\R} = \sset{ (\rho, \sigma), \, (\rho, !\gt{moka}.\sigma )
}
$$
  First we show that $\R$ is prefixed point of $ \FunS{}$.
  To prove that $\rho \mathrel{\FunS{}(\R)} \sigma$ we unfold
  the terms,
  $ \unfold{\rho} = {!\gt{moka}.\rho}$, 
  $ \unfold{\sigma} = {!\gt{moka}.!\gt{moka}.\sigma}$ and check that 
  $ \rho \R {!\gt{moka}.\sigma}$. The last fact is true by
  construction.
  To prove that $ \rho \mathrel{\FunS{}(\R)} !\gt{moka}.\sigma$,
  we unfold the terms, and check that 
  $ \rho \R \sigma$ this is true by construction.
  
  To prove that $ \unfold{\rho} \mathrel{\FunS{}(\R)} \unfold{\sigma}$
  we have to show that $\unfold{\unfold{\rho}}$ and $\unfold{\unfold{\sigma}}$
  are related as required by \rdef{sempeerleq}.
  In view of the idempotency of $\unfoldSym$, it is enough
  to check that $\unfold{\rho}$ and $\unfold{\sigma}$ are related are required.
  But this is exactly what we proved to show that $\R$ is prefixed point of $ \FunS{}$.
\end{exa}
}

The reflexivity of $\sempeerleq[\B]$ depends tightly on the reflexivity
of $\B$.
\begin{lem}
\label{lem:reflexive-sempeerleq}
  If $\B$ is reflexive then $\sempeerleq[\B]$ is reflexive.
\end{lem}
\begin{proof}
  It suffices to prove that the identity relation $\I$ is contained
  in $\sempeerleq[\B]$. To prove this, we show that ${\I} \subseteq
 {\apFunS{ \I, \, \B}}$. Fix a pair $ \sigma \I \sigma $. The reasoning
  is by case analysis on $\sigma$, and the only two cases worthwhile involve a higher-order $\sigma$. We discuss one such case. 

Suppose that $\sigma = {!(\sigma^m).\sigma'}$. To show that $ \sigma \apFunS{ \I,  \, \B} \sigma$, we have to explain why $\sigma' \I \sigma'$ and $\sigma^m \B \sigma^m$.
The first fact follows form the reflexivity of $\I$, and the second from the reflexivity of the relation~$\B$.
\end{proof}
\noindent
In the previous lemma the hypothesis of $\B$ being reflexive is not
only sufficient, but also necessary for $\sempeerleq[\B]$\leaveout{$\apFunS{\R, \, \B}$} to be reflexive.
\begin{exa}
\label{exa:reflexivity-Speer}
In this example we prove that if $\B$ is not reflexive then
$\sempeerleq[\B]$ need not be reflexive.

Let $ \sigma = {!( \Unit ). \Unit}$.
The empty binary relation $\emptyset$ is not reflexive because
$(\Unit, \, \Unit)$ is not in $\emptyset$, so we take $\emptyset$ as our candidate $\B$.
In turn this implies that $ (\sigma, \sigma) \not \in {\apFunS{\sempeerleq[\emptyset], \, \emptyset}}$, because $\Unit$ does not satisfy the conditions
required by case (\ref{pt:ho-output}) of \rdef{sempeerleq}. 
Since
${\sempeerleq[\emptyset]} = {\apFunS{ \sempeerleq[\emptyset], \,
  \emptyset}}$, it follows that $ \sigma \Nsempeerleq[\emptyset]
\sigma$.
\leaveout{
In this example we prove that if $\B$ is not reflexive then
$\apFunS{\I, \, \B}$ need not be reflexive.

The empty binary relation $\emptyset$ is not reflexive because
$(\Unit, \, \Unit)$ is not in $\emptyset$.
In turn this implies that $ \sigma \NapFunS{\sempeerleq[\semptyset], \, \emptyset} \sigma$, because $\Unit$ does not satisfy the conditions required by case (\ref{pt:ho-output}) of \rdef{sempeerleq}. Since $\sempeerleq[\emptyset] = \apFunS{ \empeerleq[\emptyset], \, \emptyset}$, it follows that $ \sigma \Nsempeerleq[\emptyset] \sigma$.
In turn this implies that $ \sigma \NapFunS{\R, \, \emptyset} \sigma$ for every $\R$, because $\Unit$ does not satisfy the conditions required by case (\ref{pt:ho-output}) of \rdef{sempeerleq}.
It follows that $ \sigma \Nsempeerleq[\emptyset] \sigma$.
}
\end{exa}

Our intention is to show that the set-theoretic relation 
$\sigma_1  \peerleq[\B] \sigma_2$ coincides with the syntactically defined relation $\sigma_1 \sempeerleq[\B] \sigma_2$, provided $\B$ satisfies
some simple properties.
In one direction the proof requires the following technical lemma
showing that $ \peerleq[\B] $ preserves the ability of contracts to
interact.

\begin{lem}
\label{lem:compl.moves}
Suppose $\rho \peercmpl[\B] \sigma_1$ and  $\sigma_1  \peerleq[\B]
\sigma_2$ for some $\B$,
where all of $\sigma_1,\sigma_2,\rho$ are stable.
Then 
$\rho \Par \sigma_2 \ar[\B]{\tau}$ implies
$\rho \Par \sigma_1 \ar[\B]{\tau}$.
\end{lem}
\begin{proof}
  We know that $\rho \Par \sigma_2 \ar[\B]{\tau}   \rho' \Par \sigma_2'$ for
  some pair $\rho', \sigma'_2$. Since $\rho,\sigma_2$ are stable
  this reduction must involve interaction between these contracts. That is the
  derivation of the reduction must end with an application of rule 
  $\rname{I-Synch}$ from \rfig{interaction}.  The side condition and
  the premises of the rule ensure that $\rho \ar{ \lambda_1}$ for some
  $\lambda_1$. In turn this move can only be inferred by an
  application of rule \rname{a-Pre} or \rname{a-Ext} of
  \rfig{opsem}. In both cases one sees that $\rho \nar{\ok}$. But by
  hypothesis $\rho \peercmpl[\B] \sigma_1$, so \rpt{peercmpl-stuck-ok}
  of \rdef{peercmpl} ensures that $ \rho \Par \sigma_1
  \ar[\B]{\tau}$.
\end{proof}

\begin{cor}
  \label{cor:sound.ancillary}
Suppose  
$\sigma_1 \sempeerleq[\B]  \sigma_2$ where $\B$ is transitive
and 
$ \rho \peercmpl[\B] \sigma_1$.
Then
 $\rho \Par \sigma_2 \nar[\B]{\tau}$ 
implies 
$\rho \ar{\ok}$ and  $\sigma_2 \ar{\ok}$
 \end{cor}
\begin{proof}
We know that $\rho$ and $\sigma_2$ are stable, and we let $\sigma'_1$ be
the stable contract guaranteed by the second part 
\rlem{FunS-closed-unfold}  such that 
$\sigma_1 \arstar{\tau} \sigma'_1$ and $\sigma'_1 \sempeerleq[\B]
\sigma_2$. Simple properties of compliance ensure that
$ \rho \peercmpl[\B] \sigma'_1$. Therefore by \rlem{compl.moves} 
we know  that $\rho \ar{\ok}$ and $\sigma'_1 \ar{\ok}$. 
This means that $\sigma'_1$ is actually $\Unit$, so,  since 
$\sigma'_1 \sempeerleq[\B] \sigma_2$ and $\sigma_2$ is
stable, $\sigma_2$ must also be $\Unit$; that is $\sigma_2 \ar{\ok}$
as required. 
\end{proof}

\begin{thm}
\label{thm:sem-sound}
  Let $\B$ be a transitive relation on session contracts. Then 
  $\sigma_1\! \sempeerleq[\B]\! \sigma_2$ implies  $\sigma_1 \peerleq[\B]
  \sigma_2$.
\end{thm}
\begin{proof}
  It suffices to show that the following relation is a 
  $\B$-coinductive peer compliance,
  $$
  {\R} = \setof{ (\rho, \, \sigma_2)}{ \rho \peercmpl[\B] \sigma_1,
    \, \sigma_1 \sempeerleq[\B] \sigma_2, \, \text{for some } \sigma_1 \in \sc}
  $$
This requires establishing two properties. 
\begin{enumerate}[label=(\roman*)]
\item If  $ \rho \Par \sigma_2 \nar[\B]{\tau}$ whenever $ \rho \R
  \sigma_2$ then   $\rho \ar{\ok}$ and $\sigma_2 \ar{\ok}$.
  This follows from \rcor{sound.ancillary}. 

\item 
  Suppose $ \rho \Par \sigma_2 \ar[\B]{\tau} \rho' \Par
  \sigma'_2$,  where   $ \rho \R \sigma_2$.     We have to prove that
  $ \rho' \R \sigma'_2$.
  
Because $ \rho \R \sigma_2$ we know there exists some 
$\sigma_1$ such that $\rho \peercmpl[\B] \sigma_1$
  and $\sigma_1 \sempeerleq[\B] \sigma_2$. We have to find some 
$\sigma'_1$ such that   $\rho' \peercmpl[\B] \sigma_1'$
  and $\sigma'_1 \sempeerleq[\B] \sigma'_2$.
The nature of $\sigma'_1$ depends on the form of the reduction
$ \rho \Par \sigma_2 \ar[\B]{\tau} \rho' \Par \sigma'_2$.

If this is a silent move by $\rho$ to $\rho'$ we can take $\sigma'_1$ 
to be $\sigma_1$ itself, since $\rho' \peercmpl[\B] \sigma_1$. 
If it is a silent move by $\sigma_2$ then there are two cases.
If the move is inferred using the rule 
\rname{a-Unfold} from \rfig{opsem} then 
$\unfold{\sigma_2} = \unfold{\sigma'_2}$, which ensures that 
$\sigma_1 \sempeerleq[\B] \sigma'_2$.  

The only other possible way for $\sigma_2$ to do a silent move is
by an application of rule $\rname{a-Int}$. Here $\sigma_2$ has the form 
$\intsum_{i \in I} !\gt{l}_i.\sigma^2_i$ and $\sigma'_2$ must be 
$ !\gt{l}_k.\sigma^2_k$ for some $k \in I$.  In this case again we can
take $\sigma'_1$ to be $\sigma_1$, since $\sigma_1 \sempeerleq[\B] \sigma'_2$.

  If the reduction  is due to an interaction between $\rho$ and
  $\sigma_2$
  then there are six cases
  to discuss. The argument for each of them is similar, so we discuss
  only one case involving higher-order communication.
  Suppose that $ \rho = {?( \rho^m ).\rho'}$. Since $\sigma_2$ engages
  in a communication with $\rho$ it must be the case that
  $\sigma_2 = {!(\sigma^m_2).\sigma'_2}$ with  $  \sigma^m_2 \B \rho^m$.
  Since $ \sigma_1 \sempeerleq[\B] \sigma_2 $, it follows that
  $\unfold{\sigma_1} = {!(\sigma^m_1).\sigma'_1}$ with $ \sigma^m_1 \B
  \sigma^m_2$ and $ \sigma'_1 \sempeerleq[\B] \sigma'_2$.
  It remains to show that $\rho' \peercmpl[\B] \sigma'_1$.
  The transitivity of $\B$ ensures that $  \sigma^m_1 \B  \rho^m $,
  so we can infer $ \rho \Par \sigma_1 \arstar[\B]{\tau} \rho' \Par
  \sigma'_1$, which implies that $\rho' \peercmpl[\B] \sigma'_1$.\qedhere
\end{enumerate}
\end{proof}

\begin{exa}
  \label{ex:soundness-transitivity}
  We show that if $\B$ is not a transitive relation, then 
  $ \sempeerleq[\B]$ need not be contained in $\peerleq[\B]$,
  that is 
  $ {\sempeerleq[\B] } \not \subseteq { \peerleq[\B] }$.

  Let 
  $
  {\B} = \sset{ (\Unit, \, !\gt{l}.!\gt{l}.\Unit), ( !\gt{l}.!\gt{l}.\Unit, \, !\gt{l}.\Unit)}
  $.
  This relation is not transitive, because $\Unit \B {!\gt{l}.!\gt{l}.\Unit}$ and $!\gt{l}.!\gt{l}.\Unit \B {!\gt{l}.\Unit}$,
  while 
$ (\Unit,  {!\gt{l}.\Unit}) \not \in {\B}$.

  Let $ \sigma_1 = {!( !\gt{l}.!\gt{l}.\Unit ).\Unit}$ and let $\sigma_2
  = {!( \Unit ).\Unit}$. We show that $\sigma_1  \sempeerleq[\B]
  \sigma_2  $. 
  The witness of this fact is the relation 
  $
  {\R}  = \sset{ ( \sigma_1, \, \sigma_2 ), ( \Unit, \, \Unit ) } 
  $.
  We are required to prove that $ {\R} \subseteq {\apFunS{\R, \,
      \B}}$.
  This amounts in showing that 
  a) $ \Unit \apFunS{\R, \, \B} \Unit$, and b)
  $ \sigma_1 \apFunS{\R, \, \B} \sigma_2$.
  Point a) is true thanks to case (\ref{pt:base-sempeerleq}) of
  \rdef{sempeerleq}, and point b) follows from
  case (\ref{pt:ho-output}) of the same definition.

  Now we prove that $\sigma_1  \Npeerleq[\B] \sigma_2$.
  We have to exhibit a session contract $\rho$, such that $ \rho
  \peercmpl[\B] \sigma_1$, and $ \rho \Npeercmpl[\B] \sigma_2$.
  Let $ \rho = { ?(!\gt{l}.\Unit).\Unit }$. 
  To see why $ \rho \peercmpl[\B] \sigma_1 $, note 
  that the relation
  $
  \sset{ (\rho, \, \sigma_1), ( \Unit, \, \Unit) }
  $
  is a $\B$-coinductive mutual compliance.

  To conclude the example, we have to prove that $ \rho \Npeercmpl[\B] \sigma_2 $.
  The witness that $\B$ is not transitive is
  the fact that 
  $(\Unit, !\gt{l}.\Unit) \not \in {\B}$. 
  This implies that
  $!( \Unit ) \not \bowtie_{\B} ?( !\gt{l}.\Unit )$, and in turn that
  $ \rho \Par \sigma_2 \nar[\B]{\tau}$. 
  Since $\rho \nar{ \ok }$, it follows that $ \rho \Npeercmpl[\B]
  \sigma_2$.
\end{exa}

The converse to \rthm{sem-sound} relies on
following property of session contracts, whose proof is relegated to
 \rsec{complement}; see
 \rthm{complements-comply}.

\begin{thm}
\label{thm:dual-cmpl-with-ctr}
  Let $\B$ be a preorder on session contracts.
  For every session contract $\rho$ there exists a session contract
  $\prdual{\rho}$ such that $\rho \,\peercmpl[\B]\, \prdual{\rho}$.
\end{thm}

\begin{thm}
\label{thm:sem-complete}
  Let $\B$ be a preorder on session contracts.
  Then $ \sigma_1 \peerleq[\B] \sigma_2$ implies $\sigma_1 \sempeerleq[\B] \sigma_2$.
\end{thm}
\begin{proof}
  Since $ \sigma_1 \peerleq[\B] \sigma_2$ implies that
  $ \unfold{\sigma_1} \peerleq[\B] \unfold{\sigma_2}$, it is enough to
  prove that $\R$  is a $\B$-coinductive peer preorder, 
${\R} \subseteq {\apFunS{\R,\B}}$, 
where $\R$ is given by 
$$
{\R} = \setof{ (\sigma_1, \, \sigma_2)}{ \unfold{\sigma_1} \peerleq[\B] \unfold{\sigma_2}}
$$
Pick a pair $\sigma_1 \R \sigma_2$.
To show that   
 ${\sigma_1} \apFunS{\R,\B} { \sigma_2 }$ 
we reason by case analysis on the
unfoldings of these contracts; the argument for many cases are
similar, so we only discuss two cases.
\begin{itemize}
\item Suppose that $\unfold{\sigma_1} = \Unit$. The relation $\sset{
    (\Unit, \, \Unit)}$ is a coinductive $\B$-mutual compliance, so
  $\Unit \peercmpl[\B] \unfold{\sigma_1}$. As $\unfold{\sigma_1} \peerleq[\B]
  \unfold{\sigma_2}$, it follows that $\Unit \peercmpl[\B]
  \unfold{\sigma_2}$. 
A simple argument, based on the possible structure of $\sigma_2$,
will show that  $\unfold{\sigma_2}$ is stable,
  from which it follows that 
$\Unit \Par \unfold{\sigma_2} \nar{\tau}$. 
Compliance now ensures that $\unfold{\sigma_2} \ar{\ok}$. 
Because of the restrictive syntax for contracts this is only possible
if $\unfold{\sigma_2}$ is actually $\Unit$. 
It follows that  ${\sigma_1} \apFunS{\R,\B} {\sigma_2}$.

\item Suppose that $\unfold{\sigma_1} = { !( \sigma^m_1).\sigma'_1}$.
We have to prove the equality 
$$
\unfold{\sigma_2} = {!( \sigma^m_2).\sigma'_2}
$$ 
where $ \sigma^m_2 \B \sigma^m_1$, and $ \sigma'_1 \R \sigma'_2$.

  \rthm{dual-cmpl-with-ctr} and the
  hypothesis on $\B$ ensures the existence of some contract 
  $\prdual{\sigma'_1}$
  such that 
  $\sigma'_1 \peercmpl[\B]  \prdual{\sigma'_1}  $.
  Let $\rho = {?( \sigma^m_1 ). \prdual{\sigma'_1} }$ and let 
  $${\R'} = \sset{ ( \rho, \, \unfold{\sigma_1} ) }  \, \cup \peercmpl[\B]$$
  the relation $\R'$ is a $\B$-coinductive peer compliance. This is the case because
  the preorder $\B$ is reflexive, thus there exists the derivation 
  $$
  \infer[ ?( \sigma^m_1 )  \bowtie_{\B} !( \sigma^m_1 ) ]
  { \rho \Par  \unfold{\sigma_1} \ar[\B]{\tau} \prdual{\sigma'_1} \Par \sigma'_1}
  { \rho \ar{ ?( \sigma^m_1 ) } \prdual{\sigma'_1} \quad  \unfold{\sigma_1} \ar{  !( \sigma^m_1) } \sigma'_1  }
  $$
  and because, thanks to the symmetry of $\peercmpl[\B] $,  $ \prdual{\sigma'_1} \peercmpl[\B] \sigma'_1 $.

  It follows that $\rho  \peercmpl[\B]  \unfold{\sigma_1}$, 
  and since 
  $\unfold{\sigma_1} \peerleq[\B] \unfold{\sigma_2}$, 
  we obtain immediately that $ \rho \peercmpl[\B] \unfold{\sigma_2}  $.
  
  Since $\rho \nar{\ok}$, the composition $\rho \Par
  \unfold{\sigma_2}$ performs a silent move.
  The syntax of session contracts ensures that
  $\unfold{\sigma_2}$ cannot be an internal sum,
  and therefore $\unfold{\sigma_2}$ has to interact
  with $\rho$. 
  The definition of $\bowtie_{\B}$ ensures that 
  $\unfold{\sigma_2}$  has to have the form 
  $ !( \sigma^m_2).\sigma'_2$ where 
  $ \sigma^m_2 \B \sigma^m_1$, our first requirement.
  Moreover 
  $\rho \Par \unfold{\sigma_2} \ar[\B]{\tau}
  \rho' \Par \sigma'_2$ from which  $\rho' \peercmpl[\B] \sigma'_2$
  follows. 
  
  To show the second requirement, 
  $ \sigma'_1 \R \sigma'_2$, let $\rho'$ be any contract satisfying the condition
  $\rho' \peercmpl[\B] \unfold{\sigma'_1}$;
  we have to prove that 
  $\rho' \peercmpl[\B] \unfold{\sigma'_2}$.
  Note that because of \rlemP{FunS-closed-unfold}{FunS-closed-unfold}
  we can also assume that  $\rho' \peercmpl[\B] \sigma'_1$.
  We \leaveout{can therefore}repeat the above argument to establish that
  $?( \sigma^m_1 ). \rho'  \peercmpl[\B] \unfold{\sigma_2}  $,
  and since 
  $$
  ?( \sigma^m_1 ). \rho'  \Par \unfold{\sigma_2}  \ar[\B]{\tau}  
  \rho' \Par \sigma'_2
  $$
  we know that $\rho'  \peercmpl[\B] \sigma_2  $.
  The required 
  $\rho' \peercmpl[\B] \unfold{\sigma'_2}$
  now follows by another application of 
  \rlemP{FunS-closed-unfold}{FunS-closed-unfold}.\qedhere
\end{itemize}
\end{proof}

\begin{exa}
\label{ex:complentess-requires-reflexivity}
We show that if $\B$ is not a reflexive relation, then 
$ \peerleq[\B]$ needs not be contained in $\sempeerleq[\B]$:
$ {\peerleq[\B]} \not \subseteq {\sempeerleq[\B]}$.

Recall \rexa{reflexivity-Speer}, and the session contract we
employed there, $ \sigma = {!( \Unit ). \Unit}$. 
We know that $ \sigma \peerleq[\emptyset] \sigma$, because no peer
can interact with $  !( \Unit ) $, so no peer complies with $\sigma$.
However in \rexa{reflexivity-Speer} we have proven that 
$ \sigma \Nsempeerleq[\emptyset] \sigma$.
\end{exa}

\begin{cor}\label{cor:set.syn}
  For any preorder $\B$ over session contracts, 
  $\sigma_1 \sempeerleq[\B] \sigma_2$ if and only if   $\sigma_1 \peerleq[\B] \sigma_2$.
\end{cor}
\begin{proof}
  Follows immediately from Theorem~\ref{thm:sem-sound} and 
Theorem~\ref{thm:sem-complete}. 
\end{proof}

\section{Modelling session types}
\label{sec:results}

In our treatment session types and contracts are obviously just syntactic variations on each other. We formalise the relationship between them
  as a function which 
maps session types from ${\st}$ to session contracts from $\sc$, $$\encSym: L_{\st} \longrightarrow L_{\sc}$$ 
We then show that the subtyping relation between session types,
$S \subt T$, 
can be modelled precisely by the set-based contract preorder, 
$\M{S} \peerleq[\B]  \M{T}$,
for a particular choice of $\B$.

The interpretation of types into contracts is defined by the following 
syntactic translation:
\[
\M{S} = 
\begin{cases}
  \Unit &\text{if }S=\e,\\
  !\gt{t}.\M{S'}&\text{if }S=\sout{\gt{t}}S',\\
  ?\gt{t}.\M{S'}&\text{if }S=\sinp{\gt{t}}S',\\
  !( \M{T} ).\M{S'}&\text{if }S=\sout{T}S',\\
  ?( \M{T} ).\M{S'}&\text{if }S=\sinp{T}S',\\
  \extsum_{i\in[1;n]}?\gt{l}_i.\M{S_i}&\text{if
  }S=\branch{\gt{l}_1\as S_1,\ldots, \gt{l}_n\as S_n},\\
  \intsum_{i\in[1;n]}!\gt{l}_i.\M{S_i}&\text{if
  }S=\choice{\gt{l}_1\as S_1,\ldots, \gt{l}_n\as S_n},\\
\Rec[x]{\M{S'}}&\text{if }S=\Rec{S'},\\
x&\text{if }S=X
\end{cases}
\]
\noindent
The function $\invEncSym : L_{\sc} \longrightarrow L_{\st}$ is the
obvious inverse of $\encSym$, for instance $$\invM{!\gt{t}.\sigma} = \sout{\gt{t}}\invM{\sigma},$$
and we omit its definition.

Because of the syntactic nature of $\encSym$ and $\invEncSym$
the following properties are easy to establish.
\begin{lem}
  \label{lem:M-sanity-checks}
 For every $S, T \in L_{\st}$ and $\rho,\sigma \in L_{\sc}$,
  \begin{enumerate}[label=\alph*)]
  \item \label{pt:M-subst-distribute}$\M{S\subst{T}{X}}=(\M{S})\subst{\M{T}}{\M{X}}$
  \item \label{pt:invM-subst-distribute}$\invM{ \rho \subst{\sigma}{x}}=(\invM{\rho})\subst{\invM{\sigma}}{\invM{x}}$
  \item \label{pt:M-unfold-commute}$\unfold{\M{T}}=\M{\unfold{T}}$
  \item \label{pt:invM-unfold-commute}$\unfold{\invM{\sigma}}=\invM{\unfold{\sigma}}$
  \item \label{pt:M-unfold} $\unfold{\invM{\sigma}}=T\text{ iff }\unfold{\sigma}=\M{T}$
  \end{enumerate}
\end{lem}
\begin{proof}
  \rPt{M-subst-distribute} and \rpt{invM-subst-distribute} are proven
  by structural induction, respectively on $S$ and $\rho$.
  \rPt{M-unfold-commute} uses rule induction on $\unfold{T}$ together
  with an application of
  \rpt{M-subst-distribute}. \rPt{invM-unfold-commute} is proven by
  rule induction on $\unfold{\sigma}$, and uses
  \rpt{invM-subst-distribute}. \rPt{M-unfold} is a consequence of
  (\ref{pt:M-unfold-commute}) and (\ref{pt:invM-unfold-commute}).
\end{proof}

In order to find the appropriate $\B$ \leaveout{with which }that captures the
subtyping relation $S \subt T$ via the interpretation $\encSym$,
we need to develop some properties of  functionals over contracts. 
Let ${\mathcal Pre}$ denote the collection of preorders over the set
of contracts $\sc$.
\begin{lem}
  $(\mathcal Pre, \subseteq)$ is a complete lattice. 
\end{lem}
\begin{proof}
  We have to show that all the subsets of ${\mathcal Pre}$ have
  infimum and supremum. Let ${X} \subseteq {{\mathcal Pre}}$.
  The infimum of $X$ is defined as the intersection of
  the elements of $X$, that is $ \bigsqcap X = \bigcap \setof{ \B }{{\B} \in X}$. 
  The supremum of $X$ is defined as
  the transitive closure of the union of the elements of $X$, that is
  $
  \bigsqcup X = ( \, \bigcup \setof{ \B }{{\B} \in X} \, )^{\mathbf +}
  $.
  It is routine work to check that 
  $\bigsqcap X  \subseteq {\B}$ and ${\B} \subseteq \bigsqcup X$ 
  for every ${\B} \in X$.
\end{proof}
\noindent
Let $\Fun{\peer}:  {\mathcal Pre} \longrightarrow {\mathcal Pre}$ be
defined by letting $\Fun{\peer}(\B)$ be  $\peerleq[\B]$.
\begin{prop}
\label{prop:Funpeer-monotone}
  $\Fun{\peer}$ is a monotone endofunction.
\end{prop}
\begin{proof}
  A priori there is no simple direct argument to show, using \rdef{peer.subc}, 
  that
  $\Fun{\peer}$ is monotonic. But the result is now a direct
  corollary of \rpt{FunS-monotone-B} of \rlem{FunS-monotone}, and of
  \rcor{set.syn}.
\end{proof}

\begin{defi}\name{Peer subcontract preorder}
\\We use  $\Peerleq$ to denote $\nu X . \Fun{\peer}(X)$, the greatest
\fp of the monotone function $\Fun{\peer}$. The existence of this \fp is
guaranteed by \rprop{Funpeer-monotone}. We refer to $\Peerleq$ as the {\em Peer subcontract preorder}.

We also let $\Peereq$ denote the {\em Peer equivalence} generated by
$\Peerleq$ in the obvious way.
\end{defi}

The properties of $\Peerleq$ alluded to in (\ref{eq:fixpoint})  of the
Introduction are now easy to establish. 
\begin{prop}
  $\Peerleq$ is the largest preorder $\B$ over $\sc$ satisfying: 
  $\sigma_1 \B   \sigma_2$ if and only if $\sigma_1 \peerleq[\B]
  \sigma_2$. 
\end{prop}
\begin{proof}
  A direct consequence of the fact that $\Peerleq$ is the greatest
  \fp of $\Fun{\peer}$. 
\end{proof}

The proof that 
$\Peerleq$ provides  a fully-abstract model of  subtyping $\subt$ on
session types {, \leaveout{via the
interpretation $\M{-}$, }relies on another characterisation, which in
turn uses a standard result from lattice theory
\cite[pag. 19]{arnold2001rudiments}.
\begin{lem}\name{Golden lemma}\\
  Let $L$ be a complete lattice and $f : L\times L \longrightarrow L$
  an endofunction monotone in both arguments. Then  $\nu y. \nu
  x. f(x, y) = \nu x . f( x, x)$.\footnote{The result proven in
    \cite{arnold2001rudiments} is more general; it pertains to
    both least and greatest fixed points of endofunctions.}
\end{lem}

\begin{lem}\label{lem:golden.consequence} 
  ${\Peerleq} = \nu X . \FunS{}(X, \, X)$.
\end{lem}
\begin{proof}
  By definition ${\Peerleq} = \nu X. \Fun{\peer}(X) = \nu X. {\peerleq[X]}$.
But by \rcor{set.syn}  the preorder $\peerleq[\B]$ coincides with 
the relation $\sempeerleq[\B]$ for any preorder $\B$. 
Since $\Fun{\peer}$ is a function over ${\mathcal Pre}$, 
 we have  ${\Peerleq} = \nu Y. {\sempeerleq[Y]}$.
\rdef{sempeerleq} lets us expand  $\sempeerleq[\B]$, thereby obtaining
the equality ${\Peerleq} = \nu Y.\nu X.{\apFunS{X, \, Y}}$.
\leaveout{
Expanding $\sempeerleq[\B]$ from \rdef{sempeerleq}
it follows the equality $\Peerleq = \nu Y.\nu X. \apFunS{X, \, Y}$.
}
The result is now a consequence of the Golden lemma. 
\end{proof}

To obtain the \fabs result, we show how the \pfps 
of the functionals $\FunS{}$ and~$\Fun{\subt}$ are related via~$\encSym$.

\newcommand{\T}{\mathrel{\mathcal{T}}}

\begin{lem}
\label{lem:model-sound}
  Fix a relation $\B$ such that ${\B} \subseteq {  \apFunS{\B, \,
    \B}}$ and let $${\T } = \setof{ ( \encSym^{-1}(\sigma_1), \,
    \encSym^{-1}(\sigma_2))}{ \sigma_1 \B \sigma_2 }$$
  Then $ {\T} \subseteq { \Fun{\subt}(\T)}$.
\end{lem}
\begin{proof}(Outline)
Fix a pair $ S_1 \T S_2$. These types are the images via
$\encSym^{-1}$ of two session contracts, respectively $\sigma_1$
and $\sigma_2$, such that $ \sigma_1 \B \sigma_2 $.

The proof proceeds by a case analysis on the structure of
$\unfold{\invM{\sigma_1}}$; we give the details of two cases.

\begin{itemize}
\item {\sloppy Suppose $\unfold{\invM{\sigma_1}} = \e$.  
     According to \rdef{subtyping}
     we then have to show that $\unfold{\invM{\sigma_2}} = \e$.
     Because of
     \rlemP{M-sanity-checks}{M-unfold} we know that $\unfold{\sigma_1}
     = \Unit$; case~(\ref{pt:base-sempeerleq}) of
     \rdef{sempeerleq} ensures that $\unfold{\sigma_2}
     = \Unit$, and \rlem{M-sanity-checks} therefore implies
     the syntactic equality
     $\unfold{\invM{\sigma_2}} = \e$.}

\item Suppose  $\unfold{\invM{\sigma_1}} = {\sout{ T_1 }S_1}$. 
  We are required to prove that 
  \begin{align}
    \label{eq:1}
    \unfold{\invM{\sigma_2}} = { \sout{T_2}S_2   }
  \end{align}
for some $T_2$ and $S_2$ such that 
  $T_2 \T T_1$ and $ S_1 \T S_2$.

  \rlemP{M-sanity-checks}{M-unfold} 
  ensures that $\unfold{\sigma_1}= { !( \M{T_1}).\M{S_1} }$. 
  Since we know that $\sigma_1 \B \sigma_2$, the hypothesis that ${\B} \subseteq 
  {\apFunS{\B, \,  \B}}$ and \rdef{sempeerleq} imply
  the equality 
  $$
  \unfold{\sigma_2} = {!( \sigma^m_2 ).\sigma'_2 }
  $$ 
  with
  $\sigma^m_2$ such that
  $\sigma^m_2 \B \M{T_1}$ and some~$\sigma'_2$ such that $\M{S_1} \B
  \sigma'_2$.
  The construction of~$\T$ implies that
  $$
  {\invM{\sigma^m_2}} \T T_1, \quad S_1 \T \invM{\sigma'_2}
  $$
  Because of $\sigma_2 = \M{S_2}$,
  \rlem{M-sanity-checks} implies the syntactic equality
  $ \unfold{\sigma_2} =
  \M{\unfold{S_2}}$, thus $ \unfold{S_2} =
  \sout{\invM{\sigma^m_2}}\invM{\sigma'_2}$.
  This ensures that (\ref{eq:1}) above is satisfied.
\end{itemize}

\noindent The proof for the remaining cases is similar to the argument already
shown, and  left to the reader.
\end{proof}

\begin{lem}
\label{lem:model-complete}
  Let $\T$ be a type simulation, and 
  $${\B } = \setof{ ( \M{S}, \,
    \M{T})}{ S \T T }$$
  Then $ {\B } \subseteq {  \apFunS{\B, \, \B} }$.
\end{lem}
\begin{proof}(Outline)
Suppose $\sigma_1 \B \sigma_2$. By construction $\sigma_1 = \M{S_1}$
and $\sigma_2 = \M{S_2}$ for some $S_1$ and $S_2$ related by $\T$.
The proof is a case analysis on $\unfold{\sigma_1}$.

\begin{itemize}
\item
  If $\unfold{\sigma_1} = \Unit$ we have to prove that
  $\unfold{\sigma_2} = \Unit$.
  An application of \rlemP{M-sanity-checks}{M-unfold}
  shows that $\unfold{S_1} = \e$.
  The hypothesis that $\T$ is a type simulation
  ensures that $\unfold{S_1} = \e$, so another application of
  \rlem{M-sanity-checks} leads to $\unfold{\sigma_2} =
  \Unit$.

\item
  If $\unfold{\sigma_1} = {?( \sigma^m_1 ).\sigma'_1}$ we have to show
  that 
  $$ \unfold{\sigma_2} = { ?( \sigma^m_2).\sigma'_2 }$$ 
  with $\sigma^m_1 \B \sigma^m_2$ and $\sigma'_1 \B \sigma'_2$.
  We apply \rlemP{M-sanity-checks}{M-unfold} and obtain the equality
  $$\unfold{S_1} = { \sinp{ \invM{\sigma^m_1} }\invM{\sigma'_1} }$$ 
  By hypothesis the relation $\T$ is a type simulation, so
  $ S_1 \T S_2 $ let us deduce that 
  $$\unfold{S_2} = \sinp{ T_2 } S'_2$$ 
  with 
  $$ \invM{\sigma^m_1} \T T_2, \quad \invM{\sigma'_1} \T S'_2$$
  This implies that
  $$
  \sigma^m_1 \B \M{T_2}, \quad \sigma'_1 \B \M{S'_2}
  $$
  \rlem{M-sanity-checks}, $\unfold{S_2} = \sinp{ T_2 }
  S'_2$, and $\sigma_2 = \M{S_2}$ ensure the equality
  $$\unfold{\sigma_2} = { ?( \M{T_2} ).\M{S'_2} }.$$
\end{itemize}
The other cases are analogous and left to the reader. 
\end{proof}

The last two results make the proof of \fabs
straightforward.

\begin{thm}\name{\Fabs}
\label{thm:full-abstraction}\\
For every $T, S \in {\st}$, $S \subt T$ if and only if $\M{S} \Peerleq \M{T}$.
\end{thm}
\begin{proof}
Suppose that $S \subt T$.
\rlem{model-complete} implies that the relation
$$
{\B} = \setof{ ( \M{S}, \, \M{T})}{ S \subt T }
$$
is contained in $ \apFunS{\B, \, \B}$, thus $ { \B } \subseteq  \nu X
.{\apFunS{X, \, X}}$. \rlem{golden.consequence}
implies that ${ \B } \subseteq { \Peerleq }$.
It follows that $\M{S} \Peerleq \M{T}$.

Suppose that $\M{S} \Peerleq \M{T}$.
Note that ${ \Peerleq } \subseteq { \apFunS{ \Peerleq, \, \Peerleq} }$.
\rlem{model-sound} implies that the relation
$$
{\T } = \setof{  ( \invM{\rho}, \, \invM{\sigma}) }{ \rho \Peerleq
  \sigma }
$$
is a type simulation, so ${ \T } \subseteq { \subt }$.
Since $S \T T$, it follows that $S \subt T$.
\end{proof}
\noindent
\Fabs has two immediate consequence.
The first is a result on the decidability of $\Peerleq$.
\begin{prop}
 If $\bsubtype$ is decidable, then relation $\Peerleq$ is decidable.
\end{prop}
\proof

  First we describe the an algorithm to decide $\subt$.
  In \cite[Figure~11, Lemma~10,
  Corollary~2]{DBLP:journals/acta/GayH05} an algorithm is presented,
  which decides $\subt$ but for a language of types with no 
  input/output of base types.
  Adding the following two rules to the ones in Figure~11 of that paper
  we obtain an algorithmic subtyping relation $\leqslant$, 
  that works also for types with input/output of base types.
  $$
\begin{array}{c}
  \infer[ \gt{t_1} \bsubtype \gt{t_2}]
  { \Sigma \typerel \sinp{\gt{t_1}}S_2 \leqslant \sinp{\gt{t_2}}S_2}
  { \Sigma \typerel S_1 \leqslant S_2 }
\\[1em]
  \infer[ \gt{t_2} \bsubtype \gt{t_1}]
  { \Sigma \typerel \sout{\gt{t_1}}S_2 \leqslant \sout{\gt{t_2}}S_2}
  { \Sigma \typerel S_1 \leqslant S_2 }
\end{array}
  $$
  Thanks to the hypothesis that $\bsubtype$ is decidable,
  Lemma~10, Corollary~2 of \cite{DBLP:journals/acta/GayH05} are true
  also for $\leqslant$ and $\subt$, that is for every session
  type $S_1$ and $S_2$
  \begin{enumerate}[label=\roman*)]
  \item\label{pt:alg-terminates}
    The algorithmic subtyping $ \typerel S_1 \leqslant S_2$ terminates
  \item\label{pt:alg-snd-comp} $\typerel S_1 \leqslant S_2$ if and only if $ S_1 \subt S_2$
  \end{enumerate}

  Now we show how to decide whether two session contracts $\sigma_1$
  and $\sigma_2$ are in the relation $\Peerleq$.
  \begin{enumerate}[label=\arabic*)]
    \item Let $S_1 = \invM{\sigma_1} $ and $S_2 = \invM{\sigma_2}$.
      The applications of the function~$\invEncSym$ terminates 
      because~$\invEncSym$ is defined inductively,
    \item apply the algorithmic subtyping to decide whether
      $\typerel S_1 \leqslant S_2$. \rPt{alg-terminates} above ensures
      that the algorithm terminates,
    \item \rPt{alg-snd-comp} above ensures that the algorithm has
      decided whether $ S_1 \subt S_2 $,
    \item \rthm{full-abstraction} now implies that if $ S_1 \subt S_2$
      then $\sigma_1 \Peerleq \sigma_2$, and if $ S_1 \not \subt S_2$
      then $\sigma_1 \not \Peerleq \sigma_2$.\qed
\end{enumerate}

The second immediate consequence of \rthm{full-abstraction} is an
explanation of type equivalence.
{\em Type equivalence}, denoted $\typeEQ$, is the equivalence
generated by the subtyping, so that
\begin{equation}
\label{eq:typeEQ}
T \typeEQ S \text{ whenever }  T \subt S \text{ and } S \subt T
\end{equation}
The explanation of $\typeEQ$ is alternative to the standard 
one based on tree models of types \cite{DBLP:journals/fuin/BrandtH98}.

\begin{prop}\name{\Fabs type equivalence}\\
For every $T, S \in {\st}$, $S \typeEQ T$ if and only if $\M{S} \Peereq \M{T}$.
\end{prop}
\begin{proof}
  A direct consequence of the definitions of $\typeEQ$, $\Peereq$ and
  of \rthm{full-abstraction}.
\end{proof}

\section{Complements of Contracts}
\label{sec:complement}

\begin{figure}
\hrulefill
$$
\begin{array}{c}
\stdual{\e} = \e, \quad \stdual{X} = X, \quad 
\stdual{\Rec[X]{S}} = \Rec[\stdual{X}]{\stdual{S}}
\\[.5em]
\stdual{\sinp{M}S} = \sout{M}\stdual{S}, \quad 
\stdual{\sout{M}S} =
\sinp{M}\stdual{S}
\\[.5em]
\stdual{\branch{ \gt{l}_1 \as S_1, \ldots, \gt{l}_n\as S_n}} =
\choice{\gt{l}_1 \as \stdual{S_1}, \ldots, \gt{l}_n\as
  \stdual{S_n}}
\\[.5em]
\stdual{\choice{ \gt{l}_1 \as S_1, \ldots, \gt{l}_n\as S_n}} =
\branch{ \gt{l}_1 \as \stdual{S}_1, \ldots, \gt{l}_n\as \stdual{S_n} }
\end{array}
$$
\caption{Standard syntactic definition of dual
  session types \label{fig:std-dual}}
\hrulefill
\end{figure}

The converse to \rthm{sem-complete} relies on the existence
for every session contract~$\sigma$ of a ``complementary'' session
contract $\prdual{\sigma}$ that is in $\B$-peer compliance with~$\sigma$,
at least for~$\B$s that satisfy certain minimal conditions.
The well-known {\em syntactic duality} of session types,
discussed in the Introduction and 
defined inductively in \rfig{std-dual}, is an
obvious candidate. It is defined for the language of session contracts
$L_\sc$ 
by structural induction in \rfig{std-dual}; we are primarily
interested in it as applied to session contracts $\sc$.
Intuitively to obtain the dual of a contact $\sigma$, denoted
$\stdual{\sigma}$,
\begin{itemize}
\item every internal choice is transformed into an external choice

\item every external choice is transformed into an internal choice

\item inputs are turned into outputs, and outputs into inputs. 
\end{itemize}
But it should be emphasised that in the transformation from
$\sigma$ to $\stdual{\sigma}$ all messages are left unchanged. 

Unfortunately, as we will see in the following example, this standard
duality transformation does not satisfy our requirements for
complementary contracts. First a definition.

\begin{defi}
\label{def:reasonable}
We say that $\B$ is {\em reasonable} whenever $\sigma_1 \B \sigma_2$ implies 
\begin{enumerate}[label=\roman*)]
\item $\unfold{\sigma_1} \B \unfold{\sigma_2}$
\item if $\sigma_1 \ar{\lambda_1}$ and $\sigma_2 \ar{\lambda_2}$ then $\lambda_1$ and
$\lambda_2$ are {\em both} input actions or output actions.
\item\label{pt:reasonable-continuations} %
  if $\sigma_1 \ar{ ?( \sigma^m_1 ) } \sigma'_1$ and
  $\sigma_2 \ar{ ?( \sigma^m_2 ) } \sigma'_2$ then 
$\sigma^m_1 \B \sigma^m_2$
and 
$\sigma'_1 \B \sigma'_2$
\qed
\end{enumerate}
\end{defi}

\noindent
For instance, if   $ !(\Unit).\Unit \B {?(\Unit).\Unit}$ for some $\B$,
then $\B$ is not reasonable. 

The family of reasonable relations is not arbitrary.
\rthm{sem-complete} implies that the $\peerleq[\B]$ for
every preorder $\B$ is reasonable, thus reasonable relations are an
over-approximation of the behavioural preorders that we are concerned
with.

\begin{figure}
  \hrulefill
  \begin{center}
    \begin{tabular}{c}
      \begin{tikzpicture}
        \node[state] (r) at (0,0) {$\sigma$};
        \node[state] (r1) at (2,0) {$ ?( \sigma ).\Unit $};
        \node[state] (r2) at (4,0) {$\Unit$};
        \node[state] (r3) at (6,0) {$\Cnil$};

\path[->]
(r) edge node [above] {$\tau$} (r1)
(r1) edge node [above] {$?( \sigma )$} (r2)
(r2) edge node [above] {$\ok$} (r3);
      \end{tikzpicture}
      \\[1em]
      \begin{tikzpicture}
        \node[state] (s) at (0,0) {$\stdual{\sigma}$};
        \node[state] (s1) at (2,0) {$ !( \stdual{\sigma} ).\Unit $};
        \node[state] (s2) at (4,0) {$\Unit$};
        \node[state] (s3) at (6,0) {$\Cnil$};
\path[->]
(s) edge node [below] {$\tau$} (s1)
(s1) edge node [below] {$!( \stdual{\sigma} )$} (s2)
(s2) edge node [below] {$\ok$} (s3);
      \end{tikzpicture}
    \end{tabular}
  \end{center}
If $\B$ is reasonable, then $ !( \stdual{\sigma} )  \not \bowtie_{\B}?( \sigma )$, and so $\sigma$ and $\stdual{\sigma}$ are not in $\B$-mutual compliance.
  \caption{The behaviours of the contracts $\rho$ and $\sigma$ of \rexa{issues-stdual}}
  \label{fig:lts-comp-wrong}
  \hrulefill
\end{figure}

\begin{exa}
\label{exa:issues-stdual}
In general it is not true that 
$\sigma$ complies with its dual $\stdual{\sigma}$; if 
$\B$ is reasonable then we can find a contract $\sigma$ such that   
$\sigma \Npeercmpl[\B] \stdual{\sigma}$.

For example take $\sigma$ to be  $\Rec[x]{?(x).\Unit}$; here
$\stdual{\sigma}$ is $ \Rec[x]{!(x).\Unit}$. The behaviour of these
contracts is depicted in \rfig{lts-comp-wrong}.
Since $\unfold{\sigma}$ performs inputs, while
$\unfold{\stdual{\sigma}}$ performs outputs, 
and $\B$ is reasonable, one sees that 
$( \unfold{\stdual{\sigma}}, \unfold{\sigma} ) \not \in {\B}  $, and in turn 
$(\stdual{\sigma}, \sigma) \not \in {\B} $.
This implies that $ !( \stdual{\sigma} ) \not \bowtie_{\B} ?( \sigma
)$, and so $ \sigma \Par \stdual{\sigma} \ar[\B]{\tau} \unfold{\sigma} \Par \unfold{\stdual{\sigma}} \nar[\B]{\tau}$.
But this means that $ \sigma \Npeercmpl[\B] \stdual{ \sigma }$ because 
$\unfold{\sigma}$ does not perform  $\ok$.
\end{exa}

\subsection{Message-closed contracts}

In this  section we give a restriction on contracts which ensures that they do indeed comply with 
their duals. The essential idea is that terms used as messages should have no free occurrences of 
recursion variables. 

\begin{defi}\name{Message-closed}
\label{def:Message-closed}\\
For any $\sigma \in L_\sc$ we say that it is \emph{message-closed}, or \emph{m-closed}, whenever
a sub-term of the form $?(\sigma^m).\sigma'$ or  $!(\sigma^m).\sigma'$ occurs in the main body of
$\sigma$ then $\sigma^m$ is 
a closed contract, that is, is in $\sc$. More formally:
\begin{enumerate} 
\item The terms $\Unit$ and $x$ are \emph{m-closed}

\item The term $\Rec[x]{\sigma'}$ is \emph{m-closed} if $\sigma'$ is \emph{m-closed}.

\item The terms $!(\sigma^m).\sigma'$ and $?(\sigma^m).\sigma'$ are \emph{m-closed} if
      $\sigma'$ is \emph{m-closed} and $\sigma^m$ is closed. 

\item The terms $ \extsum_{i \in I}?\gt{l}_i.\sigma_i$ and 
        $ \intsum_{i \in I}!\gt{l}_i.\sigma_i$ are \emph{m-closed} if 
     all $\sigma_i$ are \emph{m-closed}. 
\end{enumerate}
\end{defi}
\noindent
It is important to note the \emph{m-closed} is quite a strong condition; if $\sigma$ is
closed then it does not automatically follow that it is \emph{m-closed}. As a counterexample
we can take the contract used in Example~\ref{exa:issues-stdual},
$\Rec[x]{?(x).\Unit}$.

The crucial property of \emph{m-closed} terms is that the dual function $\stdual{ -  }$ 
is preserved by substitutions. This is expressed in the following lemma where we use 
$\stdual{s}$ to denote the substitution which maps each variable $X$ to $\stdual{s(X)}$. 
\begin{lem}
\label{lem:subs-stdual}\qquad
  \begin{enumerate}[label=(\roman*)]
  \item 
  Suppose $\sigma \in L_\sc$ is   \emph{m-closed}.
Then 
\begin{enumerate}[label=\({\alph*}]
\item \label{pt:stdual-distributes}
$\stdual{(\sigma \subsS) } = (\stdual{\sigma}) \stdual{\subsS}$

\item \label{pt:stdual-preserves-mclosed}
$\stdual{\sigma}$ is \emph{m-closed}

\item \label{pt:application-subs}
if $\subsS(x)$ is \emph{m-closed} for every $x \in \dom{\subsS}$ then
$\sigma \subsS$ is also \emph{m-closed}. 
\end{enumerate}

\item \label{pt:lts-moves-preserve-mclosed}
For every \emph{m-closed} $\sigma \in \sc$ and $\mu \in 
\Acttt$, $\sigma \ar{\mu} \sigma'$ implies $\sigma'$ is also \emph{m-closed}. 
  \end{enumerate}
\end{lem}
\begin{proof}
  Part (i) is proved by structural induction on $\sigma$. Part (ii) uses rule induction on 
the judgements  $\sigma \ar{\mu} \sigma'$; the case when $\sigma$ has the form
$\Rec[x]{\sigma_1}$ relies on Part (i) (c). 
\end{proof}
\noindent
The requirement that $\sigma$ be \emph{m-closed} in
Lemma~\ref{lem:subs-stdual} is essential. Again the contract 
$\sigma = \Rec[x]{?(x).\Unit}$,
used
in Example~\ref{exa:issues-stdual}, provides a counterexample. 
Let $\subsS$ be the substitution 
$\Lsubs{x}{\sigma}$ and let $\sigma'$ be the body of the
recursive definition, $?(x).\Unit$, which is not \emph{m-closed}. 
Then $\stdual{(\sigma' \subsS) }$ is the contract $!(   \Rec[x]{?(x).\Unit}   ).\Unit$ whereas,
since $\stdual{\sigma} = \Rec[x]{!(x).\Unit}$, 
$ (\stdual{\sigma'}) \stdual{\subsS}$ is the different contract 
$!(  \Rec[x]{!(x).\Unit}   ).\Unit$.

\begin{lem}
\label{lem:mclosed-impl-stdual-unfold-commute}
For every $\sigma \in \sc$, if $\sigma$ is m-closed then
$ \stdual{ \unfold{ \sigma}} = \unfold{\stdual{ \sigma }}$.
\end{lem}
\begin{proof}
Let $\unfold{ \sigma } = \rho$.
We have to show that 
\begin{align}\label{eq:du}
  \stdual{\rho} = \unfold{\stdual{\sigma}}
\end{align}

\noindent We reason by rule induction on the derivation of the 
judgement $\unfold{\sigma} = \rho$.
The base case is when $\sigma$ is not a recursive term, in which 
case $\rho$ coincides with $\sigma$ itself. Examining Definition~\ref{def:dual}
it is obvious that if $\stdual{\sigma}$ is a recursive term then so is
$\sigma$; from this we conclude that  $\unfold{\stdual{\sigma}}$
is $\stdual{\sigma}$ itself, from which the required (\ref{eq:du}) is
trivially true.

The inductive case is when $\sigma$ has the form $\Rec[x]{\sigma'}$
and $\unfold{\sigma} = \rho$ because 
$$\unfold{\sigma'\Lsubs{ x }{\sigma}}  = \rho $$
The hypothesis that $\sigma$ is m-closed, and \rdef{Message-closed}
imply that also $\sigma'$ is m-closed, thus \rpt{application-subs} 
of \rlem{subs-stdual} lets us prove that  $\sigma'\Lsubs{ x }{\sigma}$
is m-closed. Thanks to this property of $\sigma'\Lsubs{ x }{\sigma}$
we apply the inductive hypothesis to prove that
$\stdual{\rho} = \unfold{\stdual{\sigma'\Lsubs{ x}{\sigma}}}$. So the
required  (\ref{eq:du}) will follow if we show 
  $\unfold{\stdual{\sigma}} = 
\unfold{\stdual{\sigma'\Lsubs{x}{\sigma}}}$.
Since $\sigma$ and $\sigma'$ are m-closed, \rpt{application-subs} of
\rlem{subs-stdual} ensures that
\begin{equation}
  \label{eq:commutativity-crux}
  \stdual{\sigma' \Lsubs{x}{\sigma}} =
  \stdual{ \sigma' } \Lsubs{x}{ \stdual{\sigma}}
\end{equation}
The equality we are after is now easy to prove,
$$
\begin{array}{llll}
\unfold{ \stdual{ \sigma }} & = & \unfold{  \Rec[x]{ \stdual{\sigma'}
  }} & \text{Because } \stdual{\sigma} = \Rec[x]{ \stdual{\sigma'}}\\
& = & 
\unfold{ \stdual{ \sigma'}
  \Lsubs{x }{ \Rec[x]{\stdual{ \sigma' }}} } &\text{By definition of
}\unfoldSym\\
& = & \unfold{ \stdual{ \sigma'}
  \Lsubs{x }{ \stdual{ \sigma }} } & \text{Because }  
\stdual{ \sigma } = \Rec[x]{\stdual{ \sigma' }}\\
& = & \unfold{ \stdual{ \sigma' \Lsubs{x}{\sigma} } } &\text{Because of } (\ref{eq:commutativity-crux})
\end{array}
$$
\end{proof}
\noindent
This last result implies the main property
of $\stdual{-}$ over {\em m-closed} contracts:
given a m-closed peer $\rho$, its dual $\stdual{\rho}$
has indeed a complementary behaviour, that is $\stdual{\rho}$
is in compliance with $\rho$, with respect to any preorder $\B$.

\begin{lem}
\label{lem:stduals-interact-or-happy}
  Suppose $\B$ is a reflexive relation.
  For every session contract $\rho$,
  we have that either $\rho \ar{ \ok }$
  or
  $ \rho \Par \stdual{\rho} \ar[\B]{\tau}$.
\end{lem}
\begin{proof}
  To prove the lemma we assume that $\rho \nar{\ok}$,
  and show that $ \rho \Par \stdual{\rho} \ar[\B]{\tau}$.
  Either $\rho \ar{\tau}$ or $\rho$ is stable.
  In the first case we apply \rname{I-Left} to derive
  the desired $ \rho \Par \stdual{\rho} \ar[\B]{\tau}$.
  Suppose now that $\rho$ is stable. The argument proceeds by case
  analysis on the shape of $\rho$, which, in view of our
  assumption, cannot be $\Unit$ and has no top-most recursion.
  We discuss only three cases.
  If $ \rho = {!( \sigma ).\rho'}$,
  then $\stdual{ \rho } = {?( \sigma ).\stdual{ \rho' }}$.
  Since by hypothesis $\B$ is reflexive, $ !(\sigma) \bowtie_{\B} ?(\sigma)$, 
  thus we apply \rname{I-Synch}
  to derive the required $ \rho \Par \stdual{\rho} \ar[\B]{\tau} $.
\leaveout{  If $ \rho = \#( \sigma ).\Unit$ the argument is analogous to the
  previous one.}
  Suppose now that 
$\rho = \intsum_{i \in I}
  ?\gt{l}_i . \rho_i$. Since $\rho$ is stable, it must be the case
  that $| I | = 1$, that is $I = \sset{ k }$ for some $k$.
  We know by definition that $\stdual{ \rho } = ?\gt{l}_k . \stdual{\rho_k} $,
  and that $!\gt{l}_k \bowtie_{\B} ?\gt{l}_k$, thus we apply rule
  \rname{I-Synch}
  to infer the hand-shake  $\rho \Par \stdual{\rho} \ar[\B]{\tau}$.
\end{proof}

\begin{thm}
\label{thm:standard-duality-comply}
Suppose $\B$ is a preorder. 
Then $\rho \peercmpl[\B] \stdual{\rho}$ for every 
\emph{m-closed} session contract $\rho$.
\end{thm}
\begin{proof}
Let 
$$
    {\R} = \setof{ (\rho_1, \rho_2) }
    { \rho \arstar{\tau} \rho_1,\ \stdual{\rho} \arstar{\tau} \rho_2\, \text{for every  \emph{m-closed} } \rho}
$$
It is sufficient to show the set inclusion $ {\R} \subseteq \FunCpeer( \R, \B )$ for any preorder $\B$.
Pick a pair of contracts $(\rho_1, \, \rho_2)$ in the relation $\R$. The construction of $\R$ ensures that
there exists a contract $\rho$ which is m-closed, and such that 
\begin{equation}\label{eq:tau-steps}
\rho \arstar{\tau} \rho_1, \qquad \stdual{\rho} \arstar{\tau} \rho_2
\end{equation}

\rdef{peercmpl} requires us to prove two properties of the pair $(\rho_1, \, \rho_2)$,
depending on whether the composition $ \rho_1 \Par \rho_2 $ is stable or not.
First assume that 
\begin{equation}
  \label{eq:stable}
  \rho_1 \Par \rho_2 \nar[\B]{\tau} 
\end{equation}
We have to show that
$\rho_1 \ar{ \ok }$ and that $\rho_2 \ar{ \ok }$.
The assumption $\rho_1 \Par \rho_2 \nar[\B]{\tau}$ ensures that
$\rho_1 \nar{\tau}$; a property of the unwind function, given in 
Lemma~\ref{lemma:unfold} of the Appendix, then ensures that
$$\unfold{\rho} \arstar{\tau} \rho_1$$
Similar reasoning 
establishes 
$\unfold{\stdual{\rho}} \arstar{\tau} \rho_2$. 
In fact by Lemma~\ref{eq:commutativity-crux} the latter may be rewritten 
as 
$\stdual{\unfold{\rho}} \arstar{\tau} \rho_2$. 
Introducing $\sigma$ to denote $\unfold{\rho}$ 
we have now established 
\begin{equation}
\label{eq:unfolds-reduce}
 \sigma \arstar{\tau} \rho_1 \;\,\text{and}\,\;
 \stdual{ \sigma } \arstar{\tau} \rho_2
\end{equation}

\noindent There is very little scope for performing $\tau$ transitions in 
(\ref{eq:tau-steps}) above, and even less in 
(\ref{eq:unfolds-reduce}), since $\sigma$ cannot be a 
recursive term of the form $\Rec[x]{\ldots}  $. 
In fact because of the assumption
(\ref{eq:stable}) we can show that
\begin{equation}\label{eq:notaus}
  \sigma = \rho_1 \;\,\text{and}\,\; \stdual{\sigma} = \rho_2 
\end{equation}
The proof of this fact proceeds by a case analysis on the
syntactic structure of $\sigma$. The only possibility for generating
a $\tau$ transition in (\ref{eq:unfolds-reduce}) is when either 
$\sigma$ or its dual $\stdual{\sigma}$ is an internal choice
$\setof{ !\gt{l}_i.\sigma_i }{ i \in I }$, where $|I| > 1$.
Without loss of generality suppose the former. 
Then  $\stdual{\sigma}$ is
$\extsum_{I} ?\gt{l}_i.\sigma_i $ and 
$\rho_1 \Par \rho_2$ must take the form 
$ {!\gt{l}_k.\sigma_k} \Par \extsum_{I} ?\gt{l}_i.\sigma_i$,
for some $k \in I$. 
An application of  \rname{I-Synch} gives a $\tau$ transition, 
$\rho_1 \Par \rho_2 \ar[\B]{\tau}$, thereby contradicting the
assumption (\ref{eq:stable}).

An application of \rlem{stduals-interact-or-happy}, together with the
just established (\ref{eq:notaus}), immediately gives  lets us show $ \rho_1 \ar{ \ok
}$,
from which $\rho_2 \ar{ \ok }$ also follows, since 
$\rho_2$ is $ \stdual{\sigma}$.

Now suppose that $ \rho_1 \Par \rho_2 \ar[\B]{ \tau } \rho'_1 \Par \rho'_2$,
we have to prove that  $\rho'_1 \R \rho'_2$. This
will follow if we exhibit a m-closed contract $\hrho$ such that
$\hrho \arstar{\tau} \rho'_1$,
and $\stdual{ \hrho } \arstar{\tau} \rho'_2$.
We reason by case
analysis on the rule of \rfig{interaction} used to infer the silent move
$\rho_1 \Par \rho_2 \ar[\B]{ \tau } \rho'_1 \Par \rho'_2$.

If \rname{I-Left} was applied, then $ \rho_1 \ar{\tau} \rho'_1 $
and $\rho'_2 = \rho_2$. The desired $\hrho$ is the contract $\rho$
itself. If rule \rname{I-Right} again the $\hrho$ we are after is $\rho$
itself.

The last case to discuss is when
$ \rho_1 \Par \rho_2 \ar[\B]{ \tau } \rho'_1 \Par \rho'_2$ is inferred
applying rule \rname{I-Synch}. In this case
$\rho_1 \ar{ \lambda_1 } \rho'_1$, $\rho_2 \ar{\lambda_2} \rho'_2$,
and $\lambda_1 \bowtie_{\B} \lambda_2 $.
We have already proven that the contract $\rho_1$ 
is m-closed, thus \rPt{lts-moves-preserve-mclosed} of
\rlem{subs-stdual} imply that $\rho'_1$ 
is m-closed. We pick as candidate $\hrho$ the contracts $\rho'_1$.
To finish the proof it suffices to show that $\rho'_2 = \stdual{
  \rho'_1 }$. 
We proceed by case analysis on $\lambda_1$, and there are four
cases to discuss, for $\lambda_1$ is either an input or an output,
\leaveout{or a $\#$ }action, and there are two 
subcases depending on the action being first-order or higher-order.
\leaveout{
If $\lambda_1$ is an input or an output there are
other two subcases on the actions being first-order or higher-order.
}
We discuss only two cases involving higher-order actions and one
involving first-order actions.

Suppose that $\lambda_1 = {?( \sigma_1 )}$. In view of the restrictive
syntax of contracts, it must be the case that $ \rho_1 = {?( \sigma
).\rho'_1 }$. Now $\rho \arstar{\tau} \rho'_1$ implies that 
$\unfold{ \rho } = {?( \sigma ).\rho'_1 }$.

Since $\lambda_1 \bowtie_{\B} \lambda_2$,  $\lambda_2 = {!( \sigma' )}$
for some $\sigma'$, and $\rho_2 = {?( \sigma' ).\rho'_2}$. The last
fact and $\stdual{ \rho } \arstar{ \tau } \rho_2$ imply that
$\unfold{ \stdual{\rho} } = \rho_2 $.
An application of \rlem{mclosed-impl-stdual-unfold-commute} lets
us obtain that $\rho_2 = \unfold{ \stdual{ \rho } } = {!( \sigma ).\stdual{
  \rho'_1 }}$. It follows that $ \rho'_2 = \stdual{ \rho'_1 }$,
as required.

If $\lambda_1 = {!( \sigma )}$, 
or $\lambda_1 = {!\gt{ t }}$, or $\lambda_1 = {?\gt{ t }}$,
the argument is analogous to the previous one.

Suppose now that $\lambda_1 = {! \gt{l}} $.
It must be the case
that $\unfold{ \rho } = \intsum_{i \in I} ! \gt{l}_i. \rho_i$, and for some
$k \in I$, $\rho_1 = {!\gt{l}_k.\rho_k}$.
The construction of $\R$ ensures that 
\begin{equation}
\label{eq:stdual-peercmpl-aux}
\stdual{ \rho } \arstar{\tau} \rho_2
\end{equation} and this implies that $\unfold{ \stdual{ \rho }}
\arstar{\tau} \rho_2$. By definition $\stdual{ \unfold{ \rho } } = 
\extsum_{i \in I} ? \gt{l}_i. \stdual{ \rho_i }$, thus
commutativity (\rlem{mclosed-impl-stdual-unfold-commute}),
ensures that $ \unfold{ \stdual{ \rho } } = 
\extsum_{i \in I} ? \gt{l}_i. \stdual{ \rho_i }$. Observe that
$ \unfold{ \stdual{ \rho } } \nar{ \tau}$, thus
\req{stdual-peercmpl-aux}
above implies that $ \rho_2 = \unfold{ \stdual{ \rho } }$, so
$\rho_2 = \extsum_{i \in I} ? \gt{l}_i. \stdual{ \rho_i }$.
As $\rho_2 \ar{ \lambda_2 } \rho'_2$ and $!\gt{ l }_k \bowtie_{\B}
\lambda_2$,
it follows that $\lambda_2 = {?\gt{l}_k}$, and $\rho'_2 = \stdual{
  \rho_k }$. The equality $\rho_k = \rho'_1$ now lets us conclude that
$ \rho'_2 = \stdual{ \rho'_1}$, which is the fact we were after.
\end{proof}

\subsection{Extension to arbitrary session contracts}

Our intention here is use \rthm{standard-duality-comply}
to find a complement for all session contracts. 
This is achieved in two steps. First for each $\sigma \in \sc$ we
construct
a \emph{behaviourally equivalent} contract $\mcl{\sigma}$ which is
\emph{m-closed}. 
The required complement of $\sigma$, which we will denote by
$\prdual{\sigma}$ will be taken to be the standard dual of 
$\mcl{\sigma}$.

\begin{defi}\name{M-closure}\label{def:dual}\\
For any $\sigma \in L_\sc$ 
and any $\subsS$ such that $\fv{\sigma} \subseteq \dom{\subsS}$ the
term
$\mclo{\sigma}{\subsS}$ is defined by structural induction as follows:
$$
\mclo{\sigma}{\subsS} =
\begin{cases}
  \Unit & \text{if } \sigma = \Unit,\\
  x & \text{if } \sigma = x,\\
  !\gt{t}.\mclo{\sigma'}{\subsS} & \text{if } \sigma = {!\gt{t}.\sigma'},\\
  ?\gt{t}.\mclo{\sigma'}{\subsS} & \text{if } \sigma = {?\gt{t}.\sigma'},\\
  !(\sigma^m\, \subsS).\mclo{\sigma'}{\subsS} & \text{if } \sigma = {!(\sigma^m).\sigma'},\\
  ?(\sigma^m\, \subsS).\mclo{\sigma'}{\subsS} & \text{if } \sigma = {?(\sigma^m).\sigma'},\\
%
%
  \extsum_{i \in I}!\gt{l}_i.\mclo{\sigma_i}{\subsS} & \text{if } \sigma =
  \extsum_{i \in I}?\gt{l}_i.\sigma_i,\\
  \intsum_{i \in I}?\gt{l}_i.\mclo{\sigma_i}{\subsS} & \text{if } \sigma =
  \intsum_{i \in I}!\gt{l}_i.\sigma_i,\\
  \Rec[x]{\mclo{\sigma'}{\subsS \cdot \Lsubs{x}{\sigma\subsS}}} & \text{if } \sigma = \Rec[x]{\sigma'}
\end{cases}
$$
Note that in the last clause the substitution 
$\subsS \cdot \Lsubs{x}{\sigma\subsS}$
is well-defined, as $\sigma\subsS$ is closed.

For $\sigma \in \sc $ we let $\mcl{\sigma}$, called the  \emph{m-closure} of
$\sigma$,  denote 
$\mclo{\sigma}{\varepsilon}$.
\end{defi}

The intuition behind $\mclo{\sigma}{\subsS}$ is that an m-closed term
equivalent to $\sigma$, can be constructed by (1) keeping track in
the accumulator $\subsS$ of all the substitutions that take place
unfolding $\sigma$, and by (2) applying these substitutions only
to the messages of $\sigma$, and not the continuations.

\begin{exa}
\label{exa:applications-mcl}
  In this example we apply the function
  $\mcl{-}$ to two session contracts.
  The first is the contract $\rho = \Rec[x]{?(x).\Unit}$ 
  that we already used in \rexa{issues-stdual}.
  By definition we have the equalities
  $$
  \begin{array}{lll}
    \mcl{ \rho }
    & = &   \mclo{ \Rec[x]{?(x).\Unit} }{ \varepsilon }\\
    & = &   \Rec[x]{ \mclo{?(x).\Unit}{ \Lsubs{ x }{ \rho } } }\\
    & = &   \Rec[x]{ ?(\rho ).\mclo{\Unit}{ \Lsubs{ x }{ \rho } }
    }\\
    & = &   \Rec[x]{ ?(\rho ). \Unit }\\
  \end{array}
  $$
  Since $\rho$ is closed, the contract $\mclo{ \rho }{ \varepsilon }$ is
  m-closed.

  Now we apply $\mcl{ - }$ to a more involved contract,
  namely $\sigma = \Rec[x]{\Rec[y]{ ?(y).x }}$.
  The following equalities are true by definition,
$$
\begin{array}{lll}
  \mcl{ \sigma }
  & = & \mclo{ \Rec[ x ]{ \Rec[ y ]{ ?( y ) .x } } }{ \varepsilon } \\
  & = &\Rec[ x ]{( \mclo{ \Rec[ y ]{ ?( y ). x } }{ \Lsubs{ x }{
        \sigma } } )}\\
  & = &\Rec[x]{\Rec[y]{ ( \mclo{  ?( y ).x }{ \subsS } )}}\\
  & = &\Rec[x]{\Rec[y]{ ?( \Rec[ y ]{ ?(y).\sigma } ). \mclo{ x }{ \subsS } }}\\
  & = &\Rec[x]{\Rec[y]{ ?( \Rec[ y ]{ ?(y).\sigma } ). x }}\\
  & = &\Rec[x]{\Rec[y]{ ?( \Rec[ y ]{ ?( y ).\sigma } ). x }}
\end{array}
$$
where $\subsS = \LLsubs{x}{ \sigma}{y }{ \Rec[ y ]{ ?(y).\sigma } }$.
Since $\sigma$ is closed also the contract $ \Rec[ y ]{ ?( y ).\sigma}$ is closed, thus the contract 
$\mclo{ \sigma }{\varepsilon}$ is m-closed.
\end{exa}
\leaveout{
 \leaveout{ 
  Recall the type $T$ of \rexa{no-commutativity},
  \begin{align}
    \label{eq:ex-badcomp}
    \badcomp{ T } & =   
    \Rec[X]{\Rec[Y]{ \sout{ \Rec[Y]{ \sinp{Y}{ X } } }{ X }}}
    \\[.5em]
    \label{eq:ex-comp}
    \comp{ T } & = 
    \Rec[X]{\Rec[Y]{ \sout{ \Rec[Y]{ \sinp{Y}{T} } }{ X }}}
  \end{align}
We have shown (\ref{eq:ex-badcomp}) in \rexa{no-commutativity},
while we obtain the equality in (\ref{eq:ex-comp}) as follows,
}
The difference between $\badcomp{ T }$ and $\comp{ T }$
lies in the messages of these types.
The message of $\badcomp{ T }$ is  $ \Rec[Y]{ \sinp{Y}{ X
  } } $, and it contains a free occurrence of $X$, whereas the message of
$\comp{ T }$ is $\Rec[Y]{ \sinp{Y}{T} }$, and it is a closed term.
This means that the term $\unfold{ \comp{ T } }$ contains the same messages
of $\comp{ T }$, and this is the reason why the 
functions $\unfoldS$ and $\compS$ commute.
}

\begin{lem}\label{lem:m-closed}
  Suppose $\fv{\sigma} \subseteq \dom{\subsS} $. Then 
  $\mclo{\sigma}{\subsS}$ is \emph{m-closed}.
\end{lem}
\begin{proof}
  By structural induction on $\sigma$. 

  First suppose $\sigma$ has the form $\Rec[x]{\sigma'}$.
  By definition 
    $\mclo{\sigma}{\subsS} = \Rec[x]{\mclo{\sigma'}{ \subsS \cdot \Lsubs{x}{\sigma \subsS}}}$.
    Since $\dom{ \subsS} \subseteq \dom{ \subsS \cdot \Lsubs{x}{\sigma \subsS} }$,
    also the inclusion
    $\fv{ \sigma' } \subseteq \dom{  \subsS \cdot \Lsubs{x}{\sigma \subsS} }$ is true,
    and we apply induction to conclude that 
    $\mclo{ \sigma' }{ \subsS \cdot \Lsubs{x}{\sigma \subsS}}$ is \emph{m-closed};
    by definition this means $\sigma$ is \emph{m-closed}. 

\leaveout{
By definition 
$\mclo{\sigma}{\subsS} = \Rec[x]{\mclo{\sigma'}{ \subsS \cdot \Lsubs{x}{\sigma}}}$.
Since $\fv{\subsS \cdot \Lsubs{x}{\sigma}} \subseteq \sigma'$ we can
apply induction to conclude that 
$\mclo{\sigma'}{ \subsS \cdot \Lsubs{x}{\sigma}}$ is \emph{m-closed};
by definition this means $\sigma$ is \emph{m-closed}. 
}

As another case suppose $\sigma$ has the form 
$!(\sigma^m). \sigma'$. 
Here $\mclo{\sigma}{\subsS}$ is 
$!(\sigma^m \subsS). \mclo{\sigma'}{\subsS}$. 
Induction gives that $\mclo{\sigma'}{\subsS}$ is \emph{m-closed}, and 
since $\fv{\sigma^m} \subseteq \fv{\subsS} $ we know 
$\sigma^m \subsS$ is closed; this means that by definition
$\sigma$ is \emph{m-closed}. 

All remaining cases are either similar, or trivial. 
\end{proof}

Because of \leaveout{this Lemma} \rlem{m-closed}
we know from
Theorem~\ref{thm:standard-duality-comply} that $\stdual{\mcl{\rho}}$
complies with $\mcl{\rho}$ for every session contract $\rho$. We now
show that that $\mcl{\rho}$ and $\rho$ are \leaveout{essentially} behaviourally
equivalent, in that they comply with exactly the same contracts.  This
involves first establishing a sequence of technical lemmas.

\begin{lem}\label{lem:closed-mclo}
For every $\sigma \in L_{\sc}$, and substitutions $\subsS_1, \subsS_2$,
if
$\subsS_1(x) = \subsS_2(x)$ for every $x \in \fv{\sigma}$,
then $\mclo{\sigma}{\subsS_1} = \mclo{\sigma}{\subsS_2}$, whenever both
are defined.
\leaveout{
  Suppose $\subsS_1(x) = \subsS_2(y)$ for every $x \in \fv{\sigma}$.
Then $\mclo{\sigma}{\subsS_1} = \mclo{\sigma}{\subsS_2}$, whenever both
are defined.
}
\end{lem}
\begin{proof}
  Straightforward by structural induction on $\sigma$. 
\end{proof}

Given a substitution $\subsS = [ x_1 \mapsto \sigma_1, \ldots, x_n \mapsto \sigma_n ]$,
we let $\mcl{ \subsS } = [ x_1 \mapsto \mcl{\sigma_1}, \ldots, x_n \mapsto \mcl{\sigma_n} ]$.
\begin{prop}\label{prop:mclo-subs}
Let $\sigma \in L_{\sc}$, and suppose that 
$\fv{\sigma} \subseteq \dom{\subsS_1} \cup \dom{\subsS_2} $. Then 
$\mclo{\sigma \subsS_1}{\subsS_2} = \mclo{\sigma}{ \subsS_2 \cdot \subsS_1} \, \mcl{\subsS_1}$.
\end{prop}
\leaveout{
\MHc{\MHf{To be omitted}MH: Explanation:

\noindent
LHS: First the substitution $\subsS_1$ is applied to $\sigma$; there
may still be free variables in $\sigma \subsS_1$, those in
$\dom{\subsS_2}$. 
 Then $\mclo{-}{\subsS_2}$ is applied.

\noindent
RHS: Here $\mcl{\subsS_1}$ is a substitution. Its domain is the same
as that of $\subsS_1$ and it maps each $x$ in the domain to
$\mcl{\subsS_1(x)}$.
So first $\mclo{-}{ \subsS_2 \cdot \subsS_1}$ is applied to
$\sigma$. There may still be free variables in the resulting term. 
To obtain the LHS we need then to apply the substitution 
$\Lsubs{x}{\mcl{\subsS_1(x)}}$, for each $x \in \dom{\subsS_1}$. 

}}
\begin{proof}
  By structural induction on $\sigma$. We examine the three most interesting cases. 

\begin{enumerate}

\item 
First suppose $\sigma$ is a variable $z$. The hypothesis implies that
$z$ is either in $\dom{\subsS_1}$ or $\dom{\subsS_2}$. 
Suppose that $z \in \dom{\subsS_1}$. The
left hand side is by definition $\mclo{\subsS_1(z)}{ \subsS_2}$ but because 
$\subsS_1(z)$ is always a closed term, by Lemma~\ref{lem:closed-mclo}, this is the same as 
 $\mclo{\subsS_1(z)}{ \varepsilon}$, that is
$\mcl{\subsS_1(z)}$.
This is precisely the  right hand side:
by definition $\mclo{z}{  \subsS_2 \cdot \subsS_1} = z$, so applying the
substitution $\mcl{\subsS_1}$ to $z$ we have 
$\mcl{\subsS_1(z)}$.

On the other hand if $z \in \dom{\subsS_2}$ and $z \not\in \dom{\subsS_1}$ then both 
sides evaluate to the term $z$.

\item 
Suppose $\sigma$ has the form $\Rec[x]{\sigma'}$.
Here 
$\mclo{\sigma}{ \subsS_2 \cdot \subsS_1} = 
\Rec[x]{\mclo{\sigma'}
            { \subsS_2 \cdot \subsS_1 \cdot \Lsubs{x}{\sigma (\subsS_2
              \cdot \subsS_1) }}
         }$.
So the right hand side is equal to
$$
(\Rec[x]{ ( \mclo{\sigma'}{ \subsS_2 \cdot \subsS_1 \cdot
      \Lsubs{x}{\sigma (\subsS_2 \cdot \subsS_1)}}) }\, \,\mcl{\subsS_1}
$$
\leaveout{
\begin{align*}
  &(\Rec[x]{ ( \mclo{\sigma'}{ \subsS_2 \cdot \subsS_1 \cdot
      \Lsubs{x}{\sigma (\subsS_2 \cdot \subsS_1)}}) }\, \,\mcl{\subsS_1}
\end{align*}
}
Applying the definition of substitution we get 
\begin{align}\label{eq:rhs}
  &\Rec[x]{ (\mclo{\sigma'}
                { \subsS_2 \cdot \subsS_1 \cdot \Lsubs{x}{\sigma \subsS_2
              \cdot \subsS_1}}  \mcl{\SnotX_1 })}
\end{align}
\leaveout{
\MHc{\MHf{To be omitted}MH Explanation:
Recall $\mcl{\subsS_1}$ is the substitution which maps $y$ to 
$\mcl{\subsS_1(y)}$ for each $y \in \dom{\subsS_1}$.
So $ \mcl{\SnotX_1 }$ has as domain $\dom{\subsS_1} - \sset{x}$ and
maps each $y$ in the domain to $\mcl{\subsS_1(y)}$. 
}
}
We show that the left hand side can be rewritten in this form also. 
First we apply the substitution $\subsS_1$ to $\sigma$ and get
\begin{align*}
  \sigma\,\subsS_1 &= \Rec[x]{(\sigma'\,\SnotX_1)}
\end{align*}
Then applying the definition of $\mclo{ - }{ - }$ we have 
\begin{align}\label{eq:astep}
\mclo{\sigma \subsS_1}{\subsS_2} &= 
\Rec[x]{(\,\mclo{\sigma'\SnotX_1}{ \subsS_2 \cdot \Lsubs{x}{ (\sigma \subsS_1)
      \subsS_2 }}\,)}
\notag
\\
&=
\Rec[x]{(\,\mclo{\sigma'\SnotX_1}{ \subsS_2 \cdot \Lsubs{x}{\sigma
      (\subsS_2 \cdot \subsS_1)}}\,)}
\end{align}
Since $\fv{\sigma'} \subseteq \dom{\SnotX_1} \cup 
\dom{ \subsS_2 \cdot \Lsubs{x}{\sigma (\subsS_2
              \cdot \subsS_1)} }$
induction lets us infer
\begin{align*}
\mclo{\sigma'\SnotX_1}{ \subsS_2 \cdot \Lsubs{x}{\sigma (\subsS_2
              \cdot \subsS_1)}} =
\mclo{\sigma'} { \subsS_2 \cdot \Lsubs{x}{\sigma (\subsS_2 \cdot \subsS_1)}  \cdot \SnotX_1}\,
\mcl{\SnotX_1}
\end{align*}
%

The substitution $\Lsubs{x}{\rho}  \cdot \SnotX_1$
can also be written as 
$\subsS_1 \cdot \Lsubs{x}{\rho}$, for any $\rho$, 
thus
$$
\subsS_2 \cdot \Lsubs{x}{\sigma (\subsS_2 \cdot \subsS_1)}  \cdot \SnotX_1
= 
\subsS_2 \cdot \subsS_1 \cdot \Lsubs{x}{\sigma (\subsS_2 \cdot \subsS_1)}
$$ 
which means that rewriting (\ref{eq:astep}) above, 
we get
\begin{align*}
\mclo{\sigma \subsS_1}{\subsS_2} =
\Rec[x]{( \mclo{\sigma'}
                { \subsS_2 \cdot \subsS_1 \cdot \Lsubs{x}{\sigma (\subsS_2
              \cdot \subsS_1)}}  \mcl{\SnotX_1 } ) }
\end{align*}
Thus the left hand side
coincides with (\ref{eq:rhs}) above, as required.

\item 
Suppose $\sigma$ has the form 
$?(\sigma^m).\sigma'$. 
By definition $\mclo{\sigma}{\subsS_2 \cdot \subsS_1} = ?(\sigma^m\, (\subsS_2 \cdot \subsS_1)). \mclo{\sigma'}{\subsS_2 \cdot \subsS_1}$
and we know by \rlem{m-closed} that 
$\sigma^m\, (\subsS_2 \cdot \subsS_1)$ is a closed term. 
So the right hand side can be written as
$$
?(\sigma^m\, (\subsS_2 \cdot \subsS_1)). (\mclo{\sigma'}{\subsS_2
  \cdot \subsS_1}
            \,\mcl{\subsS_1})
$$
Here induction can be applied to the residual, and therefore  the right hand side now looks like
 $$
 ?(\sigma^m\, (\subsS_2 \cdot \subsS_1)).\, \mclo{\sigma'\subsS_1}{\subsS_2}
$$
But this is exactly $\mclo{\sigma\subsS_1}{\subsS_2}$, that is the left hand side. 
\qedhere
\end{enumerate}
\end{proof}
\noindent
We will use a specific instance of this Proposition, captured in the following Corollary:
\begin{cor}
 \label{cor:mclo-subs}
\leaveout{
 For $\sigma_1 \in L$ satisfying $\fv{\sigma_1} \subseteq \sset{x}$,
 and $\sigma_2 \in \sc$,
$\mcl{\sigma_1\subst{\sigma_2}{x}} =  \mclo{\sigma_1}{ \Lsubs{x}{\sigma_2}}\, \subst{\mcl{\sigma_2}}{x}$.
}
For every $\sigma_2 \in \sc$, and $\sigma_1 \in L_{\sc}$,
if $\fv{\sigma_1} \subseteq \sset{x}$, then
$$\mcl{\sigma_1\Lsubs{x}{\sigma_2}} =
\mclo{\sigma_1}{ \Lsubs{x}{\sigma_2}}\, \Lsubs{x}{\mcl{\sigma_2}}$$
\end{cor}
\begin{proof}
  An immediate application of  Proposition~\ref{prop:mclo-subs}, with $\subsS_2$ instantiated to the empty
  substitution, and $\subsS_1$ the singleton substitution 
  $\Lsubs{x}{\sigma_2}$. 
\end{proof}

\begin{prop}\label{prop:moves}
  Let $\sigma_1 \in \sc$.
  \begin{enumerate}
  \item   $\sigma_1 \ar{\ok}$ if and only if $\mcl{\sigma_1} \ar{\ok}$, moreover
  \item  for every $\mu \in \Actt$, 
  \begin{enumerate}
  \item $\sigma_1 \ar{\mu} \sigma_2$ implies $\mcl{\sigma_1} \ar{\mu} \mcl{\sigma_2}$
  \item conversely, $\mcl{\sigma_1} \ar{\mu} \sigma'$ implies $\sigma' = \mcl{\sigma_2}$  for
    some $\sigma_2$ satisfying $\sigma_1 \ar{\mu} \sigma_2$
  \end{enumerate}
\end{enumerate}
\end{prop}
\begin{proof}
  Suppose that $\sigma_1 \ar{\ok}$, the syntax of session contracts implies that
  $\sigma_1 = \Unit$, thus by definition $\mcl{\sigma_1} = \Unit$, hence $\mcl{ \sigma_1 } \ar{\ok}$. Conversely, if $ \mcl{\sigma_1}  \ar{\ok}$ then the syntax ensures that 
$\mcl{ \sigma_1 } = \Unit$, thus the definition of $\mcl{ - }$ implies that $\sigma_1 = \Unit$,
and plainly $\sigma_1 \ar{ \ok }$.

  Now let $\mu \in \Actt$.
  The proofs of both (1) and (2) are by structural induction on $\sigma_1$. 
  We look briefly at (2).
  Observe that $\mcl{ \sigma_1 } \ar{ \mu }$ ensures that $\sigma_1 \neq \Unit$.
  \begin{enumerate}[label=\({\alph*}]
    \item\label{pt:ho-out-proof} Suppose $\sigma_1 = {!( \sigma^m ).\sigma'_1}$.
      By definition $\mcl{ \sigma_1 } = {!( \sigma^m).\mcl{ \sigma'_1 }}$,
      hence the contract $\mcl{\sigma_1}$ performs only the action
      $ \mcl{\sigma_1} \ar{ !( \sigma^m ) } \sigma'$ with $\sigma' = \mcl{ \sigma'_1 } $.
      By letting $\sigma_2 = \sigma'_1$ we obtain immediately $\sigma' = \mcl{ \sigma_2}$,
      and $\sigma_1 \ar{ !( \sigma^m )  } \sigma_2$.
\leaveout{
    \item Suppose $\sigma_1 = \extsum_{i \in I} ?\gt{l_i}.\sigma^i_1$. 
      By definition $\mcl{ \sigma_1 } = \extsum_{i \in I} ?\gt{l_i}.\mcl{ \sigma^i_1 }$,
      thus $\mcl{ \sigma_1} \ar{ \mu } \sigma'$ implies that
      $\sigma' = \mcl{ \sigma^k_1 }$ and $\mu = {?\gt{l}_k}$ for some $k \in I$.
      Let $\sigma_2 = \sigma^k_1 $, plainly $\sigma' = \mcl{ \sigma_2 }$
      and $\sigma_1 \ar{ ?\gt{l_k}  } \sigma_2$.
}
    \item 
      Suppose $\sigma_1$ is $\Rec[x]{\sigma'_1}$. 
      Here $\mcl{\sigma_1}$ by definition is $\Rec[x]{\mclo{\sigma'_1}{\Lsubs{x}{\sigma_1}}}$. 
      So the only  possible move   
      $ \mcl{\sigma_1}  \ar{\mu} \sigma'$ is with $\mu = \tau$ and 
      $\sigma'$
      of the form 
      $\mclo{\sigma'_1}{\Lsubs{x}{\sigma_1}}\, \Lsubs{x}{\mcl{\sigma_1}}$.
      By \rcor{mclo-subs} this is the same as 
      $\mcl{\sigma'_1  \Lsubs{x}{\sigma_1}}$. 
      Setting $\sigma_2$ to be $\sigma'_1\, \Lsubs{x}{\sigma_1}$
      the result follows, because
      $\sigma_1 \ar{\tau} (\sigma'_1\, \Lsubs{x}{\sigma_1})$. 
    \item All other possibilities for $\sigma$ are treated as in case (\ref{pt:ho-out}). 
  \end{enumerate}
\end{proof}

\begin{prop}
\label{prop:bisim-mcl}
  For every session contract $\sigma \in \sc,\, \sigma \bisim{} \mcl{\sigma}$.
\end{prop}
\begin{proof}
  Let $\R$ be the set 
$\setof{(\sigma, \mcl{\sigma})}{\sigma \in \sc}$. It is straightforward to use Proposition~\ref{prop:moves} to 
show that $\R$ is a strong bisimulation, as given in Definition~\ref{def:sb}. 
\end{proof}

\begin{defi}\name{ Peer-duality }
\label{def:novel-dual}
  For every session contract $\rho \in \sc$, let $\prdual{\rho}$ denote $\stdual{\mcl{\rho}}$. 
\end{defi}

\begin{thm}
\label{thm:complements-comply}
Suppose $\B$ is a preorder. Then  $\rho \peercmpl[\B] \prdual{\rho}$ for every session contract
$\rho \in \sc$.
\end{thm}
\begin{proof}
  From Lemma~\ref{lem:m-closed} we know that $\mcl{\rho}$ is \emph{m-closed} and therefore 
an application of Theorem~\ref{thm:standard-duality-comply} gives 
$\mcl{\rho} \peercmpl[\B] \prdual{\rho}$. 

We also know from \rprop{bisim-mcl} that $\rho \bisim{} \mcl{\rho} $, 
and so the required $\rho \peercmpl[\B] \prdual{\rho}$ follows by \rprop{bisim}. 
\end{proof}

\subsection{The complement function}
In the original extended abstract of the current paper,
 \cite{DBLP:conf/concur/BernardiH14} we  proposed 
an alternative function $\comp{-} $ in order to construct
the complementary contracts required by
\rthm{dual-cmpl-with-ctr}; this function was proposed independently in
\cite{DBLP:journals/corr/abs-1202-2086}, but for a different purpose.  
%

In this section we show that the function $\comp{ - }$ suffers
the same issue of the standard duality: there exists a $\sigma$
which is not in $\B$-peer compliance with its proposed complement $\comp{
  \sigma }$, for every reasonable $\B$ (\rexa{comp-broken}).

We begin by recalling the necessary definitions, and then
we present a series of examples.

\begin{defi}\name{Complement}\label{def:comp}\\
Let $\compSym : L_\sc \longrightarrow L_\sc $ be defined inductively as
follows,
$$
\comp{\sigma} =
\begin{cases}
  \Unit & \mbox{if } \sigma = \Unit,\\
  x & \mbox{if } \sigma = x,\\
  ?\gt{t}.\comp{\sigma'} & \mbox{if } \sigma = { !\gt{t}.\sigma' },\\
  !\gt{t}.\comp{\sigma'} & \mbox{if } \sigma = { ?\gt{t}.\sigma' },\\
  ?(\sigma^m).\comp{\sigma'} & \mbox{if } \sigma = { !(\sigma^m).\sigma' },\\
  !(\sigma^m).\comp{\sigma'} & \mbox{if } \sigma = { ?(\sigma^m).\sigma' },\\
  \intsum_{i \in I}!\gt{l}_i.\comp{\sigma_i} & \mbox{if } \sigma =
  \extsum_{i \in I}?\gt{l}_i.\sigma_i,\\
  \extsum_{i \in I}?\gt{l}_i.\comp{\sigma_i} & \mbox{if } \sigma =
  \intsum_{i \in I}!\gt{l}_i.\sigma_i,\\
  \Rec[x]{\comp{\sigma' \inner{\sigma}{x}}} & \mbox{if } \sigma = \Rec[x]{\sigma'}
\end{cases}
$$
We say that $\comp{\sigma}$ is the {\em complement} of $\sigma$.
\end{defi}
In the definition above the application of $\inner{\sigma}{x} $ to
$\sigma'$ stands for the substitution of $\sigma$ in place of $x$ in
the message fields that appear in $\sigma'$. The definition is
the following,

$$
\rho \inner{ \sigma }{ x } =
\begin{cases}
\Unit & \text{if }\rho = \Unit,\\
y & \text{if }\rho = y,\\
?\gt{t}.(\rho'\inner{\sigma}{x}) & \text{if }\rho = {?\gt{t}.\rho'},\\
!\gt{t}.(\rho'\inner{\sigma}{x}) & \text{if }\rho = {!\gt{t}.\rho'},\\
?(\rho''\inner{\sigma}{x}).(\rho'\inner{\sigma}{x}) & \text{if }\rho = {?(\rho'').\rho'},\\
!(\rho''\inner{\sigma}{x}).(\rho'\inner{\sigma}{x}) & \text{if }\rho = {!(\rho'').\rho'},\\
\Rec[y]{( \rho' \inner{\sigma}{x} )} &\text{if } \rho = \Rec[y]{\rho'}\text{ and } y \neq x,\\
\Rec[x]{\rho'} &\text{if }\rho = \Rec[x]{\rho'}
\end{cases}
$$

\begin{exa}[ Inner substitution ]
In this example we show how the application of inner substitutions acts on terms.
Fix a term $\sigma \in \sc$, by definition $   ( ?(y).y
)\inner{ \sigma }{ y } = \sinp{ \sigma }{ y } $. This equality shows that
inner substitutions operate on the variables that appear free in the
message parts of terms, and not on the ones in the continuations. 
We have likewise the equality
$  ( ?(y).x )\inner{ \sigma }{ x}  =  ?(y).x$ 
because $x$ appears free only in the continuation of the term $ ?(y).x $.
As usual the substitution does not act on closed variables, thus
$(\Rec[ y ]{?( y ).x }) 
\inner{ \sigma }{ y } 
  = \Rec[ y ]{?( y ).x}$.
\end{exa}

As the application of $\inner{\sigma}{x}$ does not change the
number of prefixes that appear in a term, $\comp{\sigma}$ is defined
for every session type term $\sigma$.

\begin{exa}[ Complement function ]
\label{exa:comp-not-mclosed}
In this example we prove that in general the function $\comp{ - }$ 
does not compute m-closed terms.
Recall the type $\sigma$ we used in \rexa{applications-mcl}, namely 
$\sigma = \Rec[ x ]{\sigma'}$, where $\sigma' = \Rec[ y ]{ ?(y). x }$.
Calculations let us prove the equality
$$ 
\comp{ \sigma } = \Rec[ x ]{\Rec[ y ]{ !( \Rec[y]{ ?(y). x } ).x }}
$$
The crucial observation here is that the contract
$\comp{ \sigma }$ is not m-closed, for it contains the message
$ \Rec[y]{ \sinp{y}{x}} $, which is an open term.
\end{exa}
\leaveout{
Let $\sigma'' = \Rec[ y ]{ ?(y).\sigma}$,
$$
\begin{array}{lll}
\comp{ \sigma } & = & \comp{ \Rec[x]{ \sigma' } }\\
& = & \Rec[x]{ \comp{ \sigma' \inner{ \sigma }{ x } } }\\
& = & \Rec[x]{ \comp{ \sigma' } }\\
& = & \Rec[x]{\Rec[y]{ \comp{ ( ?( y ).x  ) \inner{ \sigma' }{ y } }}}\\
& = & \Rec[x]{\Rec[y]{ \comp{ ?( \sigma' ).x }}}\\
& = & \Rec[x]{\Rec[y]{ !(  \Rec[y]{ \sinp{y}{x}} ). x }}
\end{array}
$$
}

\begin{exa}
  \label{exa:comp-broken}
  In this example we show a contract $\sigma$ such that
  $ \sigma  \Npeercmpl[\B] \comp{\sigma}$ for every reasonable $\B$.
  Recall the type $\sigma = \Rec[ x ]{  \Rec[ y ]{ ?(y). x } }$ of
  \rexa{applications-mcl}. In \rexa{comp-not-mclosed} we have proven
  that 
\begin{center}
$\comp{ \sigma } = \Rec[ x ]{\Rec[ y ]{ !( \sigma' ).x }}$
  where $\sigma' = \Rec[ y ]{ ?(y). x }$
\end{center}
  Fix a reasonable relation $\B$.
  We want to prove that $ \sigma \Npeercmpl[\B] \comp{ \sigma }$. 
  Since $ \sigma \ar{\tau} \unfold{\sigma}$
  and $\comp{ \sigma } \ar{\tau} \unfold{ \comp{ \sigma} }$,
  it suffices to prove that 
  \begin{equation}
    \label{eq:comp-bad}
    \unfold{ \sigma } \Npeercmpl[\B] \unfold{ \comp{ \sigma } }
  \end{equation}
  Let us unfold the terms $\sigma$ and $\comp{ \sigma }$,
  $$
  \begin{array}{lll}
    \unfold{ \sigma } & = &  
    \unfold{ \Rec[ x ]{  \Rec[ y ]{ ?(y). x } } }\\
    & = &
    \unfold{ \Rec[ y ]{ ?( y ). \sigma }  }\\
    & = &
    ?(  \Rec[ y ]{ ?( y ). \sigma }  ). \sigma \\[1em]
    \unfold{ \comp{ \sigma} } & = & 
    \unfold{ \Rec[ x ]{\Rec[ y ]{ !( \Rec[ y ]{ ?(y). x } ).x }}  } \\
    & = & \unfold{ \Rec[ y ]{ !(  \Rec[y]{ ?( y ).\comp{ \sigma }} ). \comp{ \sigma }} }\\
    & = & !(  \Rec[ y ]{ ?( y ). \comp{ \sigma } } ). \comp{ \sigma }
  \end{array}
  $$
  As preliminary fact, observe that $(\unfold{\comp{ \sigma }}, \unfold{\sigma}) \not \in {\B} $,
  because $\unfold{\comp{\sigma}}$ performs an output action,
  $\unfold{\sigma}$ performs an input action, and $\B$ is reasonable.
  It follows that $(\comp{ \sigma}, \sigma) \not \in {\B} $.

  Let $\sigma_1 = \Rec[ y ]{ ?( y ). \comp{ \sigma } }$ and
  $\sigma_2 = \Rec[ y ]{ ?( y ). \sigma }$. Plainly,
  $\unfold{\comp{ \sigma }} = {!( \sigma_1 ).\comp{ \sigma }}$,
  and 
  $ \unfold{ \sigma } = {?( \sigma_2 ).\sigma}$.

  Since $\unfold{ \sigma } \nar{ \ok }$, \req{comp-bad} will follow if
  we show that 
  $$ \unfold{ \sigma } \Par \unfold{ \comp{ \sigma }} \nar[\B]{ \tau
  }$$
  Thanks to the syntax of the contracts $ \unfold{ \sigma }, \unfold{\comp{ \sigma }}$,
  and rule \rname{I-Synch}, we have to show that 
  $ !( \sigma_1)  \not \bowtie_{\B} ?(\sigma_2 )$.
  In view of the definition of $\bowtie_{\B}$ we have to prove that
  $ (\sigma_1, \sigma_2) \not \in {\B} $.
  The visible moves $\unfold{ \sigma_1 } \ar{ ?(\sigma_1) }
  \comp{\sigma} $,
  and
  $\unfold{ \sigma_2 } \ar{?(\sigma_2)} \sigma$, together with
  $ (\comp{ \sigma}, \sigma) \not \in {\B} $, and \rpt{reasonable-continuations}
  of \rdef{reasonable},
  ensure the desired $(\sigma_1, \sigma_2) \not \in {\B} $, thus
  $\sigma \Npeercmpl[\B] \comp{ \sigma }$.

\leaveout{
In this example we prove that 
the functions $\unfoldSym$ and $\compSym$ do not commute. 
In view of (\ref{eq:commutativity}), we are required
to exhibit a type $T$ such that 
\begin{equation}
\label{eq:nocomm}
\comp{\unfold{ T }} \not \typeEQ \unfold{ \comp{ T }}
\end{equation}
We use the type $T$ of examples (\ref{exa:unfold}) and (\ref{exa:comp}).
Recall that $\unfold{ T } = { \sinp{ \Rec[Y]{ \sinp{Y}{ T }} }{ T } }$.
Calculations let us show the next equalities,
$$
\begin{array}{lll}
  \comp{\unfold{ T }} & = & \sout{ \Rec[Y]{ \sinp{Y}{ T }} }{ \comp{ T } }
  \\[1em]

\end{array}
$$

To prove \req{nocomm} it suffices to show that
$\comp{\unfold{ T }}$ and $\unfold{ \comp{ T } }$ do not satisfy
\rpt{ho-out} of \rdef{typeq} with respect to the greatest coinductive type equivalence,
namely $\typeEQ$.
This will follow if we show that the messages of $\comp{\unfold{ T }}$ and $\unfold{ \comp{ T } }$ are not related by $\typeEQ$, that is
$$
 \Rec[Y]{ \sinp{Y}{ T }} \not \typeEQ \Rec[Y]{ \sinp{Y}{ \comp{ T }}}
$$
This will follow if we show that $\Rec[Y]{ \sinp{Y}{ T }}$ and $\Rec[Y]{ \sinp{Y}{ \comp{ T }}}$ do not satisfy \rpt{ho-inp} of \rdef{typeq}. To prove this is enough to show that the continuations of these types are not related, that is
$T \not \typeEQ \comp{ T }$, which we have
already established this in \req{T-neq-compT}.}
  \end{exa}
  \noindent
  The argument used in the previous example can be adapted to
  show that that the complement function on session types does
  not commute with the unfolding function:
  \begin{center}
    there exists a $T \in \types$, such that $\comp{\unfold{ T }} \not \typeEQ
    \unfold{ \comp{ T }}$
  \end{center}

\subsection{Discussion}
\label{sec:discussion}

\begin{figure}[t]
    \hrulefill
$$
\infer[ T^+ \D T^- \quad \rname{CRes}]
{ \Theta \, ; \, \Gamma \typerel \nn{\kappa}P \rhd \Delta }
{ \Theta \, ; \, \Gamma \typerel P \rhd \Delta \cdot \kappa^+ \at T^+ \cdot \kappa^- \at T^-}
$$
    \caption{The use of duality in type inference \'{a} la \cite{DBLP:journals/entcs/YoshidaV07}}
    \label{fig:cres}
\hrulefill
\end{figure}

Here we briefly discuss the use of the duality  operator
$\stdual{T}$  in
type-checking systems for session types. For processes we use the 
syntax of \cite{DBLP:journals/entcs/YoshidaV07},  and the
type-checking rules given on page 89,  and in Figure~6 on page
80 of the same paper.  For convenience we display in
\rfig{cres} a slightly generalised version of the main
rule involving duality;  there $\kappa^+, \kappa^-$ must have associated 
types  $T^+, T^{-}$ satisfying 
$T^+  \D T^{-}$. Here  $\D$ is some relation over
types which intuitively captures the notion of ensuring complementary
behaviour. In  \cite{DBLP:journals/entcs/YoshidaV07} this is
actually instantiated to duality: ${T^+}  \mathrel{\mathcal{D}} {T^{-}}$ if 
$\stdual{T^+} = T^{-}$.

Let us now reconsider the program $P$ in
\rexa{intro-funny-types-needed}
from the Introduction and discuss, informally, how it can be assigned 
a type. In order to use this instantiation of the  rule in \rfig{cres} we need to
assign to $\kF^+,\kF^-$  types satisfying $T^+ \D T^{-}$. 
For convenience we work up to the type equivalence 
$\typeEQ$ generated by the subtyping relation.

\leaveout{\MHf{
This informal discussion on types must now up 
with proposed types $T^+$ and $T^{-}$ which are not dual. 
}}
Assume that $z$ be at type $T_z$ (this could be stated in the syntax itself by using the annotation $ z \at T_z$). 
Since $z$ replaces the formal parameter $x$ in the
recursion~$ X \appto{ z, \, \kF^- } $, one expects $T_x$, the
type of $x$, to be equivalent to $T_z$, $ T_x \typeEQ T_z$.
By inspecting the syntax $$ \throw{x}{\kF^+}\void \Par \catch{y}{z} \ldots$$ we also know that
the endpoint $x$ is used according to the type $ T_x \typeEQ \sout{ T^{+} }\e $, 
and that $T^{+}$ must be a subtype of $T_z$, and so of $T_x$, $ T^{+} \subt T_x$.
The last inequality is trivially satisfied by letting $T^{+} \typeEQ T_x$, and this
leaves us with the equation $T^{+} \typeEQ \sout{ T^{+} }\e$. 
A session type that satisfies it is the following one, 
$$
T^{+} = \Rec[X]{\sout{X} \e}
$$

For $P$ to be typable it is necessary also that the type of $\kF^-$, namely $T^{-}$,
be complementary to $T^{+}$. Since in $ X \appto{ z, \, \kF^- }$ the endpoint $\kF^-$ replaces $y$,
it must be the case that $T^{-} \typeEQ T_y$.
At each iteration $y$ is used to read {\em only once} an endpoint 
of type $T_z$, so
$$T_y \typeEQ \sinp{T_z}\e \typeEQ \sinp{T_x}\e \typeEQ
\sinp{T^{+}}\e $$
and thus $ T^{-} = \sinp{T^{+}}\e$.
It may seem counter-intuitive that the behaviour of a process that
uses an endpoint according to the type $T^{+}$, i.e. $\Rec[X]{\sout{X} \e}$, 
be complementary to the behaviour of a process that uses the other endpoint
according to type $T^{-}$, i.e. $\sinp{T^{+}}\e$. But this is easily understood
if we unfold the types:
$$
\begin{array}{lll}
\unfold{T^{+}} & = & \sout{T^{+}}\e 
\\[.2em]
\unfold{T^{-}} & = & \sinp{T^{+}}\e
\end{array}
$$
More formally, since $\M{T^{-}} = \unfold{\M{\prdual{T^{+}}}}$
\rthm{complements-comply} ensures that for any preorder $\B$,
we have $ \M{T^{-}} \peercmpl[\B] \M{T^{+}}$, and by using the
types $T^{+}$ and $T^{-}$ we can type $P$ (see \rapp{type-inference}).
Our discussion shows that in general if an endpoint $\kappa^+$ is
used as prescribed by a type $T^+$, then the other endpoint of the
session, i.e. $\kappa^+$, needs to be used according to a type $T^-$
which contains the type of $\kappa^+$. In other terms, to know
how a system uses $\kappa^-$ one needs to know how the system
uses~$\kappa^+$.

Our proposal is to amend the type-checking system in 
 \cite{DBLP:journals/entcs/YoshidaV07}  by using the variation of 
the rule \rname{CRes} in \rfig{cres} where
\begin{align}\label{eq:D.comp}
  {T^+}  \mathrel{\mathcal{D}} T^{-} \,\,\text{whenever}\,\, \prdual{T^+} = T^{-}
\end{align}
Here we use $\prdualSym$ in the obvious manner as an operator on types
although formally it has only been defined on contracts. This version
of the rule has a behavioural justification due to
\rthm{complements-comply}
and the interpretation $\encSym$ of types as contracts. It ensures
that $T^+$ interpreted as a contract is in $\B$-mutual compliance with $T^{-}$: 
$\M{T^+}  \peercmpl[\B] \M{T^{-}}$.

This version of the rule, using $\prdualSym$ also leads to a more
powerful type-checking system. Using it we can type the program 
$P$ from \rexa{intro-funny-types-needed}. In 
\rapp{type-inference} we give all the details of a derivation tree for
of $\typerel P$, but using the generic rule from \rfig{cres}. 
The construction of this derivation tree gives rise to a series of
conditions on types, which boil down to:
\begin{align}
\label{eq:cond-Tx}
T_x & \typeEQ \sout{ T_x }\e \\[.5em]
\label{eq:cond-Ty}
T_y & \typeEQ \sinp{ T_x }\e  \\[.5em]
\label{eq:cond-D} T_x & \mathrel{\mathcal{D}}  T_y
\end{align}
We have already argued that (\ref{eq:cond-Tx}) is satisfied by $T^{+}$,
and that (\ref{eq:cond-Ty}) is satisfied by the $T^{-}$ we discussed earlier on,
so if we choose $\D$ as in (\ref{eq:D.comp}) above
(\ref{eq:cond-D}) is also satisfied. Thus the derivation of $ \typerel P $
exists with our suggested modified inference rule. 

If we choose $\D$ to be the standard duality, then 
(\ref{eq:cond-D}) is not satisfied, because $ \stdual{T^{+}} \not \typeEQ T^{-}$,
and so we cannot derive $\typerel P$.

Thus far, we have discussed only session types.
The type systems that use session types, though,
use the duality also to type channel (i.e. non session) types.

As an example, consider the type discipline of
\cite{DBLP:journals/iandc/Vasconcelos12}.
In that paper channels are resources that can be replicated, 
while session endpoints are resources that cannot be replicated.
This distinction is borne out by the types, which are pairs
$(q, \, T)$, or simply $q \,  T$, where $q$ is a qualifier that can be
either $\un$ (unbound) or $\lin$ (linear), and $p$ is a {\em
  pretype}. Pretypes are elements of $\st$.

\begin{exa}\name{Replication and ever-lasting communications}
\label{exa:replication}\\
we write the next process using the syntax of \cite{DBLP:journals/iandc/Vasconcelos12},
  $$
  Q = \nn{xy}( \, \vOut{x}{x}\void \Par \un \vIn{y}{z}\vOut{z}{z}\void \,)
  $$
  The process $Q$ creates two fresh names $x$ and $y$,
  then sends $x$ sends over itself, to the process
  $ \un \vIn{y}{z}\vOut{z}{z}\void $, which is a replicated input.
  Upon reception of $x$ over $y$, the replicated input reduces to
  $ \vOut{x}{x}\void \Par \un \vIn{y}{z}\vOut{z}{z}\void$. 
  This entails a livelock, and no communication takes place.

  Since $x$ sends itself, it must be an unbound resource,
  that is it can be duplicated. The syntax $\un
  \vIn{y}{z}\vOut{z}{z}\void$ means that also the name $y$ can be
  duplicated. In turn the types of both $x$ and $y$ will be decorated
  with the qualifier $\un$. 
\end{exa}

To show that $\typerel Q$ we use the derivation rules of
\cite{DBLP:journals/iandc/Vasconcelos12}, but replacing 
the rule that  depends on the duality, namely \rname{T-Res} of that
paper, with the rule in \rfig{tres}.
\begin{figure}[t]
\hrulefill
$$
\infer[ T^{+} \D T^{-} \quad \rname{T-Res}]
{ \Gamma \typerel (\nu{xy})P }
{\Gamma, x\at T^{+}, y\at T^{-} \typerel P}
$$
\caption{The use of duality in type inference \'{a} la \cite{DBLP:journals/iandc/Vasconcelos12}}
\label{fig:tres}
\hrulefill
\end{figure}
The derivation of $\typerel Q$, whose details are in
\rapp{type-inference}, depends on the satisfaction of three
requirements, namely
$$
T_x  \typeEQ \un \sout{ T_x }\e, \quad
T_y  \typeEQ \un \sinp{ T_x }\e, \quad
T_x  \mathrel{\mathcal{D}}  T_y
$$
Also in this case the power of the type system can be improved
by using $\prdualSym$ in place of the standard
duality.
If, modulo the qualifiers, we instantiate $\D$ as in (\ref{eq:D.comp})
then we can derive $\typerel Q$, while if we instantiate $\D$ to the
standard type duality, then we cannot derive $\typerel Q$.
The details are in \rapp{type-inference}.


\section{Conclusion}
\label{sec:conclusions}
\newcommand{\sepi}{\textsc{SePI}\xspace}

In this paper we proposed a new model for recursive higher-order
session types \cite{DBLP:conf/esop/HondaVK98}, which is \fllyabs with
respect to the subtyping relation \cite{DBLP:journals/acta/GayH05}.
The interpretation of session types maps them into higher-order
session contracts. This is a sublanguage of higher-order contracts for
web-services. 

To construct the model, we have equipped
those contracts with a novel behavioural theory.
In our theory the observable behaviour of contracts is expressed via
a standard LTS, but the interactions of contracts
are parametrised over preorders $\B$s.
The result is a family of LTS. For each one of them we defined a
mutual compliance, $\peercmpl[\B]$. Then we defined a family of
behavioural preorders, $\peerleq[\B]$, such that $ \sigma_1
\peerleq[\B] \sigma_2$ is all the contracts in $\B$-mutual compliance
with $\sigma_1$ is in $\B$-mutual compliance with $\sigma_2$.
The preorder that models the subtyping is the greatest
solution of the equation $ X = {\peerleq[X]}$, and
\rthm{full-abstraction} shows that the model is \fllyabs.

The technical development relies on a coinductive syntactic
characterisation of the preorders $\peerleq[\B]$
(\rcor{set.syn}).
The completeness of the characterisation depends on the existence
for every contract $\rho$ of a contract $\sigma$ in $B$-mutual compliance
with $\rho$, at least when $\B$ is a preorder.
To prove this we introduced a novel function called 
\emph{peer-dual},
and we have shown that this function improves on  the standard
duality, in that it allows us to type more well-formed processes.

Moreover, the examples in \rsec{discussion} suggest that type checking
algorithms based on the standard inductive duality, are not complete
with respect to type disciplines based on the {\em coinductive
  duality} of \cite[Def. 9]{DBLP:journals/acta/GayH05}.

In summary, the contributions of this paper are two, namely
\begin{itemize}
\item the first fully-abstract model of the subtyping for session types of \cite{DBLP:journals/acta/GayH05}, 
\item the definition of a novel type duality, which leads to type systems more powerful than the one relying on the standard definition of type duality \cite{DBLP:conf/esop/HondaVK98}.
\end{itemize}

\subsection{Related work}

The material in sections~2,3, and 4 is adapted from chapter~8 of
\cite{gbthesis}. \rlem{model-complete} in this paper is analogous to
Lemma~8.1.5 of \cite{gbthesis}, of that thesis. Lemma~8.1.5 relies on
Lemma~7.2.18 (also of \cite{gbthesis}) which turns out to be false:
and a counter example is the session contract $ \Rec[x]{!(x).\Unit}$.
Spurred on by this problem, we investigated a novel definition of
type duality, that we described in \rsec{complement}.\\

\paragraph{\bf Contracts for web-services:}
First-order contracts for web-services and the notion of contract
compliance have been proposed first in
\cite{DBLP:conf/wsfm/CarpinetiCLP06}, and have been improved on in
\cite{LP07}. In \cite{DBLP:conf/wsfm/CarpinetiCLP06} a sub-contract
relation is defined, which leads to the definition of compliance.
In contracts, in \cite{LP07} the compliance is defined in terms of the
LTS of contracts, and then, in the style of testing theory \cite{dNH84},
the sub-contract preorder is defined using the compliance.
All the subsequent works - including this paper - adhere to that style.

The most recent accounts of first-order contracts for web-services are
\cite{DBLP:journals/tcs/Padovani10,DBLP:journals/toplas/CastagnaGP09}.
A striking difference between the two papers is the treatment of
infinite behaviours.
In \cite{DBLP:journals/tcs/Padovani10} infinite behaviours are
expressed by recursive contracts, whereas in
\cite{DBLP:journals/toplas/CastagnaGP09} there is no recursive
construct, $\Rec{-}$, and the theory accounts for infinite behaviours by 
using a {\em coinductively} defined language.
Our treatment of infinite behaviours follows the lines of
\cite{DBLP:journals/tcs/Padovani10}.

Both
\cite{DBLP:journals/tcs/Padovani10,DBLP:journals/toplas/CastagnaGP09}  
define a subcontract preorder for contracts of server processes,
a more generous {\em weak} subcontract. They propose mechanisms to
coerce contracts, namely orchestrators \cite{DBLP:journals/tcs/Padovani10} and filters \cite{DBLP:journals/toplas/CastagnaGP09}, and show that
if two contracts, $\sigma_1, \sigma_2$ are in the weak subcontract,
then there exists a coercion $f$ such that $\sigma_1$ and
$f(\sigma_2)$ are in the subcontract relation
\cite[Corollary~3.11]{DBLP:journals/tcs/Padovani10},
\cite[Theorem~3.9]{DBLP:journals/toplas/CastagnaGP09}.

The authors of \cite{DBLP:journals/toplas/CastagnaGP09} introduce
also a compliance for processes, and show that if contracts are
associated to processes by a relation that satisfies a
number of properties (Definition~4.2 in that paper), then two
processes are in compliance if the associated contracts are in
compliance (see Theorem~4.5 there).
\leaveout{
First-order contracts for web-services and the sub-contract relation 
were proposed as model of the subtyping of first-order session types
already in \cite{DBLP:conf/birthday/LaneveP08}, but the first sound
model was shown in \cite{DBLP:conf/ppdp/Barbanerad10}, by using a
sublanguage of contracts for web-services, called {\em session
  behaviours}. Our work \cite{DBLP:conf/sac/BernardiH12} uses yet
another sublanguage of contracts for web-services, namely {\em session
  contracts}, to show that the model proposed in
\cite{DBLP:conf/ppdp/Barbanerad10} is also complete. In this paper we
have used a language of {\em higher-order} session contracts.
}

\paragraph{\bf Session types:}
Session types has been presented for the first time in
\cite{DBLP:conf/parle/TakeuchiHK94}. There the language for types
is higher-order and without recursion, and the type duality is
defined inductively in the obvious way. The main result of
\cite{DBLP:conf/parle/TakeuchiHK94} is that well-typed programs cannot
incur in communication errors (see Theorem 5.10 there). This
type-safety result is a landmark of session types, and is
proven is almost every presentation of session typed languages.

The original presentation of \cite{DBLP:conf/parle/TakeuchiHK94} has
been extended extended to recursive higher-order session types in
\cite{DBLP:conf/esop/HondaVK98}, where also the definition of type
duality that we reported in \rfig{std-dual} has been proposed.
The authors of \cite{DBLP:conf/esop/HondaVK98} argue in favour of
program abstractions, that help programmers structure the interaction
of processes around sessions. As in
\cite{DBLP:conf/parle/TakeuchiHK94}, the proposed result is that
a ``typable program never reduces into an error'' (see Theorem 5.4 (3)
of \cite{DBLP:conf/esop/HondaVK98}). In \cite[pag.~86, paragraph
4]{DBLP:journals/entcs/YoshidaV07}, though, it is shown that that
result is not true, that is the type
system of \cite{DBLP:conf/esop/HondaVK98} 
does not satisfy type-safety. 
The authors of \cite{DBLP:journals/entcs/YoshidaV07} amend the type
system of \cite{DBLP:conf/esop/HondaVK98}, thereby achieving
type-safety (see Theorem~3.4 of \cite{DBLP:journals/entcs/YoshidaV07}).

Subtyping for recursive higher-order session types has been introduced
in \cite{DBLP:journals/acta/GayH05}, along with a {\em coinductive} 
definition of the duality.
In addition to the standard type-safety result (Theorem~2), the
authors show also a type-checking algorithm which they prove sound
(Theorem~5) wrt the type system.
The proof of soundness, though, relies on a relation between
the inductive and the coinductive dualities (Proposition~5 there)
which in general is false; a counter example is provided by the
session type $\Rec{\sout{X}\e}$. The consequence is that there is the
possibility that the algorithm of  \cite{DBLP:journals/acta/GayH05},
if employed in more general settings, may not be sound, that is reject
programs which are well-typed.

An alternative ``fair'' subtyping has been proposed recently
in \cite{Padovani11c}. There session types are higher-order and
recursive, their operational semantics is defined by parametrising the
interactions of session types on pre-subtyping relations, and the
fair subtyping is defined as a greatest \fp (Definition~2.4).

In our development we adopted the same technique of
\cite{Padovani11c}. However, our aim was to model the standard
subtyping of \cite{DBLP:journals/acta/GayH05}, while Padovani focuses
on the properties of his fair subtyping.

An overview of the major development of session types is
\cite{DBLP:conf/wsfm/Dezani-Ciancaglinid09}.

Session types have been fertile ground for theoretical studies
as well as for implementations. For instance, a programming language
equipped with session types is \sepi \cite{franco.vasconcelos:sepi}.
In \sepi, the process $P$ of
\rexa{intro-funny-types-needed} is rendered by the code in
\rfig{sepi-code}, and the type checker accepts $P$, because it uses
the {\em coinductive} duality rather than the inductive one
\cite{vasco.email}.
\begin{figure}
\hrulefill
\begin{lstlisting}
type T = !T.end

def f (x: T, y: dualof T) =
	new a b: T
	x!a |
	y?z. f!(z, b)

{}
\end{lstlisting}
\caption{Script in \sepi for the process $P$ of \rexa{intro-funny-types-needed}}
\label{fig:sepi-code}
\hrulefill
\end{figure}

\leaveout{
{\em Behavioural models:}
}

\paragraph{\bf Models of Gay \& Hole subtyping:}
The first attempt to model the Gay \& Hole subtyping of \cite{DBLP:journals/acta/GayH05} in terms of
a compliance preorder appeared in
\cite{DBLP:conf/birthday/LaneveP08}. 
For a comparison of that research and our work the reader is
referred to  \cite{DBLP:conf/sac/BernardiH12}.
The authors of \cite{DBLP:conf/ppdp/Barbanerad10} have shown the
first sound model of this subtyping restricted  to first-order session
types, by using a subset of contracts for web-services, a mutual compliance, called
{\em orthogonality}, and the preorder
generated by it.  The $\B$-peer compliances we used in this work generalises to
parametrised LTSs  the orthogonality of
\cite{DBLP:conf/ppdp/Barbanerad10}.

\leaveout{
Our research stemmed from the results on behavioural models for
first-order session types. The authors of
\cite{DBLP:conf/ppdp/Barbanerad10} have shown the
first sound model of the subtyping, by using a subset of contracts for
web-services, a mutual compliance, called
{\em orthogonality}, and the preorder
generated by it.

The $\B$-peer compliances we used in this work generalise to
parametrised LTS the orthogonality of
\cite{DBLP:conf/ppdp/Barbanerad10}.
}

Following the approach of \cite{DBLP:conf/ppdp/Barbanerad10}, in 
\cite{DBLP:conf/sac/BernardiH12} we have shown a
\fllyabs model of the subtyping, but using the standard asymmetric
compliance and an intersection of the obvious server and client
preorders.

An alternative definition of the model proposed in
\cite{DBLP:conf/sac/BernardiH12} can be found in
\cite[Chapter~5]{gbthesis}, where must testing of \cite{dNH84} is used
in place of the compliance.
This entails immediately the differences between the standard 
subtyping of \cite{DBLP:journals/acta/GayH05} and the fair one of
\cite{Padovani11c}, for they are just the differences between must
preorder and the should preorder.

This paper subsumes \cite{DBLP:conf/sac/BernardiH12} in
the sense that \rthm{full-abstraction} of this work implies
Theorem~5.2 of \cite{DBLP:conf/sac/BernardiH12}.

The authors of \cite{DBLP:conf/ppdp/Barbanerad10} have also extended their work
to higher-order session types. 
However in their recent work \cite{Barbanera-deLiguoro-2013} only
a subset of types is modelled, for instance
the types $\Rec{\sout{X}\e}$ is ruled out, and has no behavioural
interpretation at all. This is because the authors of
\cite{Barbanera-deLiguoro-2013} use the standard
definition of syntactic duality, and define the LTS of contracts by
stratification.
In contrast with \cite{Barbanera-deLiguoro-2013}, we have argued that 
types as $\Rec{\sout{X}\e}$ are necessary to type a series of
well-formed processes. 
Because of this, we have replaced the standard
type duality function, $\stdual{\sigma}$, with our novel 
\emph{peer-duality}
function, $\prdual{\sigma}$, and we have
defined the LTS of contracts {\em coinductively}.
As a result, the type $\Rec[X]{\sout{X}\e}$ in our theory is given the
following observable behaviour, where $\rho = \M{\Rec{\sout{X}\e
}}$:
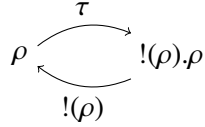
$$
\begin{tikzpicture}
  \node (rho) at (0,0) {$\rho$};
  \node (unf) at (2,0) {$!(\rho).\rho$};

  \path[->]
  (rho) edge [bend left] node [above] {$\tau$} (unf)
  (unf) edge [bend left] node [below] {$!(\rho)$} (rho);
\end{tikzpicture}
$$

\paragraph{\bf Semantic subtyping:}
We view our main result as a behavioural or \emph{semantic}
interpretation of Gay \& Hole  subtyping.  There  is an alternative
well-developed approach to semantic theories of types and subtyping
\cite{semsubtypes} in which the denotation of a type is given by the
set of values which inhabit it, and subtyping is simply subset
inclusion. This apparent simplicity is tempered by the fact that for 
non-trivial languages, such as the pi-calculus 
\cite{DBLP:journals/tcs/CastagnaNV08}, there is a circularity in the
constructions due to the fact that determining which terms are values 
depends in turn on the set of types. This circularity is broken using 
a technique called \emph{bootstrapping} or \emph{stratification},
essentially an inductive approach. 
The research using this approach which is closest to our results on
Gay \& Hole subtyping may be found in
\cite{DBLP:conf/ppdp/CastagnaDGP09}; this contains a treatment of 
a very general language of session types, an extension of Gay \& Hole
types. But there are essential differences. The most important is that
their model does not yield  a semantic theory of Gay \& Hole
subtyping.  Their subtyping relation,  $\leq$, is defined  via an LTS 
generated by considering the transmission of values rather than
session types; effectively subtyping is not allowed on messages. 
The resulting subtyping  is very different than our focus of concern, 
the Gay \& Hole subtyping
relation  $\subt$. 
For example the preorder $\leq$ has bottom elements, in contrast to 
$\subt$, and $\sinp{\int}\e \subt \sinp{\real}\e$ whereas 
$ \sinp{\int}\e \not \leq \sinp{\real}\e $.
The particular use of \emph{stratification} (Theorem~2.6) is also
complex, and rules out the use of session types such as
$\Rec[X]{\sout{X}\e}$.
Finally they use as types infinite regular trees whereas we prefer to
work directly with recursive terms, 
as proposed in \cite{DBLP:journals/acta/GayH05}; for example this
allows us to discuss the inadequacies of the type-checking rules of
\cite{DBLP:journals/entcs/YoshidaV07}. 

Nevertheless the extended language of sessions types of
\cite{DBLP:conf/ppdp/CastagnaDGP09} is of considerable significance.
It would be interesting to see if it can be interpreted
behaviourally using our co-inductive approach, particularly endowed
with a larger subtyping preorder more akin to the standard 
Gay \& Hole relation \cite{DBLP:journals/acta/GayH05}.

{\em Further behavioural models:}
Recently, a behavioural model for multi-party first-order
session types appeared \cite{DBLP:conf/icalp/DenielouY13}, which is
based on communicating automata rather than contracts for
web-services. The focus of \cite{DBLP:conf/icalp/DenielouY13} is not
to model the subtyping.

\subsection{Future work}
In \cite{DBLP:conf/sac/BernardiH12}, building on
\cite{DBLP:conf/ppdp/Barbanerad10}, we have already developed a result
similar to the \fabs of \rsec{results}, but for for
\emph{first-order} session contracts, and using a combination of 
\emph{server} and \emph{client} subcontract relations.
We leave as future work showing how this approach  
can be recovered from our \emph{peer} subcontract relation.

Even though standard models of recursive types are based on regular
trees \cite{DBLP:journals/mscs/PierceS96,DBLP:journals/fuin/BrandtH98}, tree models for  recursive higher-order session types
are still lacking. We plan to develop such a model, and to show the
connection with the notions we used in this paper.
Establishing such a connection will help us
motivating the complement function, and showing its connection with other
notions of duality, for instance the co-inductive one of
\cite[Def. 9]{DBLP:journals/acta/GayH05}.

Recently, type disciplines based on
session types have been proposed, which guarantee that well-typed
programs are free from deadlocks
\cite{DBLP:conf/tgc/Dezani-CiancaglinidY07,DBLP:conf/icfp/Wadler12,DBLP:conf/coordination/VieiraV13}.
However, all those papers deal but with finite (i.e. non recursive)
types. We plan to investigate whether the semantic techniques we used
in this work lead to type systems for recursive session types
that guarantee deadlock freedom.

\paragraph{\em Acknowledgements}
The authors would like to thank Vasco Vasconcelos for 
having provided the \sepi code shown in Figure~(\ref{fig:sepi-code}),
for his remarks on the \sepi type checker, and for having
brought the type $\Rec[ x ]{  \Rec[ y ]{ ?(y). x } }$ of \rexa{comp-broken}
to our attention.

\appendix
\section{Standard Definitions}
\label{app:defs}

The following definitions apply equally well to the language of
session types and the language of session contracts. Periodically we
change from one to the other in the exposition.  

In the language $\Lsc$ we have the standard notion of free and bound
occurrences of the recursion variables $X,\, Y,\ldots$ which lead to
the standard notion of closed terms. 

\subsection{Substitutions}

A substitution is a finite partial map from variables to closed terms. 
Here we use as an example language that of session contracts,
thus for the language $\Lsc$ a substitution takes the form
\begin{align*}
  \subsS: \vars  \rightharpoonup  \sc 
\end{align*}
where $\dom{\subsS}$ is a finite subset of $\vars$. 
Substitutions are composed by letting 
$\subsS_1 \cdot \subsS_2$ have as domain $\dom{\subsS_1} \cup \dom{\subsS_2}$,
with 
\begin{align*}
  &\subsS_1 \cdot \subsS_2 (x) = 
\begin{cases}
   \subsS_1(x), &\text{if  $x \in \dom{\subsS_1}$}\\
   \subsS_2(x), &\text{ if  $x \in \dom{\subsS_2},\, x \not\in \dom{\subsS_1}$}
\end{cases}
\end{align*}
This, together with the empty substitution, denoted $\varepsilon$, endows 
the collection of substitutions with the structure of a monoid. We use 
two further operations on substitutions: $\SnotX$ has as domain 
$\dom{\subsS} - \sset{x}$ and acts like $\subsS$ on all variables in its domain, 
while for any operator $g$, $g(\subsS)$ is the substitution with the same domain as 
$\subsS$ and maps each $X$ in this set to $g(\subsS(x))$. 

The action of a substitution on a term $T$, written $T\subsS$, is defined by 
structural induction on $T$. Thus for the language  $\Lsc$
the definition is as follows:

$$
\rho\, \subsS = 
\begin{cases}
  \Unit & \text{if }\rho = \Unit\\
  y & \text{if }\rho = y, y \not\in \dom{\subsS}\\
  \subsS(y) & \text{if } \rho = y, y \in \dom{\subsS}\\
  ?\gt{t}.(\rho'\subsS) & \text{if }\rho = ?\gt{t}.\rho'\\
  !\gt{t}.(\rho'\subsS) & \text{if }\rho = !\gt{t}.\rho'\\
  ?( \rho'' \subsS).(\rho'\subsS) &
  \text{if }\rho = ?( \rho'').\rho'\\ 
  !( \rho'' \subsS).\rho'\subsS &
  \text{if }\rho = !( \rho'').\rho'\\
  \Rec[y]{(\rho' \Snoty)}&\text{if }\rho = \Rec[y]{\rho'}
\end{cases}
$$
Here the important clause is the last one; when applying $\subsS$ to
the recursive term $ \Rec[y]{\rho'}$ it applies the restricted 
substitution $\Snoty$, which leaves the variable $y$ untouched, to the 
body $\rho'$. Note that no notion of $\alpha$-conversion is required. 

It is easy to show the standard compositional property of
substitutions, namely
$(\rho\, \subsS_1)\subsS_2  = \rho\, (\subsS_1 \cdot \subsS_2)$,
by structural induction  on $\rho$.

\subsection{Unfoldings}

\newcommand{\aconv}{\mathrel{=_\alpha}}

\begin{figure}[t]
\hrulefill
$$
\infer
{\Rec{S} \mathrel{\unfoldSym} T}
{ S \subst{\Rec{S}}{X} \mathrel{\unfoldSym} T}
\qquad\qquad
\infer[T \not = \Rec{S}]
{ T \mathrel{\unfoldSym} T }
{}
$$
\caption{Inference rules for the unfolding of terms. \label{fig:unfolding}}
\hrulefill
\end{figure}

The unfolding of a closed session term  from  $L_{\st}$  is defined {\em inductively},
as the {\em least} fixed point of the functional generated by the
rules in \rfig{unfolding}.

It is easy to check that $\unfoldSym$ is a function over closed session
terms, that is if $ S \mathrel{\unfoldSym} T_1$ and $S \mathrel{\unfoldSym} T_2$ then
  $T_1 = T_2$. For this reason we use the standard functional notation 
$\unfold{S}$ to denote the unfolding of $S$. However, this function is
not total, as the following example explains.

\begin{exa}
  Observe that 
  there is no finite inference using the rules from Figure~\ref{fig:unfolding}
  that lets us derive an unfolding 
  of $\Rec{X}$. The reason is that the last rule in the derivation has
  to be 
  $$
  \infer[]
  { \Rec{X} \mathrel{\unfoldSym} T }
  { \infer[]{\Rec{X} \mathrel{\unfoldSym} T}{ \vdots }}
  $$
  which leads to a circularity.
For this reason $\unfold{\Rec{X}}$ is undefined. 
\end{exa}
\noindent
It is easy to check also that the  function is idempotent, $\unfold{\unfold{S}} =\unfold{S}$.
In fact $\unfold{S}$ is reminiscent of a \emph{head-normal form},
in the sense that it must take one of the first five forms in the
grammar for $L_{\st}$ in \rfig{st}.

In order to isolate the terms whose subterms can be unfolded,
we use the notion of guarded. Our definition is a mild generalisation of
the standard one.
Recall that all the constructors but $\Rec[X]{-}$ are non-recursive.
\begin{defi}\name{Guarded}\label{def:guarded}\\
We say a recursion variable $X$ is guarded in $T \in L_{\st}$ if
every occurrence of $X$ in $T$ appears under a non-recursive type
constructor.\leaveout{a prefix.} 
Then we say that a term $S \in L_{\st}$ is \emph{guarded}  if
whenever  $\Rec{T}$ is a subterm of $S$ then $X$ is guarded in $T$.
\end{defi}

\begin{exa}
  Every variable $X$ is a term, and it is not guarded (in itself).
  In the terms $\Rec{\sinp{X}\e}$ and $\Rec{\sinp{\e} X}$ the variable $X$ is guarded, for
  it appears under the constructor $\sinp{-}-$, and so the whole terms are guarded.

  Let $T = \Rec[Y]{Y}$ and $S = \branch{ \gt{l}\as T}$.
  The variable $Y$ is guarded in the term $S$, for it occurs underneath the type constructor $\branch{ - }$, and it is not guarded in the
  the term $T$, because there it occurs directly after the recursive constructor $\Rec[Y]{-}$.
  
  Neither $S$ nor $T$ are guarded. The reason is that $S$ is a subterm of itself, and of $T$,
  and $Y$ is not guarded in $S$.
\end{exa}

\leaveout{
In the next example we explain why we slightly generalised the standard
definition of guarded variable.
\begin{exa}
  According to the standard definition a variable $X$ is guarded in $T
  \in L_{\st}$ if every occurrence of $X$ in $T$ appears under a {\em
    prefix}.
  In the term $T = \Rec{\sout{X}\e}$ the variable $X$ does not appear
  under a prefix, so $X$ is not guarded in $T$, and consequently $T$ is
  not guarded.
  However, the unfolding of $T$ is defined, $\unfold{T} = \sout{T}\e$,
  thus the standard definition is too demanding, in that it rules out
  terms that can be unfolded.

  Instead, according to \rdef{guarded} $T$ is guarded, thus, since it
  is also closed, it is admitted among the session types.
\end{exa}
}
The application of a substitution to a term $T$ in which every variable is guarded
preserves the top-most constructor of $T$. This phenomenon lets us prove the next lemma.

\begin{lem}
\label{lem:guarded-vars-and-unfold}
  For every $T \in L_{\st}$ and substitution $\mathbf{s}$, if 
  every variable in $T$ appears guarded, then
  $\unfold{T \mathbf{s}}$ is defined.
\end{lem}
\begin{proof}
  The proof is by structural induction on $T$. Every variable in $T$
  appears guarded, thus $T$ cannot be a variable. 

  The only case worth discussion is when $T = \Rec[Y]{T'}$, for in
  every other case an application of the axiom in \rfig{unfolding}
  ensures the result. 

So suppose $ T = \Rec[Y]{T'} $. Then $T
  \mathbf{s} = \Rec[Y]{ ( T' \SnotY ) }$. 
  The hypothesis implies  that every variable in $T'$ is guarded,
  and since $T'$ is a subterm of $T$, by structural induction 
  we know that $\unfold{T' \mathbf{s'}} $ is defined for every 
  substitution $\mathbf{s'}$. 

Let us pick the particular substitution
$\mathbf{s'} =  \SnotY \cdot \mathbf{s''} $, where $\mathbf{s''}$ maps 
the variable $Y$ to $T\mathbf{s} $.
We know by induction that we can infer
\begin{align*}
  T' (\SnotY \cdot \mathbf{s''} ) \mathrel{\unfoldSym} R
\end{align*}
for some $R$. But by compositionality we know
$T' (\SnotY \cdot \mathbf{s''} )  = ( T' \SnotY ) \mathbf{s''}$ 
and so one application of the rule on the right in
  \rfig{unfolding} gives the required
\begin{align*}
  T \mathsf{s} \mathrel{\unfoldSym} R
\end{align*}
\end{proof}

\begin{lem}
  If $T$ is a guarded closed term in $L_{\st}$ then $\unfold{T}$ is
  defined. 
\end{lem}
\begin{proof}
For every substitution $\mathbf{s}$, since $T$ is closed,
$T \subsS = T$. \rlem{guarded-vars-and-unfold} ensures that $\unfold{T \mathbf{s}}$ is defined, so $\unfold{ T }$ is defined.
\leaveout{
$T \mathbf{s}$ performs only the renaming of the bound variables of $T$.
A result analogous to \cite[Corollary 3.11]{DBLP:journals/tcs/Stoughton88} 
ensures that $ T \mathbf{s} \aconv T$.
\rlem{guarded-vars-and-unfold} ensures that $\unfold{T \mathbf{s}}$ is defined,
so \rlem{aconv-unfold} implies that $\unfold{T}$ is defined.
}
\leaveout{
  We prove a slightly more general result. Let $\mathbf{s}$ denote a
  general substitution, that is a function from recursion variables to
  closed terms. We show by structural induction that
  $\unfold{T\mathbf{s}}$
  exists for every $\mathbf{s}$, and every $T$, provided every variable appears guarded
  in $T$. The required result then follows since if $T$ is guarded and
  closed then $T \mathbf{s}$ coincides with $T$.

The only nontrivial case in the proof of the more general result is
when $T$ has the form $\Rec{U}$, and we know that $X$ appears guarded
in $U$. Let $\mathbf{s'}$ denote the substitution which coincides with
$\mathbf{s}$ except that $X$ is mapped to $T\mathbf{s}$ and let
$\mathbf{s''}$ be defined similarly, but mapping $X$ to itself. Note that 
$T\mathbf{s}$ coincides with $\Rec{(U\mathbf{s''})}$.

Therefore, because of the first rule in Figure~\ref{fig:unfolding}, to
show
$\unfold{T\mathbf{s}}$ is defined it is sufficient to show that
$\unfold{\,(U\mathbf{s''})\subst{ T\mathbf{s}}{X}\,}$ exists. 
But $(U\mathbf{s''})\subst{ T\mathbf{s}}{X}$ is the same as 
$U\mathbf{s'}$ and by structural induction we know that 
$\unfold{U\mathbf{s'}}$ exists. }
\end{proof}

The main properties used of the unfold function are now collected in the
following lemma:
\begin{lem}\label{lemma:unfold}
  For every session contract $\sigma \in \sc $, 
  \begin{enumerate}[label=\({\alph*}]

  \item  if $\sigma \arstar{\tau} \sigma' $ and $\sigma'$ is stable,
  then $\unfold{\sigma} \arstar{\tau} \sigma'$.   

  \item if $\sigma$ is \emph{m-closed} then so is $\unfold{\sigma}$. 
  \end{enumerate}
\end{lem}
\begin{proof}
  Recall that all terms in $\sc$ are closed and guarded, and therefore
  by the previous lemma, applied to session contracts,
  $\unfold{\sigma}$ is defined. For convenience let it be denoted by
$\rho$. 

The proof of (a) now proceeds by induction on the derivation of
$\sigma  \,\unfoldSym \,  \rho$ from the rules in 
Figure~\ref{fig:unfolding}, and a case analysis on the structure 
of $\sigma$. 
The only non-trivial case is when $\sigma$ has the form 
$\Rec[x]{\sigma_1}$. 

Since $\sigma'$ is stable and there is exactly
one rule from Figure~\ref{fig:opsem} which can apply to 
$\Rec[x]{\sigma_1}$, the sequence of transitions $\sigma \arstar{\tau}
\sigma' $ must actually take the form
\begin{equation*}
  \Rec[x]{\sigma_1} \ar{\tau}  \sigma_1\subst{\sigma}{x} \arstar{\tau} \sigma'
\end{equation*}
By induction 
$\unfold{ \sigma_1\subst{\sigma}{x}} \arstar{\tau} \sigma'$. 
The result now follows since by the rules of 
 Figure~\ref{fig:opsem} we know 
$\unfold{ \rho\subst{\sigma}{x}}$ coincides with $\rho$.

The proof of statement (b) has a similar structure.
\end{proof}
\noindent
Note that in the statement of Lemma~\ref{lemma:unfold}(a) the hypothesis 
that $\sigma'$ be stable is essential. 
As a counterexample consider the case when 
$\sigma$ is $\Rec[x]{\Rec[y]{!\gt{l}.x}}$.
We have $\sigma \arstar{\tau} \Rec[y]{!\gt{l}.\sigma}$ but
$\unfold{\sigma} \Narstar{\tau} \Rec[y]{!\gt{l}.\sigma}$, since
$\unfold{\sigma}$ is $!\gt{l}.\sigma$.

\leaveout{
\subsection{Some technical results}
\label{sec:trs}

The length $\lenSym : \Lsc \longrightarrow \mathbb{N}$ of a term $\rho$ is the amount of prefixes and recursive constructors in the longest series of prefixes that $\rho$ contains; the definition is inductive,
$$
\len{\rho} =
\begin{cases}
  0 &\text{if }\rho \in \sset{ \Unit, \, x }\\
  1 + \len{\rho'}&\text{if }\rho \in \sset{ !\gt{t}.\rho', \, ?\gt{t}.\rho', \, ?(\sigma).\rho', \, !(\sigma).\rho' }\\
  1 + \max \setof{ \len{\sigma_i}}{ i \in I }&\text{if }\rho = \extsum_{i \in I}?\gt{l}_i.\sigma_i\\
  1 + \max \setof{ \len{\sigma_i}}{ i \in I }&\text{if }\rho = \intsum_{i \in I}?\gt{l}_i.\sigma_i\\
1 + \len{\rho'} & \text{if } \rho = \Rec[x]{\rho'}
\end{cases}
$$

The application of the substitution $\inner{\cdot}{\cdot}$ does not
change the length of terms, and so $\comp{\rho}$ is defined for every
$\rho$.
\begin{lem}
\label{lem:comp-defined}
  For every $\rho \in \Lsc$,
\begin{enumerate}[a)]
\item for every $\sigma \in \Lsc$,  $\len{\rho} = \len{\rho \inner{\sigma}{x}}$ 
\item if $ \rho = \Rec[x]{\rho'}$ then $\len{\rho' \inner{ \rho }{ x }} < \len{\rho}$
\item the term $\comp{\rho}$ is defined for every $\rho$
\end{enumerate}
\end{lem}
\begin{proof}
The proof of part a) is a standard induction on the structure of terms.

Part b) is a direct consequence of part a).
Since $ \len{\rho'} = \len{\rho' \inner{\rho}{x}} $ and
by definition $\len{\rho'} < \len{\rho}$, we know that $\len{\rho' \inner{\rho}{x}} < \len{\rho}$.

The proof of part c) is by induction on the length of terms.
The only interesting case is when $ \rho = \Rec[x]{ \rho'}$.
As $ \comp{ \rho } = \Rec[x]{ \comp{ \rho' \inner{ \rho }{ x } }}$,
we have to show that $ \comp{ \rho' \inner{ \rho }{ x } } $ is
defined. This follows from Part b), which ensures that the length of $  \rho' \inner{ \rho }{
  x }$ is shorter than the length of $\rho$, and the inductive
hypothesis.
\end{proof}

\leaveout{

Let us denote with $\mathbf{s}\sset{ {\sigma_1}/{x_1} , \ldots ,
  {\sigma_n}/{x_n} }$ the substitution defined as $\mathbf{s}$, expect
for the variables $x_1, \ldots, x_n$, which are mapped respectively to
$\sigma_1, \ldots, \sigma_n$.

\begin{lem}
\label{lem:inner-and-subst}
  For every $\rho, \sigma, x \in \Lsc$, and substitution
  $\ssubst{\sigma}{x}$ such that $x$ is not free in
  the range\footnote{Find better term.} of $\ssubst{\sigma}{x}$, then 
\begin{enumerate}[label=\roman*)]
\item\label{pt:idempot-subst-1}
  $ \rho \ssubst{ \sigma }{x} = (\rho \ssubst{\sigma}{x})
  \ssubst{\sigma}{x} $
\item\label{pt:idempot-subst-2}
  $ \rho \ssubst{ \sigma }{x} = (\rho \inner{\sigma}{x})
  \ssubst{\sigma}{x} $
\end{enumerate}
\end{lem}
\begin{proof}
  Since $x$ is not free in the range of $\ssubst{\sigma}{x}$,
  $x$ is not free is $\rho \ssubst{\sigma}{x}$ either. 
  It follows that the application of
  $ \ssubst{\sigma}{x}$ to $\rho \ssubst{\sigma}{x}$ is immaterial,
  thus \rpt{idempot-subst-1} is true.

  We prove \rpt{idempot-subst-2} by structural induction on $\rho$.
  \paragraph{Base cases}
  If $\rho = \Unit$ or $\rho$ is a variable,  then
  $\rho \inner{\sigma}{x} = \rho$, and so
  $\rho \ssubst{ \sigma }{x} = (\rho \inner{\sigma}{x}) \ssubst{
    \sigma }{x}$.

  \paragraph{Inductive cases}
  We discuss two cases, for the others are similar.

  Suppose that $\rho = \, ?(\rho^m).\rho'$. Then
  $$ 
  (\rho \inner{\sigma}{x})\ssubst{\sigma}{x} =  \,
  ?((\rho^m \inner{\sigma}{x})
  \ssubst{\sigma}{x}). ((\rho' \inner{\sigma}{x}) \ssubst{\sigma}{x})$$
  The inductive hypothesis implies two syntactic identities, namely
  \begin{itemize}
  \item
    $ 
    \rho^m \ssubst{\sigma}{x} =
    (\rho^m \inner{\sigma}{x}) \ssubst{\sigma}{x}
    $
  \item
    $
    \rho' \ssubst{\sigma}{x} =
    (\rho' \inner{\sigma}{x}) \ssubst{\sigma}{x}
    $
  \end{itemize}
These identities ensures that
$ (\rho \inner{\sigma}{x})\ssubst{\sigma}{x} = \rho
\ssubst{\sigma}{x}$.

Suppose that $\rho = \Rec[y]{\rho'}$.

By definition we know that $ \rho \subst{\sigma}{x} = \Rec[z]{ ( \rho'
  \sSubst{\sigma}{x}{z}{y} ) }$, where $z$ is a fresh variable.
We also know that $$
(\rho \inner{\sigma}{x}) \ssubst{\sigma}{x} =
(\Rec[y]{\rho' \inner{\sigma}{x}} ) \ssubst{\sigma}{x} =
\Rec[z]{ ( (\rho' \inner{\sigma}{x})  \sSubst{\sigma}{x}{z}{y}) }
$$
The inductive hypothesis ensures the equality
$$
\rho' \sSubst{\sigma}{x}{z}{y} = (\rho' \inner{\sigma}{x})  \sSubst{\sigma}{x}{z}{y}
$$
from which the result follows.
\end{proof}

\leaveout{
\begin{lem}
\label{lem:inner-and-subst}
  For every $\rho, \sigma, x \in \Lsc$, if $x$ is not free in
  $\sigma$, then 
\begin{enumerate}[label=\roman*)]
\item\label{pt:idempot-subst-1}
  $ \rho \subst{ \sigma }{x} = (\rho \subst{\sigma}{x})
  \subst{\sigma}{x} $
\item\label{pt:idempot-subst-2}
  $ \rho \subst{ \sigma }{x} = (\rho \inner{\sigma}{x})
  \subst{\sigma}{x} $
\end{enumerate}
\end{lem}
\begin{proof}
  Since $x$ is not free in $\sigma$, $x$ is not free is $\rho
  \subst{\sigma}{x}$ either. It follows that the application of
  $ \subst{\sigma}{x}$ to $\rho \subst{\sigma}{x}$ is immaterial,
  thus \rpt{idempot-subst-1} is true.

  We prove \rpt{idempot-subst-2} by structural induction on $\rho$.
  \paragraph{Base cases}
  If $\rho = \Unit$ or $\rho$ is a variable,  then
  $\rho \inner{\sigma}{x} = \rho$, and so
  $\rho \subst{ \sigma }{x} = (\rho \inner{\sigma}{x}) \subst{
    \sigma }{x}$.

  \paragraph{Inductive cases}
  We discuss two cases, for the others are similar.

  Suppose that $\rho = \, ?(\rho^m).\rho'$. Then
  $$ 
  (\rho \inner{\sigma}{x})\subst{\sigma}{x} =  \,
  ?((\rho^m \inner{\sigma}{x})
  \subst{\sigma}{x}). ((\rho' \inner{\sigma}{x}) \subst{\sigma}{x})$$
  The inductive hypothesis implies two syntactic identities, namely
  \begin{itemize}
  \item
    $ 
    \rho^m \subst{\sigma}{x} =
    (\rho^m \inner{\sigma}{x}) \subst{\sigma}{x}
    $
  \item
    $
    \rho' \subst{\sigma}{x} =
    (\rho' \inner{\sigma}{x}) \subst{\sigma}{x}
    $
  \end{itemize}
These identities ensures that
$ (\rho \inner{\sigma}{x})\subst{\sigma}{x} = \rho
\subst{\sigma}{x}$.

Suppose that $\rho = \Rec[y]{\rho'}$. If $y = x$ then
$ \rho \subst{\sigma}{x} = \rho$, and so the result follows.
If $y \neq x$ then 
$$ 
(\rho \inner{\sigma}{x}) \subst{\sigma}{x} = \Rec[y]{ (\, (\rho'
  \inner{\sigma}{x}) \subst{\sigma}{x} \, )}
$$
The result follows from the inductive hypothesis and the definition of
$\subst{\cdot}{\cdot}$.
If $y \neq x$ then let $z$ be variable that is not free in
$\sigma$.
By definition we know that
$$
(\rho \inner{\sigma}{x}) \subst{\sigma}{x} =
(\Rec[y]{\rho' \inner{\sigma}{x}} ) \subst{\sigma}{x} =
\Rec[z]{ (\rho' \inner{\sigma}{x}  \Subst{\sigma}{x}{z}{y}) }
$$
and that $ \rho \subst{\sigma}{x} = \Rec[z]{ ( \rho'
  \Subst{\sigma}{x}{z}{y} ) }$.
The inductive hypothesis ensures that
$ \rho' \Subst{\sigma}{x}{z}{y} = \rho' \inner{\sigma}{x}
\Subst{\sigma}{x}{z}{y} $.

\leaveout{
If $y \neq x$ then 
$$ 
(\rho \inner{\sigma}{x}) \subst{\sigma}{x} = \Rec[y]{ (\, (\rho'
  \inner{\sigma}{x}) \subst{\sigma}{x} \, )}
$$
The result follows from the inductive hypothesis and the definition of
$\subst{\cdot}{\cdot}$.
}
\end{proof}
}

\begin{lem}
\label{lem:subst-idempot-2}
  Let $\rho, \rho', \sigma, x, y \in \Lsc$ and $\mathbf{s}$ a
  substitution such that $x$ is not free in the range of $\ssubst{\sigma}{x}$.
\begin{enumerate}[label=\roman*)]
\item $( \rho \ssubst{ \rho' \subst{\sigma}{x} }{y}) \ssubst{\sigma}{x} =
  (\rho \ssubst{ \rho' }{y}) \ssubst{\sigma}{x}$

\item
  $( \rho \inner{ \rho' \subst{\sigma}{x} }{y}) \ssubst{\sigma}{x} =
  (\rho \inner{ \rho' }{y}) \ssubst{\sigma}{x}$
\end{enumerate}
\end{lem}
\begin{proof}
  We prove part (i), as the proof of part (ii) is similar.
  The argument is by structural induction on $\rho$.

  \paragraph{Base case}
  If $\rho = \Unit$ or $\mathbf{s}(\rho)$ is not defined,
  then the substitutions have no effect, and so $ (\rho \ssubst{ 
    \rho' \subst{\sigma}{x} }{y})
  \ssubst{\sigma}{x} = \rho = (\rho \ssubst{ \rho' }{y})
  \ssubst{\sigma}{x}$.

  If $\rho$ is a variable and $\mathbf{s}(\rho)$ is defined then
  there are two subcases.

  If $ \rho = y$, then by definition
  $ ( \rho \ssubst{ \rho' \subst{\sigma}{x} }{y}) \ssubst{\sigma}{x} = (\rho' \subst{\sigma}{x}) \ssubst{\sigma}{x}$.
  By hypothesis the variable $x$ is not free in $\sigma$, so 
  $ (\rho' \subst{\sigma}{x}) \ssubst{\sigma}{x} = \rho' \ssubst{\sigma}{x}$. 
  Plainly $ \rho' = \rho \ssubst{\rho'}{y}$, thus $\rho' \ssubst{\sigma}{x} = (\rho \ssubst{\rho'}{y})  \ssubst{\sigma}{x}$,
which in turn implies the result.\leaveout{ we obtain the equality
  $( \rho \subst{ \rho' \ssubst{\sigma}{x} }{y}) \ssubst{\sigma}{x} = (\rho \ssubst{\rho'}{y}) \ssubst{\sigma}{x}$.}

  If $\rho \neq y$, then
  $\rho \ssubst{ \rho' \subst{\sigma}{x} }{y} = \mathbf{s}(z) = \rho \ssubst{ \rho' }{y}$,
  and the result follows from an application of $\ssubst{\sigma}{x}$ to the left-hand and right-hand
  terms.


  \paragraph{Inductive cases}
  We discuss two case.
  Suppose that $ \rho = \, !(\rho^m).\hrho$. 
  The inductive hypothesis ensures the following equalities,
\begin{itemize}
\item $ (\rho^m \ssubst{  \rho' \subst{\sigma}{x} }{y}) \ssubst{\sigma}{x} =
  ( \rho^m \ssubst{ \rho' }{y} ) \ssubst{\sigma}{x}$
\item $ (\hrho \ssubst{  \rho' \subst{\sigma}{x} }{y}) \ssubst{\sigma}{x} =
  ( \hrho \ssubst{ \rho' }{y} ) \ssubst{\sigma}{x}$
\end{itemize}
The definition of application of a substitution to a term implies the
equality
\begin{equation}
\label{eq:idempot-2-eq1}
  (\rho \ssubst{ \rho' \subst{\sigma}{x} }{y}) \ssubst{\sigma}{x}
= !( \rho^m \ssubst{ \rho' \subst{\sigma}{x} }{y} \ssubst{\sigma}{x}).
(\hrho \ssubst{ \rho' \subst{\sigma}{x} }{y} \ssubst{\sigma}{x})
\end{equation}
The inductive hypothesis implies that the right-hand term of
(\ref{eq:idempot-2-eq1}) equals the next term
$$
!( (\rho^m \ssubst{ \rho' }{y}) \ssubst{\sigma}{x}).
((\hrho \ssubst{ \rho' }{y}) \ssubst{\sigma}{x})
$$
The result now follows from the definition of application of
substitutions.
\leaveout{
The definition of substitution application now ensures the following
equalities, and
$$
\begin{array}{ll}
  & (\rho \ssubst{ \rho' \subst{\sigma}{x} }{y}) \ssubst{\sigma}{x}\\
= & !( (\rho^m \ssubst{ \rho' \subst{\sigma}{x} }{y}) \ssubst{\sigma}{x}
).(\hrho \ssubst{ \rho' \subst{\sigma}{x} }{y}) \ssubst{\sigma}{x})\\
= & !( \rho^m \ssubst{ \rho' }{y} \ssubst{\sigma}{x}).
(\hrho \ssubst{ \rho' \subst{\sigma}{x} }{y} \ssubst{\sigma}{x})\\
= & ( !( \rho^m \ssubst{ \rho' }{y}. 
(\hrho \ssubst{ \rho' }{y} ) ) \ssubst{\sigma}{x}\\
= & (( !(\rho^m).\hrho  )\ssubst{ \rho' }{y}) 
\ssubst{\sigma}{x}\\
= & (\rho \ssubst{ \rho' }{y})  \ssubst{\sigma}{x}
\end{array}
$$
}

Suppose that $\rho = \Rec[x]{\rho''}$.
Then by definition
$$
\begin{array}{lll}
(\rho \ssubst{ \rho' \subst{\sigma}{x} }{y}) \ssubst{\sigma}{x} & = &
(\Rec[z]{ \rho'' \sSubst{  \rho' \subst{\sigma}{x} }{y}{ z }{ x } })
\ssubst{\sigma}{x} \\
& = & \Rec[z']{ ((\rho'' \sSubst{  \rho' \subst{\sigma}{x} }{y}{ z }{ x } )
\sSubst{\sigma}{x}{z'}{z} )}\\
& = & \Rec[z']{ (  (\rho'' \Nsubst{1}{  \rho' \subst{\sigma}{x} }{y})
  \Nsubst{2}{\sigma}{x} )}
\end{array}
$$
where $z$ is not free in $\rho' \subst{\sigma}{x}$,
$\mathbf{s}_1 = \ssubst{z}{x}$, and $\mathbf{s}_2 = \ssubst{z'}{z}$ with $z'$ not
free in $\sigma$.

Since $z$ is not free in $\rho' \subst{\sigma}{x}$, it follows that
$z$ is not free in $\rho'$ either.
In view of this, we expand the definition of application of a
substitution as follows,
$$
\begin{array}{lll}
 (\rho \ssubst{ \rho' }{y}) \ssubst{\sigma}{x} & = &
 \Rec[z]{ ( \rho'' \sSubst{ \rho' }{y}{z}{x} ) } \ssubst{\sigma}{x}\\
& = &
\Rec[z']{ ( \rho'' \sSubst{ \rho' }{y}{z}{x} \sSubst{\sigma}{x}{z'}{z}
  ) } \\
& = &
\Rec[z']{ ( \rho'' \Nsubst{1}{ \rho' }{ y } \Nsubst{2}{\sigma}{x}
  ) } \\
\end{array}
$$
The inductive hypothesis implies that
$$
( \rho'' \Nsubst{1}{ \rho' \subst{\sigma}{x} }{y}) \Nsubst{2}{\sigma}{x} =
  (\rho'' \Nsubst{1}{ \rho' }{y}) \Nsubst{2}{\sigma}{x}
$$
from which the result follows.
\end{proof}

\leaveout{
\begin{lem}
\label{lem:subst-idempot-2}
  Let $\rho, \rho', \sigma, x \in \Lsc$ with $x$ not free in $\sigma$.
\begin{enumerate}[label=\roman*)]
\item $( \rho \subst{ \rho' \subst{\sigma}{x} }{y}) \subst{\sigma}{x} =
  (\rho \subst{ \rho' }{y}) \subst{\sigma}{x}$

\item
  $( \rho \inner{ \rho' \subst{\sigma}{x} }{y}) \subst{\sigma}{x} =
  (\rho \inner{ \rho' }{y}) \subst{\sigma}{x}$
\end{enumerate}
\end{lem}
\begin{proof}
  We prove part (i), as the proof of part (ii) is similar.
  The argument is by structural induction on $\rho$.

  \paragraph{Base case}
  If $\rho = \Unit$ or 
  $\rho = z$, so $\rho \not = x$ and $\rho \not = x$, then the substitutions have no effect, and so $ (\rho \subst{ 
    \rho' \subst{\sigma}{x} }{y})
  \subst{\sigma}{x} = \rho = (\rho \subst{ \rho' }{y})
  \subst{\sigma}{x}$.
  
  If $ \rho = y$, then $ ( \rho \subst{ \rho' \subst{\sigma}{x} }{y}) \subst{\sigma}{x} = (\rho' \subst{\sigma}{x}) \subst{\sigma}{x}$.
  By hypothesis the variable $x$ is not free in $\sigma$, so Lemma~(\ref{lem:inner-and-subst}) implies that 
  $ (\rho' \subst{\sigma}{x}) \subst{\sigma}{x} = \rho'
  \subst{\sigma}{x}$. Plainly $ \rho' = \rho \subst{\rho'}{y}$, so $( \rho \subst{ \rho' \subst{\sigma}{x} }{y}) \subst{\sigma}{x} = (\rho \subst{\rho'}{y}) \subst{\sigma}{x}$.

  If $\rho = x$, then $ ( \rho \subst{ \rho' \subst{\sigma}{x} }{y})
  \subst{\sigma}{x} = x \subst{\sigma}{x} = (x \subst{\rho'}{y})
  \subst{\sigma}{x}$.

  \paragraph{Inductive cases}
  We discuss one case.
  Suppose that $ \rho = \, !(\rho^m).\hrho$. 
  The inductive hypothesis ensures the following equalities,
\begin{itemize}
\item $ (\rho^m \subst{  \rho' \subst{\sigma}{x} }{y}) \subst{\sigma}{x} =
  ( \rho^m \subst{ \rho' }{y} ) \subst{\sigma}{x}$
\item $ (\hrho \subst{  \rho' \subst{\sigma}{x} }{y}) \subst{\sigma}{x} =
  ( \hrho \subst{ \rho' }{y} ) \subst{\sigma}{x}$
\end{itemize}
Now we know enough to prove the result,
$$
\begin{array}{ll}
  & (\rho \subst{ \rho' \subst{\sigma}{x} }{y}) \subst{\sigma}{x}\\
= & !( (\rho^m \subst{ \rho' \subst{\sigma}{x} }{y}) \subst{\sigma}{x}
).(\hrho \subst{ \rho' \subst{\sigma}{x} }{y})\subst{\sigma}{x}\\
= & !( \rho^m \subst{ \rho' }{y} \subst{\sigma}{x}).
(\hrho \subst{ \rho' \subst{\sigma}{x} }{y})\subst{\sigma}{x}\\
= & !( \rho^m \subst{ \rho' }{y} \subst{\sigma}{x}). 
(\hrho \subst{ \rho' }{y} ) \subst{\sigma}{x}\\
= & ( !(\rho^m \subst{ \rho' }{y}).\hrho\subst{ \rho' }{y} )
\subst{\sigma}{x}\\
= &  (!(\rho^m).\hrho  \subst{ \rho' }{y}) \subst{\sigma}{x}\\
= & (\rho \subst{ \rho' }{y})  \subst{\sigma}{x}
\end{array}
$$
\end{proof}
}

\begin{lem}
\label{lem:closep-dual-distributes}
  For every $\rho, \sigma, x \in  \Lsc$, if $\sigma$ is closed
  and $x$ is not free in the messages of $\rho$, then
  $
  \comp{\rho \subst{\sigma}{x}} = 
  \comp{\rho} \subst{ \comp{\sigma}}{x}
  $
\end{lem}
\begin{proof}
  The proof is by induction on the length of $\rho$.

  \paragraph{Base case}
  If $\rho = \Unit$ or $\rho = y$, so $\rho \not = x$, then the substitutions have no effect, and the result is true. For instance,
  $$
  \comp{\Unit \subst{\sigma}{x}} = \comp{\Unit}
  = \Unit
  = \Unit \subst{\comp{\sigma}}{x}
  = \stdual{ \Unit } \subst{ \comp{\sigma} }{x}
  $$
  Suppose that $\rho = x$; then 
$$ 
  \comp{ \rho \subst{\sigma}{x}}= \comp{ \sigma  } = x \subst{\comp{\sigma}}{x} = \rho \subst{\comp{\sigma} }{x}
$$

  \paragraph{Inductive case}
  We discuss two cases.
  Suppose that $ \rho = \, !(\rho').\rho''$. We have to prove that
  \begin{equation}
    \label{eq:gclosep-distr-1}
    \comp{ (!(\rho').\rho'') \subst{\sigma}{x} }  =
    \comp{  !(\rho').\rho'' } \subst{\comp{\sigma}}{x}
  \end{equation}
  The definitions of
  $\compSym$ and $\subst{\cdot}{\cdot}$ ensure the identity
  $$
  \comp{\rho\subst{\sigma}{x}} = ?(\rho'
  \subst{\sigma}{x}).\comp{\rho''\subst{\sigma}{x} }
  $$
  By hypothesis $x$ does not appear free in $\rho'$,
  so we know that $\rho' \subst{\sigma}{x} = \rho'$, thereby simplifying the last identity above,
$$
\comp{\rho\subst{\sigma}{x}} =
?(\rho').\comp{\rho''\subst{\sigma}{x}}
$$

  Since $x$ is not free in the messages of $\rho$,
  the same is true for $\rho''$; as the length of $\rho''$ is smaller than then length of $\rho$, we apply the inductive hypothesis, which ensures  
that 
  $$ \comp{\rho''\subst{\sigma}{x}}  = \comp{\rho''}
  \subst{\comp{\sigma}}{x} $$
It follows that
$$
 ?(\rho').
 \comp{\rho''\subst{\sigma}{x}} = 
?(\rho'). \comp{\rho''}
  \subst{\comp{\sigma}}{x}
$$
and so $ \comp{\rho\subst{\sigma}{x}} = \, ?(\rho').( \comp{\rho''}
  \subst{\comp{\sigma}}{x})$.
Since $x$ does not appear free in $\rho'$ we can move the substitution on the right-hand side,
$$
\comp{\rho\subst{\sigma}{x}} =  (?(\rho'). \comp{\rho''} )
  \subst{\comp{\sigma}}{x}
$$
In view of the definitions of $\compSym$ 
$$
\comp{\rho\subst{\sigma}{x}} =  (\comp{!(\rho').\rho''} ) \subst{\comp{\sigma}}{x}
$$

Now we discuss the case involving recursion.
Suppose that $\rho = \Rec[y]{\rho'}$ where clearly $y \not = x$.
We are required to show the equality
$$
  \comp{ \Rec[y]{\rho'} \subst{\sigma}{x}} = \comp{ \Rec[y]{\rho'} } \subst{\comp{\sigma}}{x}
$$
In view of the definitions of  $\subst{\cdot}{\cdot}$ and $\compSym$, we have to show that 
\begin{equation}
  \label{eq:gclosep-distr-rec}
 \Rec[y]{ \comp{ \hsigma }}
=
\Rec[y]{ \comp{ \rho' \inner{ \rho }{y }}} \subst{\comp{\sigma}}{x}
\end{equation}
where $\hsigma =  (\rho' \subst{\sigma}{x}) \inner{\Rec[y]{ (\rho'
    \subst{\sigma}{x})}}{y}$.
Before proceeding with the main argument, we simplify the term $\hsigma$.
By hypothesis $\sigma$ is closed. This implies that $y$ does not appear free in $\sigma$, and that $ x $ does not appear free in $\Rec[y]{ (\rho' \subst{\sigma}{x} )}$. These facts let us flip the substitutions in $\hsigma$,
$$
\hsigma = (\rho' \inner{\Rec[y]{ (\rho'
    \subst{\sigma}{x}) }}{y}) \subst{\sigma}{x}
$$
Observe now that $ \Rec[y]{ (\rho'
    \subst{\sigma}{x}) } = (\Rec[y]{ \rho' })\subst{\sigma}{x} = \rho \subst{\sigma}{x}$,
so $$ \hsigma = (\rho' \inner{ \rho \subst{\sigma}{x}}{y}) \subst{\sigma}{x}$$
As $\sigma$ is closed, $x$ does not appear free in $\sigma$, and this
allows us to apply Lemma~(\ref{lem:subst-idempot-2}),
$$
\hsigma = (\rho' \inner{ \rho }{y}) \subst{\sigma}{x}
$$

In view of (\ref{eq:gclosep-distr-rec}), and of the simplified form of $\hsigma$, the equality we want to prove is
$$
 \Rec[y]{ \comp{ (\rho' \inner{\rho}{y}) \subst{\sigma}{x} }}
=
(\Rec[y]{ \comp{ \rho' \inner{ \rho }{y }} }) \subst{\comp{\sigma}}{x}
$$
\leaveout{
It suffices to prove that
$$
\Rec[y]{ \comp{ \,  (\rho' \inner{ \rho }{y}) \subst{\sigma}{x} \, }}
=
\Rec[y]{ ( \comp{ \rho' \inner{ \rho }{y }} \subst{\comp{\sigma}}{x})}
$$
}
In turn this follows from 
$$
  \comp{ \,  (\rho' \inner{ \rho }{y}) \subst{\sigma}{x} \, }
= 
 \comp{ \rho' \inner{ \rho }{y }} \subst{\comp{\sigma}}{x}
$$
But the last equality is true, because it is the inductive hypothesis
applied to $\rho' \inner{ \rho }{y }$. This is sound because the
variable $x$ does not appear free in the message of $\rho' \inner{
  \rho }{y }$, and the length of $\rho' \inner{ \rho }{y }$ is
shorter than the length of $\rho$ (see Lemma~(\ref{lem:comp-defined})).
\end{proof}
}
}

\renewcommand{\appto}[1]{\langle \, #1 \, \rangle}
\newcommand{\New}[1]{\colorbox{blue!10}{\ensuremath{#1}}}

\section{Type derivations}
\label{app:type-inference}

In this appendix we discuss the derivation trees needed to
prove the typing judgement $\typerel P$ and $\typerel Q$, where $P$ and
$Q$ are the processes respectively of \rexa{intro-funny-types-needed} and
\rexa{replication}.

Recall the type equivalence $\typeEQ$ defined in \rsec{results}
(\ref{eq:typeEQ}). To use the typing rules of
\cite{DBLP:journals/entcs/YoshidaV07} and
\cite{DBLP:journals/iandc/Vasconcelos12} we reason up-to type
equivalence.

\begin{figure}
\hrulefill
$$
\infer[\rname{Cat}^B]
{X \at T_x, \, T_y \typerel \catch{y}{z} X \appto{ z , \kF^-} \rhd 
y \at \sinp{ T_z }\e,
  \,  \kF^- \at T^-_\textsf{\bf f}}
{
\infer[\rname{Var}^C]
{X \at T_x, \, T_y \typerel X \appto{ z , \kF^-} \rhd  y \at \e, \, z\at T_z, \, 
  \,  \kF^- \at T^-_\textsf{\bf f}}
{  y \at \e \quad \mathsf{completed}
}}
$$

$$
\infer[T^+_{\bf f} \mathrel{\mathcal{D}} T^-_{\bf f}; \, \rname{CRes}]
{
\begin{array}{l}
X \at T_x, \, T_y
\typerel
\nn{\kF}( \throw{x}{\kF^-}\void \Par\\
\catch{y}{z} X \appto{ z, \kF^- }) \rhd x \at T_x, \, y \at T_y
\end{array}
}
{
\infer[\rname{Conc}]
{
\begin{array}{ll}
X \at T_x, \, T_y \typerel \throw{x}{\kF^-}\void \Par \\
\catch{y}{z} X \appto{ z, \kF^- } \rhd x \at T_x, \, y \at T_y, \, \kF^-\at
  T^-_\textsf{\bf f}, \,
  \kF^+ \at T^+_\textsf{\bf f}
\end{array} }
{
\infer[\rname{Thr}^A]
{X \at T_x, \, T_y  \typerel \throw{x}{\kF^+}\void \rhd x \at
  \sout{ T^+_\textsf{\bf f} }\e, \, \kF^+\at
  T^+_\textsf{\bf f}}
{\infer[\rname{Inact}]
{X \at T_x, \, T_y \typerel \void \rhd x\at \e}
{  x\at \e \quad \mathsf{completed}}
}
\quad
\infer[]{(**)}{\vdots}
}
}
$$

$$
\infer[T_x \mathrel{\mathcal{D}} T_y; \, \rname{CRes}]{
\typerel \nn{\kO}(\defin{D} X\appto{ \kO^+, \, \kO^- }) }
{
\infer[\rname{Def}]{
 \typerel 
 \defin{D} X\appto{ \kO^+, \, \kO^- } \rhd \kO^+ \at T_x,
\, \kO^- \at T_y}
 {
\infer[]{(*)}{\vdots}
\quad
\infer[\rname{Var}]{
X  \at T_x, \, T_y \typerel X\appto{ \kO^+, \, \kO^- } \rhd \kO^+ \at T_x,
\, \kO^- \at T_y
}{ \emptyset \quad \mathsf{completed}}
}
}
$$

$$
\begin{array}{ll}
A) & T_x \typeEQ \sout{ T^+_\textsf{\bf f} }\e\\
B) & T_y \typeEQ \sinp{T_z}\e \\
C) & T_x \typeEQ T_z, \quad T_y \typeEQ T^-_\textsf{\bf f} \\
\end{array}
$$
\caption{Derivation tree to infer $\typerel P$ in the type system of
  \cite{DBLP:journals/entcs/YoshidaV07}, but with rule
  \rname{CRes} of our \rfig{cres}}
\label{fig:type-inference}
\hrulefill
\end{figure}

We discuss first the derivation of $\typerel P$, which is shown in 
In \rfig{type-inference}. The rules we used are given in
\cite[pag. 89]{DBLP:journals/entcs/YoshidaV07}, and the ones not
defined there are defined in Figure~6 of the same paper. The only difference
between our presentation and \cite{DBLP:journals/entcs/YoshidaV07}
is that in rule \rname{CRes} we use a relation $\mathcal{D}$ instead of the
duality $\stdual{\cdot}$. 
This allows us to compare the typability of
$P$ using $\prdualSym$ and $\stdual{\cdot}$} as duality relations.

The existence of the derivation tree depends on the satisfaction of
the constraints that we gather building the tree. In
\rfig{type-inference} the constraints and the rule applications that
generate them are labelled respectively with A, B, C.

We show how to satisfy the conditions, which are the following ones.
$$
\begin{array}{lcl@{\hskip 3em}lcl@{\hskip 3em}lll}
T_x & \typeEQ & \sout{T^+_\textsf{\bf f}}\e &T_x & \typeEQ & T_z\\[.5em]
T_y & \typeEQ & \sinp{ T_z }\e & T_y & \typeEQ & T^-_\textsf{\bf f} \\[.5em]
T_x & \mathrel{\mathcal{D}} &  T_y  &
T^+_{\bf f} & \mathrel{\mathcal{D}} &  T^-_{\bf f}
\end{array}
$$
We let $T_z$ and $T^-_\textsf{\bf f}$ to be
definitionally equal to, respectively, $T_x$ and $T_y$.
In view of these definitions, the constraints we have to satisfy are
the following ones,
$$
\begin{array}{lcl@{\hskip 3em}lcl}
T_x & \typeEQ & \sout{T^+_\textsf{\bf f}}\e & T_x &
\mathrel{\mathcal{D}} &  T_y \\[.5em]
T_y & \typeEQ & \sinp{ T_x }\e & T^+_{\bf f} & \mathrel{\mathcal{D}} &  T_y\\[.5em]
\end{array}
$$
Given the condition that $ T_x \mathrel{\mathcal{D}}  T_y  $ an
obvious way to satisfy $ T^+_{\bf f} \mathrel{\mathcal{D}}  T_y$ is to
let $  T^+_{\bf f}$ be definitionally $ T_x $. This leaves us with
three conditions
$$
T_x  \typeEQ  \sout{ T_x }\e, \quad
T_y  \typeEQ  \sinp{ T_x }\e, \quad
T_x  \mathrel{\mathcal{D}}   T_y
$$
The two equations are solved by letting $T_x = \Rec{\sout{X}\e}$,
and $ T_y = \sinp{ \Rec{\sout{X}\e} }\e $.
If we take $\mathcal{D}$ to be $\stdual{\cdot}$, then the types
$T_x$ and $T_y$ cannot satisfy $ T_x  \mathrel{\mathcal{D}}   T_y $,
because $ \stdual{T_x} \not \typeEQ T_y$.
If we take $\mathcal{D}$ to be $\prdualSym$, then the
condition on $\mathcal{D}$ is satisfied,
because $ \comp{T_x} \typeEQ T_y$.

\newcommand{\qcompS}{\mathsf{prdualQ}}
\newcommand{\qcomp}[1]{\qcompS(#1)}


Now we discuss the derivation of $\typerel Q$, where
  $$
  Q = \nn{xy}( \, \vOut{x}{x}\void \Par \un \vIn{y}{z}\vOut{z}{z}\void \,)
  $$
is the process we already used in \rexa{replication}. We show the
derivation tree in \rfig{replication}. Let us adapt our notions
to the setting of \cite{DBLP:journals/iandc/Vasconcelos12}. 
We use the the syntax given in Figure~3 of that paper,
$$
\begin{array}{lllll}
      q & ::= & & \textbf{Qualifiers}\\
       && {\sf lin} & \text{linear}\\
      && {\sf un} & \text{unrestricted}
      \\[1em]
      t & :: = & & \textbf{Types}\\
       && \bool & \text{boolean}\\
       && \e & \text{termination}\\
       && (q, T)& \text{qualified pretype}
\end{array}
$$
where  our terms $T, S, \ldots$ are considered {\em pretypes},
and a type $t$ is pair composed by a qualifier and a pretype.
Now we lift the relation $\typeEQ$ (i.e. the equivalence due to $\subt$)
to types by writing 
$(q, T) \typeEQ (q', T')$ whenever $ q = q'$ and $T \typeEQ T'$.
We also lift the duality function $\prdual{ - }$ to types,
and let
\begin{equation}
\label{eq:prdual-for-types-a-la-Vasco}
\qcomp{ q, T } = (q, \prdual{T})
\end{equation}
This definition  is analogous
to the one in \cite[Figure~4]{DBLP:journals/iandc/Vasconcelos12},
but there the standard duality is employed.

The premises of the rules used in the derivation tree of \rfig{replication} 
ensure that
$T_y \typeEQ \un \sinp{ T_z }\e$, and the two equations A) and B) in that figure
have the same solution, so $T_z \subt T_x$.
It turn this implies that $ T_y \typeEQ \un \sinp{ T_x }\e$.
The existence of the derivation tree depends on two conditions, namely
$$
 T_x \typeEQ \un \sout{T_x}\e, \quad T_x \D T_y
$$
We already know that the equation is satisfied by the type 
$\hat{T} = \un \Rec{\sout{X}\e}$, so we have just to find a
$\D$ such that $ \hat{T} \D T_y$. 
If we instantiate $\D$ to the function $\qcompS$ defined (\ref{eq:prdual-for-types-a-la-Vasco}) above, then the condition
$\hat{T} \D T_y$ is satisfied.
If we instantiate $\D$ to the standard duality, then that condition is
not satisfied, so a derivation tree necessary to conclude $\typerel Q$ does not exist.

\begin{figure}[t]
\hrulefill
$$
\infer[\rname{T-Out}^B]{
y\at\e, \, z\at T_z \typerel \vOut{z}{z}\void
}
{
z\at T_z \typerel z \at \un\sout{T_z}\e
\quad
z\at T_z \typerel z \at T_z
\quad
\infer[\rname{T-Inact}]{ y\at\e, \, z\at \e \typerel \void}{}
}
$$
$$
\infer[\rname{T-In}]
{ y\at T_y \typerel \un \vIn{y}{z}\vOut{z}{z}\void}
{ \un (y\at \un\sinp{T_z}\e)
\quad
y\at \un\sinp{T_z}\e \typerel y\at \un\sinp{T_z}\e
\quad
\infer[]{(**)}{\vdots}
}
$$
$$
\infer[\rname{T-Res} \quad T_x \mathrel{\mathcal{D}}  T_y]
{\typerel  \nn{xy}( \vOut{x}{x}\void \Par \un \vIn{y}{z}\vOut{z}{z}\void)}
{
\infer[\rname{T-Par}]
{ x\at T_x, \, y \at T_y \typerel \vOut{x}{x}\void \Par \un \vIn{y}{z}\vOut{z}{z}\void}
{
\infer[ \rname{T-Out}^A]
{ x\at T_x \typerel  \vOut{x}{x}\void }
{ x\at T_x \typerel x\at \un \sout{T_x}\e
\quad
x\at T_x \typerel x\at T_x
\quad
\infer[\rname{T-Inact}]{x\at\e \typerel \void}{} }
\quad
\infer[]{(*)}{\vdots}
}
}
$$
$$
\begin{array}{ll}\\
A) & T_x \typeEQ \sout{ T_x }\e\\
B) & T_z \typeEQ \sout{ T_z }\e\
\end{array}
$$
\caption{Derivation tree to infer $\typerel Q$ in the type system of
  \cite{DBLP:journals/iandc/Vasconcelos12}, but with rule
  \rname{T-Res} of our \rfig{tres}}
\label{fig:replication}
\hrulefill
\end{figure}

\bibliographystyle{alpha}
\bibliography{model}
\end{document}